\documentclass[11pt,letterpaper]{article}
\usepackage[margin=1in]{geometry}

\usepackage{amsmath,amsthm,amssymb,amsfonts}
\usepackage{xspace}
\usepackage{url}
\usepackage{hyperref}
\usepackage[dvipsnames]{xcolor}
\usepackage[ruled,linesnumbered,noend]{algorithm2e}
\usepackage{appendix}
\usepackage{babel}

\usepackage{multirow}
\usepackage{tablefootnote}
\usepackage{afterpage}
\usepackage{lipsum}
\usepackage{longtable}
\usepackage{booktabs}
\usepackage{thm-restate}
\usepackage{graphicx}
\usepackage{subcaption}
\usepackage[normalem]{ulem}

\usepackage{bm}
\usepackage{enumerate}
\usepackage{bbm}
\usepackage{mathtools}
\usepackage{mathrsfs}
\usepackage[capitalise,noabbrev]{cleveref}

\usepackage{comment}

\hypersetup{colorlinks=true,allcolors=blue}

\theoremstyle{plain}
\newtheorem{theorem}{Theorem}[section]
\newtheorem{lemma}[theorem]{Lemma}

\newtheorem{corollary}[theorem]{Corollary}

\newtheorem{proposition}[theorem]{Proposition}
\newtheorem{fact}[theorem]{Fact}

\theoremstyle{definition}
\newtheorem{definition}[theorem]{Definition}

\theoremstyle{remark}
\newtheorem{remark}[theorem]{Remark}
\newtheorem*{remark*}{Remark}

\allowdisplaybreaks[4]

\makeatletter
\newcounter{HALG@line}
\renewcommand{\theHALG@line}{\thealgorithm.\arabic{ALG@line}}
\makeatother

\newcommand{\norm}[1]{\left|{#1}\right|}
\newcommand{\Norm}[1]{\left\|{#1}\right\|}
\newcommand{\Set}[1]{\left\{{#1}\right\}}
\newcommand{\ddim}{\mathrm{ddim}}
\newcommand{\oridim}{d}
\newcommand{\tardim}{m}

\newcommand{\EE}[1]{\ee\left[{#1}\right]}
\newcommand{\PR}[1]{\Pr\left[{#1}\right]}

\newcommand{\eveA}{\mathscr{A}}
\newcommand{\eveB}{\mathscr{B}}
\newcommand{\eveC}{\mathscr{C}}
\newcommand{\eveE}{\mathscr{E}}

\newcommand{\eveD}{\mathscr{D}}
\newcommand{\eveG}{\mathscr{G}}
\newcommand{\eveH}{\mathscr{H}}
\newcommand{\eveK}{\mathscr{I}}
\newcommand{\indicator}{\mathbbm{1}}
\newcommand{\Indicator}[1]{\indicator\left(#1\right)}
\newcommand{\metrspa}{\mathcal{X}}
\newcommand{\partition}{\mathbf{\Lambda}}

\newcommand{\algPartition}{\textnormal{\textsc{Partition}}}

\newcommand{\algxchg}{\textnormal{\textsc{Modify Decomposition to Eliminate Badly-cut Pairs}}}

\newcommand{\algDecomposeOri}{\textnormal{\textsc{Randomized Hierarchical Decomposition}}}

\newcommand{\ptas}{\textnormal{\textsc{Ptas}}}
\newcommand{\algapx}{\textnormal{\textsc{Approx}}}
\newcommand{\algmedian}{\mathcal{M}}

\newcommand{\cH}{\mathcal{H}}
\newcommand{\cT}{\mathcal{T}}

\newcommand{\U}{U}

\newcommand{\cC}{\mathcal{C}}

\newcommand{\decom}{\cH}
\newcommand{\xchg}{\cT}
\newcommand{\Fa}{F_0}
\newcommand{\Fopt}{F^*}
\newcommand{\Fpiopt}{F_\pi^*}
\newcommand{\cFnew}{F^\prime}

\newcommand{\reprpi}{g_\pi}

\newcommand{\hpi}{h_\pi}
\newcommand{\cFa}{F_0}
\newcommand{\cFopt}{F^*}
\newcommand{\cFpiopt}{F_\pi^*}
\newcommand{\cFapx}{F}
\newcommand{\cFret}{F}
\newcommand{\hole}{\mathrm{Holes}}

\newcommand{\badpi}{\mathrm{Bad}_\pi}

\newcommand{\tO}{\tilde{O}}
\newcommand{\lev}{\mathop{\operatorname{level}}}
\newcommand{\Net}{N}
\newcommand{\Ball}{B}
\newcommand{\Xori}{X}
\newcommand{\sol}{S}
\newcommand{\csol}{S}
\newcommand{\mov}{\psi_\pi}
\newcommand{\cand}{D}

\newcommand{\ignore}[1]{}

\DeclareMathOperator{\E}{\mathbb{E}}
\DeclareMathOperator*{\ee}{\mathbb{E}}

\newcommand{\RR}{{\mathbb{R}}}
\newcommand{\NN}{{\mathbb{N}}}

\DeclareMathOperator{\poly}{poly}
\DeclareMathOperator{\polylog}{polylog}
\DeclareMathOperator{\cost}{cost}
\DeclareMathOperator{\OPT}{OPT}

\DeclareMathOperator{\dist}{dist}
\DeclareMathOperator{\Diam}{diam}

\DeclareMathOperator{\opt}{ufl}
\DeclareMathOperator{\med}{med}

\DeclareMathOperator{\rang}{\overline{diam}}

\DeclareMathOperator*{\argmin}{argmin}

 \newcommand*{\wopen}{\mathfrak{f}}

\usepackage[obeyFinal,textsize=footnotesize]{todonotes}

\let\epsilon\varepsilon
\let\varnothing\emptyset

\def\compactify{\itemsep=0pt \topsep=0pt \partopsep=0pt \parsep=0pt}

\title{Near-Optimal Dimension Reduction for Facility Location}

\author{Lingxiao Huang%
\thanks{
    Email: \texttt{huanglingxiao1990@126.com}
} \\
Nanjing University
\and Shaofeng H.-C. Jiang%
\thanks{
    Email: \texttt{shaofeng.jiang@pku.edu.cn}
} \\
Peking University
\and Robert Krauthgamer%
\thanks{
    Email: \texttt{robert.krauthgamer@weizmann.ac.il}
} \\
Weizmann Institute of Science
\and Di Yue%
\thanks{
    Email: \texttt{di\_yue@stu.pku.edu.cn}
} \\
Peking University
}

\begin{document}
    \begin{titlepage}
        \maketitle
\begin{abstract}
Oblivious dimension reduction, \`{a} la  the Johnson-Lindenstrauss (JL) Lemma,
is a fundamental approach for processing high-dimensional data.
We study this approach for Uniform Facility Location (UFL)
on a Euclidean input $X\subset\RR^d$,
where facilities can lie in the ambient space (not restricted to $X$).
Our main result is that target dimension $m=\tilde{O}(\epsilon^{-2}\ddim)$
suffices to $(1+\epsilon)$-approximate the optimal value of UFL
on inputs whose doubling dimension is bounded by $\ddim$.
It significantly improves over previous results, that could only achieve
$O(1)$-approximation [Narayanan, Silwal, Indyk, and Zamir, ICML 2021]
or dimension $m=O(\epsilon^{-2}\log n)$ for $n=|X|$,
which follows from [Makarychev, Makarychev, and Razenshteyn, STOC 2019].

Our oblivious dimension reduction
has immediate implications to streaming and offline algorithms,
by employing known algorithms for low dimension.
In dynamic geometric streams, 
it implies a $(1+\epsilon)$-approximation algorithm
that uses $O(\epsilon^{-1}\log n)^{\tilde{O}(\ddim/\epsilon^{2})}$ bits of space,
which is the first streaming algorithm for UFL to utilize the doubling dimension.
In the offline setting, it implies a $(1+\epsilon)$-approximation algorithm,
which we further refine to run in time
$( (1/\epsilon)^{\tO(\ddim)} d +  2^{(1/\epsilon)^{\tO(\ddim)}}) \cdot \tilde{O}(n) $. 
Prior work has a similar running time but requires some restriction
on the facilities [Cohen-Addad, Feldmann and Saulpic, JACM 2021]. 

Our main technical contribution is a fast procedure to decompose
an input $X$ into several $k$-median instances for small $k$.
This decomposition is inspired by, but has several significant differences from
[Czumaj, Lammersen, Monemizadeh and Sohler, SODA 2013],
and is key to both our dimension reduction and our PTAS. 
\end{abstract}

         \thispagestyle{empty}
\end{titlepage}

    \newpage
\section{Introduction}
\label{sec:intro}

A fundamental approach for dealing with high-dimensional data
is \emph{oblivious dimension reduction},
in which the dataset $X\subset\RR^d$ is mapped to low dimension 
using a map chosen independently of the data.
A cornerstone of this approach is the Johnson-Lindenstrauss (JL) Lemma~\cite{JL84},
which states that for all $n\ge1$ and $0<\epsilon<1$
there is a randomly chosen linear transformation 
$\pi:\RR^d \to \RR^m$ for $m=O(\epsilon^{-2}\log n)$,
such that for every dataset $X\subset\RR^d$, $|X|=n$, 
with high probability all the pairwise distances in $X$ are preserved within $(1 \pm \epsilon)$-factor, i.e.,
\begin{equation}
  \label{eqn:pairwise_dist}
  \forall x,y \in X, \qquad \|\pi(x) - \pi(y)\|_2 \in (1 \pm \epsilon) \|x - y\|_2.
\end{equation}
This bound on the \emph{target dimension} $m=m(\epsilon,n)$
is known to be asymptotically tight~\cite{LarsenN17}.
In algorithmic applications, one typically applies on the input $X\subset\RR^d$
a map $\pi$ that is chosen independently of $X$, 
and then executes on $\pi(X)\subset\RR^m$ some known algorithm for low dimension.
This approach has generally proved to be extremely useful.

However, in several fundamental algorithmic applications,
target dimension of the form $m=O(\log n)$ is too high to be effective.
We can illustrate this by examples from $3$ different computational settings: 
In offline approximation algorithms,
the traveling salesman problem (TSP) in dimension $m=O(\log n)$
does not admit a PTAS
(i.e., for a sufficiently small but fixed $\epsilon_0>0$, 
no polynomial-time algorithm can achieve $(1+\epsilon_0)$-approximation),
assuming $\mathrm{P} \neq \mathrm{NP}$ \cite{Trevisan00}.
In streaming algorithms, 
approximating the value of Euclidean minimum spanning tree (MST) in dimension $m=O(\log n)$
within $(1+\epsilon_0)$ factor (again, for some fixed $\epsilon_0>0$)
requires $\Omega(\sqrt n)$ bits of storage~\cite{ChenCJLW23}.
In fine-grained complexity,
the diameter of a point set in dimension $m=O(\log n)$
cannot be $(1+\epsilon_0)$-approximated (again, for some fixed $\epsilon_0>0$) 
in quadratic-time, under some complexity assumption \cite{Williams18}.

To break below this barrier of target dimension $m=O(\log n)$, 
one often seeks better bounds for \emph{specific} computational problems.
A prime example is that for $k$-median and $k$-means clustering,
the dimension can be reduced to $m=\tO(\epsilon^{-2}\log k)$ \cite{MakarychevMR19}.\footnote{Throughout, $\tO(f)$ suppresses factors that are logarithmic in $f$.}
This highly nontrivial bound is significantly stronger
than earlier/other bounds~\cite{BoutsidisZD10, CohenEMMP15, BecchettiBC0S19},
and offers a substantial improvement for small $k$. 
It has become famous due to its many applications,
from faster algorithms through better approximation to coreset constructions,
and is useful also in many variants of the problem, like fair clustering. 
Other problems where dimension reduction is successful are Max-Cut, 
where target dimension $m=1/\epsilon^{O(1)}$ suffices~\cite{LammersenSS09,ChenJK23}
(and has immediate implications to streaming algorithms),
and projective clustering problems like $k$-subspace and $k$-flat approximation,
where target dimension that is polynomial in $k$ (but independent of $n$)
suffices~\cite{CohenEMMP15,KerberR15,CW22}.
However, the same method cannot get below $m=O(\log n)$
for the $k$-center problem~\cite{JiangKS24}.

This research plan, which may be called ``beyond JL'',
has another thread that seeks bounds that depend
on the \emph{intrinsic dimensionality} of the dataset $X$ (instead of $n$),
and specifically on a popular measure called the \emph{doubling dimension},
introduced in~\cite{GuptaKL03} based on earlier work by~\cite{Assouad83,Clarkson99}. 
This notion, denoted $\ddim(X)$, is defined as the minimum $t\ge0$ such that
every ball in $X$ can be covered by at most $2^t$ balls of half the radius.\footnote{Formally, the centers of these balls must be in $X$ (see \Cref{def:ddim}),
but relaxing this requirement to center points in the ambient Euclidean space
would change $\ddim(X)$ by at most a constant factor.
}
Observe that $\ddim(X)$ is at most $\log n$ and can often be much smaller,
as this notion generalizes Euclidean dimension and can capture many useful cases,
like points that lie in a linear subspace or have a sparse vector representation, 
and even non-Euclidean distances~\cite{GuptaKL03,GottliebKK14}.

This line of research aims to show that fundamental problems admit
oblivious dimension reduction to dimension $m=m(\epsilon,\ddim(X))$,
and ideally obtain tight bounds.
A prime success story is nearest-neighbor search (NNS),
for which target dimension $m=\ddim(X)/\epsilon^{O(1)}$ indeed suffices~\cite{IndykN07}. 
However, for three important problems, current results fall short of the above aim:
For $k$-center, the known bound on $m$ has also, i.e., in addition to $\ddim(X)$,
an additive term term of $O(\epsilon^{-2}\log k)$ \cite{JiangKS24},
which seems inevitable. 
For MST, the known bound on $m$
has also an additive term of $O(\log\log n)$ \cite{NarayananSIZ21},
and this is still open. 
For uniform facility location (UFL),
the known result achieves only $O(1)$-approximation~\cite{NarayananSIZ21},
and our main contribution is in fact
to significantly improve this approximation factor, from $O(1)$ to $1+\epsilon$.

\paragraph{Uniform Facility Location (UFL).}
In this problem, the input is $X\subset\RR^d$
and an \emph{opening cost} $\wopen>0$,
and the goal is to find a set of \emph{facilities} $F\subset\RR^d$,
so as to minimize the objective
$$
    \cost(X, F) := \wopen \cdot |F| + \sum_{x \in P} \dist(x, F) ,
$$
where $\dist(x, F) := \min_{y \in F} \dist(x, y)$ and $\dist(x, y) := \|x - y\|_2$.
This is actually a clustering problem very similar to $k$-median
(by viewing facilities as cluster centers),
except that the number of clusters $k=|F|$ is not prescribed in advance,
which can make the problem easier, as there is no hard constraint on $k$,
but also harder, as bounds cannot depend on $k$ as a parameter.
We emphasize that our definition allows facilities to lie in the ambient space,
which is natural for a clustering problem (similarly to $k$-median).
Some literature restricts the facilities to a given set,
usually the input points, i.e., $F\subset X$, which can make the problem easier,
e.g., the algorithm or analysis can enumerate the potential facilities.
In contrast, the known dimension reduction for $k$-median~\cite{MakarychevMR19}
is widely applicable but also technically complicated,
precisely because it allows centers to lie in the ambient space.

\begin{remark*} 
A natural approach is to tackle many computational problems at once
by refining the JL Lemma so that $m$ would depend on $\ddim(X)$ instead of on $n$.
Unfortunately, this is not possible using linear maps~\cite[Remark 4.1]{IndykN07},
which is the method of choice employed in the original JL Lemma. 
An open question in the area of metric embedding,
posed by~\cite{LP01,GuptaKL03} (see also \cite[Question 41]{Naor18}),
asks whether every $X\subset\RR^d$ embeds in Euclidean space 
with target dimension and distortion that depend only on $\ddim(X)$
(and not on $d$ or $n$). 
Notice that here, the distortion bound is more relaxed
and the mapping need not be oblivious or even easy to compute,
which would be problematic for algorithmic applications. 
So far, progress on this open question has been made only for a weaker variant
of snowflake embedding~\cite{BartalRS11,NN12,GottliebK15,Neiman16}.
\end{remark*}

\subsection{Results}

We study oblivious dimension reduction for inputs
that reside in a high-dimension Euclidean space
but have a bounded doubling dimension (called in short doubling).
Our main result, in \Cref{theorem:OPT_value_const_dim},
achieves $(1+\epsilon)$-approximation for UFL
using target dimension $m = \tilde{O}(\epsilon^{-2}\ddim(X))$.
It uses a map $\pi:\RR^d\to\RR^m$ that is standard in proofs of the JL Lemma,
and is defined by $\pi: x \mapsto \frac{1}{\sqrt{m}} Gx$
where $G\in\RR^{m \times d}$ is a random matrix
with i.i.d.~entries drawn from Gaussian distribution $N(0,1)$.
We refer to it as a \emph{random linear map},
although some literature calls it random projection
(because it is similar, though not identical,
to orthogonal projection onto a random subspace with scaling).
Throughout, we assume that the opening cost is $\wopen=1$,
which holds without loss of generality by rescaling the input $X\subset\RR^d$,
and denote the optimal value of UFL on input $X\subset\RR^d$
by $\opt(X) := \min \{\cost(X,F): F\subset\RR^d\}$. 
Let $\ddim\ge1$ be a known upper bound on the doubling dimension of $X$,
and assume it is given with the input (or in some settings, computed from it).

\begin{restatable}{theorem}{theoremOPTvalueconstdim}
\label{theorem:OPT_value_const_dim}
Let $0 < \epsilon, \delta < 1$, let $\ddim,d \geq 1$,
and consider a random linear map $\pi$ with suitable target dimension
$m = O(\epsilon^{-2} \ddim \cdot \log(\delta^{-1}\epsilon^{-1} \ddim) )$.
Then for every finite $X \subset \mathbb{R}^d$ with doubling dimension at most $\ddim$,
\begin{equation}
  \label{eqn:OPT_value_const_dim}
  \Pr[ \opt(\pi(X)) \in (1 \pm \epsilon) \opt(X) ] \geq 1 - \delta .
\end{equation}
\end{restatable}

There are two previous bounds on dimension-reduction for UFL. 
For $(1+\epsilon)$-approximation,
it was known that dimension $m=\tO(\epsilon^{-2} \log n)$ suffices,
however when $X$ is doubling our bound is far better. 
That previous bound follows from dimension reduction for $k$-median~\cite{MakarychevMR19},
applied with $k=n$,
but not from the JL Lemma, because facilities in the ambient space $\RR^d$
can evade~\eqref{eqn:pairwise_dist}.
Another previous result~\cite{NarayananSIZ21} is for $O(1)$-approximation,
and shows that dimension $m=O(\ddim(X))$ suffices and is moreover optimal,
namely, the map $\pi$ requires $m=\Omega(\ddim(X))$.\footnote{Strictly speaking, UFL is defined in~\cite{NarayananSIZ21}
with facilities restricted to the input $X$,
but their $O(1)$-approximation applies also in our setting,
because one can move the facilities to lie in $X$ at the cost of factor $2$. 
Our $(1+\epsilon)$-approximation can be adapted also to their setting,
see \Cref{remark:dim_reduction_discrete}. 
}
We stress here that $O(1)$-approximation for UFL is significantly different
from $(1+\epsilon)$-approximation. 
In the former, the facilities can be assumed to lie in the dataset $X$
at the cost of factor $2$ in the approximation, 
whereas in the latter, we know of no effective way
to discretize the potential facilities in the ambient space $\RR^d$,
which is truly high-dimensional and does not satisfy the $\ddim(X)$ bound.
In a sense, \Cref{theorem:OPT_value_const_dim} handles
a regime that falls between low and high dimension.
In fact, the existing tools to tackle this difficulty are quite limited,
as in many problems, such as MST, the ambient space is completely irrelevant.
Perhaps the closest problem is NNS~\cite{IndykN07},
where query points may come from the ambient space,
although the impact of a single query point in NNS
is much less global and complicated than facilities in UFL.

It is worthwhile to juxtapose our result with other computational problems.
For $k$-median, dimension reduction is known to require $m=\Omega(\log k)$,
even for $O(1)$-approximation of doubling inputs~\cite{NarayananSIZ21},
hence we see a sharp contrast with UFL. 
For MST, which can be viewed as a clustering problem,
the known dimension reduction for doubling inputs 
has an $O(\log\log n)$-term in the target dimension~\cite{NarayananSIZ21},
hence our result for UFL may hopefully inspire future improvements.

Our oblivious dimension reduction
has immediate implications to offline and streaming algorithms,
by simply employing known algorithms for low (Euclidean) dimension.
In the offline setting, UFL (and even $k$-median) in $\RR^d$
is known to admit a PTAS,
i.e., $(1+\epsilon)$-approximation for every fixed $\epsilon>0$,
that runs in time $2^{(1/\epsilon)^{O(d)}}\cdot n (\log n)^{d+6}$ \cite{KolliopoulosR07}.
Thus, \Cref{theorem:OPT_value_const_dim} immediately implies 
$(1+\epsilon)$-approximation of the optimal value of UFL,
on input $X\subset\RR^d$ when facilities can lie in the ambient space,
in time $2^{(1/\epsilon)^{\tO(\ddim(X) / \varepsilon^2)}}\cdot dn (\log n)^{\tO(\ddim(X)/\epsilon^2)}$.
(We further improve this bound in \Cref{theorem:main_ptas}.)
We remark that for UFL in doubling metrics (but not necessarily Euclidean),
another known algorithm runs in roughly the same time \cite{Cohen-AddadFS21},
but it restricts the facilities to lie in the dataset $X$.

In the setting of dynamic geometric streams,
the input is a stream of insertions and deletions of points from the grid $[\Delta]^d$,
and $X$ is the point set at the end of the stream.
One usually assumes that its size is $n \leq \poly(\Delta)$,
and then bounds can be written in terms of $d$ and $\Delta$ (but not $n$).
The known algorithm for this setting uses space \(O(\epsilon^{-1}\log \Delta)^{O(d)}\)
and outputs a $(1+\epsilon)$-approximation to the value $\opt(X)$~\cite{CzumajLMS13}.\footnote{The results in~\cite{CzumajLMS13} are stated only for \(d=2\),
but their analysis seems to extend to every dimension \(d\).
}
This exponential dependence of $d$ is essential, because in high dimension
(which can be reduced to $d=O(\log n)$ because of the JL Lemma),
every streaming algorithm that reports an $O(1)$-approximation to $\opt(X)$
requires $\Omega(\sqrt n)$ bits of space~\cite{CzumajJK0Y22}.
Nevertheless, when the doubling dimension of $X$ is low, combining \Cref{theorem:OPT_value_const_dim} with the algorithm of~\cite{CzumajLMS13},
immediately implies a streaming algorithm that uses significantly less space. As stated below,
it essentially decreases the exponent from \(d\) to \(\ddim(X)/\epsilon^{2}\),
which can break below the $\poly(n)$ barrier mentioned above~\cite{CzumajJK0Y22},
e.g., when $\ddim(X)=O(1)$ and $n=\poly(\Delta)$ the space usage is only $\polylog(n)$.

\begin{corollary}
\label{cor:streaming}
There is a streaming algorithm that, given as input $0<\epsilon<1$,
a set $X\subseteq [\Delta]^d$ presented as a stream of point insertions and deletions,
and an upper bound $\ddim$, the algorithm uses space
$\tO\big( d\cdot\polylog(\Delta) + (\epsilon^{-1}\log \Delta)^{\tilde{O}(\ddim/\epsilon^{2})}\big)$
and outputs with high probability a $(1+\epsilon)$-approximation to $\opt(X)$.\footnote{The first term in the space usage is for implementing $\pi$,
  which naively requires $d\log\Delta \cdot \ddim/\epsilon^{2}$ bits,
  using a pseudorandom generator \cite{Indyk06},
  which is now a standard argument.
  It may be improved further if each stream update is a single coordinate
  instead of an entire point.
}
\end{corollary}

\paragraph{PTAS for UFL on Doubling Subsets.}
\Cref{theorem:OPT_value_const_dim} only asserts that the optimal \emph{value} is preserved.
While it is natural to expect that a solution for UFL on $\pi(X)\subset \RR^d$
will yield a solution also for $X\subset \RR^m$, formalizing such a connection is tricky,
because $\pi$ is not invertible 
and there is no natural way to map facilities in $\RR^m$ back to $\RR^d$.

Nonetheless, we use our dimension reduction
in conjunction with a new decomposition procedure that we devise, 
which partitions a UFL instance $X\subset\RR^d$
and effectively reduces it to several $k$-median instances in $\RR^m$,
where $m$ is the target dimension from \Cref{theorem:OPT_value_const_dim}
and $k \approx 2^{O(m)}$.
This is useful because $k$-median can be solved efficiently in this parameter regime,
for instance, one can use a known PTAS that runs
in time $2^{(k/\epsilon)^{O(1)}} dn$~\cite{KumarSS10},
or alternatively
in time $2^{\epsilon^{-O(d)}} n \log^{d+6} n$~\cite{KolliopoulosR07}.
We thus obtain in \Cref{theorem:main_ptas} the \emph{first} PTAS for UFL
on doubling subsets of $\RR^d$ where facilities can lie in the ambient space
--- previous techniques could only handle facilities that are restricted to the dataset $X$,
and we know of no effective way to enumerate the potential facilities in $\RR^d$.
The entire algorithm is very efficient and runs in near-linear time;
it does not even need the input to provide an upper bound $\ddim$,
as offline algorithms can $O(1)$-approximate $\ddim(X)$ quickly.

\begin{restatable}{theorem}{theoremmainptas}
\label{theorem:main_ptas}
There is a randomized algorithm that,
given as input $0 < \epsilon < 1$ and an $n$-point $X\subset\RR^d$,
computes a $(1 + \epsilon)$-approximation for UFL
in time $(2^{m^\prime} \oridim + 2^{2^{m'}}) \cdot \tilde{O}(n)$ for
\begin{equation*}
  m' = O\big( \ddim(X) \cdot \log(\ddim(X)/\epsilon) \big).
\end{equation*}
\end{restatable}
Our new decomposition procedure actually works for all doubling metrics
(even non-Euclidean). 
In that setting, it reduces a UFL instance $X$ to several $k$-median instances
in the same metric space (without dimension reduction),
for $k \approx 2^{O(\ddim \log(\ddim/\epsilon))}$.
These instances can be solved using known algorithms
(based on coresets for $k$-median, see \Cref{append:ptas_discrete})
to obtain a PTAS for UFL that runs in time 
$2^{2^{O(\ddim\cdot \log (\ddim/\epsilon))}} \cdot \tilde{O}(n)$,
when facilities are restricted to the dataset $X$,
and provided oracle access to distances in $X$.
Compared with recent work~\cite{Cohen-AddadFS21}
for a similar setting of all doubling metrics,\footnote{The setting in~\cite{Cohen-AddadFS21} is slightly more general,
  where the facilities are restricted to a given \emph{subset} of $X$,
  rather than all of $X$.
}
our result improves the dependence on $\ddim(X)$
in the double-exponent from quadratic to near-linear,
with comparable dependence on other parameters, e.g., near-linear in $n$.
This expands the recent line of research for pursing fast PTAS for UFL in doubling metrics~\cite{FriggstadRS19,Cohen-AddadFS21}.

\subsection{Technical Contributions and Highlights}
\label{sec:tech_contrib}

Our main technical contribution is a new metric decomposition, 
which partitions a UFL instance that is doubling (not necessarily Euclidean),
into multiple instances, each of low value.
It has the distinctive feature that facilities can lie in a general ambient space, 
while previous decompositions require that also the ambient space is doubling.
Roughly speaking, our decomposition is
a partition $\partition$ of the dataset $X$ into so-called clusters,
such that for a suitable parameter $\kappa = (\ddim/\epsilon)^{\Theta(\ddim)}$,
\begin{enumerate}[(a)] \compactify
\item \label{it:decompose_a}
  every cluster $C\in \partition$ satisfies $\opt(C) = \Theta(\kappa)$; and 
\item \label{it:decompose_b}
  $\sum_{C \in \partition} \opt(C) \in (1\pm\epsilon) \cdot \opt(X)$.
\end{enumerate}
This decomposition is key to both
our dimension-reduction result (\Cref{theorem:OPT_value_const_dim})
and our PTAS (\Cref{theorem:main_ptas}),
and it is probably the first time that metric decomposition
is used to achieve dimension reduction.
Let us highlight the power of this decomposition. 
Property~\eqref{it:decompose_b} guarantees $(1+\epsilon)$-approximation,
which is crucial for surpassing the previous dimension reduction~\cite{NarayananSIZ21},
which achieves only $O(1)$-distortion,
essentially because it is based on a well-known estimate for $\opt(X)$,
from~\cite{MettuP03}, that provides only $O(1)$-approximation.
Property~\eqref{it:decompose_a} bounds the optimal value of clusters
both from below and from above, which is extremely important. 
Moreover, achieving $\kappa$ that is independent of $n$,
and specifically $\kappa = (\ddim/\epsilon)^{\Theta(\ddim)}$,
is a major strength, because $\kappa$ determines the target dimension bound,
which is actually $O(\log \kappa)$.
For comparison, the metric decomposition proposed in~\cite{CzumajLMS13}
achieves $\kappa=\polylog(n)$, which is much weaker,
e.g., it would yield dimension reduction with target dimension $O(\log\log n)$,
and a QPTAS instead of our PTAS.

Our new decomposition uses a \emph{bottom-up} construction,
instead of the previous top-down approach of~\cite{CzumajLMS13}.
Its major advantages is that achieves also a lower bound on $\opt(C)$,
as stated in Property~\eqref{it:decompose_a},
and not only an upper bound that the top-down approach guarantees.
This, in turn, is key for achieving $\kappa$ that is independent of $n$,
because the analysis can charge to the cost of every instance locally.
This bottom-up approach is conceptually similar to sparsity decomposition,
a technique that was crucial to obtain a PTAS for TSP in doubling metrics~\cite{DBLP:journals/siamcomp/BartalGK16, DBLP:journals/siamcomp/ChanHJ18, DBLP:journals/talg/ChanJ18, DBLP:journals/talg/ChanJJ20}.
That technique employs a bottom-up approach as a preprocessing step
to break the dataset into sparse parts that are solved separately,
however the UFL problem and the details of our decomposition
are completely different.

The terms top-down and bottom-up refer to algorithms that use
a hierarchical decomposition of $X$, which is often randomized.
We use Talwar's decomposition~\cite{Talwar04} for a doubling dataset $X$,
which is analogous to a randomly-shifted quadtree in Euclidean space.
Informally, a key feature of this randomized decomposition, denoted by $\cH$,
is that nearby points are ``likely'' to be in the same cluster of $\cH$
(technically, one considers here a suitably chosen level of $\cH$).
For UFL, a crucial aspect is whether each data point $x\in X$ is
in the same cluster as its nearest facility in a fixed optimal solution $F^*$,
and this creates several challenges.
First, an optimal solution $F^*$ is not known to the algorithm
(which is not a concern if $F^*$ is needed only in the analysis),
and a common workaround is to use instead an $O(1)$-approximate solution $F'$,
however it is imperative that the $O(1)$-factor will affect only
the additional cost $\epsilon\cdot\opt(X)$. 
Second, facilities that lie in the ambient space are not even part of $\cH$,
and while conceptually we resolve it similarly to the first challenge,
by replacing $F^*$ with proxy near-optimal facilities $F''\subseteq X$, 
technically it creates complications in our decomposition and its analysis. 
Third, even if we restrict the facilities to lie in the dataset $X$,
the guarantees of $\cH$ are probabilistic,
meaning that some points $x\in X$ (most likely a small fraction)
are not in the same cluster with their ``optimal'' facility,
which precludes us from considering that cluster as a separate instance.

An approach proposed in~\cite{Cohen-AddadFS21}
is to eliminate these so-called badly-cut pairs
by simply moving each such data point $x$ to its ``optimal'' facility,
effectively creating a modified dataset $X'$ with
$\opt(X') \in (1\pm\epsilon) \opt(X)$. 
This is effective if the subsequent steps are applied to $X'$
with no regard to $X$, e.g., running a dynamic-programming algorithm on $X'$. 
However, for our purpose of decomposing $X$ into low-value clusters
(and in turn for our dimension-reduction result),
we still need the probabilistic guarantees of $\cH$,
which apply to $X$, but not to $X'$ that is derived from that same randomness. 

We thus take a different approach of
\emph{modifying the hierarchical decomposition $\cH$}
instead of the data set $X$.
This step eliminates most, but not all, badly-cut pairs,
and we crucially handle the remaining pairs
using the probabilistic guarantees of $\cH$. 
We finally construct the partition $\partition$ by employing
a bottom-up approach on the (modified) hierarchical decomposition. 
In principle, each cluster of $\partition$
arises from a cluster in the hierarchical decomposition,
however these two clusters are not equal and have a more involved correspondence
because of the modifications to $\cH$ and the bottom-up approach.

We remark that our decomposition is designed for UFL,
however many technical steps are general and may find usage in other problems.

\subsection{Proof Overview}
\label{sec:tech_overview}

As mentioned in \Cref{sec:tech_contrib},
our main technical contribution is a new decomposition for UFL instances,
that produces a partition $\partition$ of the dataset $X$
with Properties~\eqref{it:decompose_a} and~\eqref{it:decompose_b} from \Cref{sec:tech_contrib}.
We provide a technical overview of its construction and proof
in \Cref{sec:intro_metric_decomp},
and then use this decomposition to prove our dimension-reduction result
in \Cref{sec:intro_dim_reduction}.
Before proceeding, we briefly describe how this decomposition
immediately implies a PTAS for UFL.

\paragraph{An Immediate PTAS.}
With the new decomposition at hand,
we can immediately obtain a very efficient PTAS for UFL
on a doubling subset $X\subset \RR^d$ (the setting of \Cref{theorem:main_ptas}):
Compute the decomposition $\partition$, 
and then for each cluster $C\in\partition$,
compute a $(1+\epsilon)$-approximate solution for $\opt(C)$.
To implement the last step, observe that by Property~\eqref{it:decompose_a},
an optimal solution for $C$ opens at most $\opt(C)\leq O(\kappa)$ facilities
(recall $\wopen=1$),
and thus $C$ can be solved by an algorithm for $k$-median with $k=O(\kappa)$
(trying also smaller values of $k$). 
It suffices to solve $k$-median within $(1+\epsilon)$-approximation,
which can be done in time $k^{O(k/\epsilon^3)}\cdot \tO(n)$ 
via known approaches based on coresets (see \Cref{lemma:k_median_oracle}).
By Property~\eqref{it:decompose_b},
the union of these solutions for all $C\in\partition$
is a solution for $X$ that achieves $(1+O(\epsilon))$-approximation.
This PTAS almost matches that of \Cref{theorem:main_ptas},
without even using dimension reduction;
more precisely, its running time is roughly
$2^{2^{O(\ddim \log (\ddim/\epsilon))}} \cdot \tilde{O}(nd)$,
whereas \Cref{theorem:main_ptas} decouples $d$ from the doubly-exponential term,
which is significant when $d$ is large,
by using our dimension reduction.

\subsubsection{New Decomposition Procedure}
\label{sec:intro_metric_decomp}

Our new decomposition for UFL is inspired by an earlier one of~\cite{CzumajLMS13},
although our version is more involved and obtains fundamentally stronger bounds.
Let us first recall their approach for an input $X\subset\RR^2$.
Their procedure applies a randomly-shifted quadtree to partition $X\subset\RR^2$,
and then scans the quadtree nodes, which correspond to squares in $\RR^2$,
in a \emph{top-down} manner:
When a square $C$ is examined,
the procedure tests if $\opt(C \cap X) \leq \kappa$
for a suitable threshold $\kappa$. If the test passes, $C\cap X$ is declared as a cluster in the partition $\partition$;
otherwise, the procedure is executed recursively on the $4$ sub-squares of $C$.
This procedure attains $\opt(X) \approx \sum_{C\in\partition} \opt(C)$
by a clever charging argument to the parent squares of low-value clusters,
but it requires setting $\kappa = \polylog(n)$ (or higher),
because the parent squares may be nested and each point inside might be charged $O(\log n)$ times,
which originates from the number of levels in the quadtree.

Our decomposition procedure first constructs
a randomized hierarchical decomposition $\mathcal{H}$ of $X$,
by applying a standard algorithmic tool, due to Talwar~\cite{Talwar04},
that is analogous to a randomly-shifted quadtree but works for all doubling metrics.
This hierarchical decomposition $\mathcal{H}$ has, for every distance scale $2^i$,
a partition of the dataset $X$ into clusters of diameter at most $2^i$,
where the partition for each scale $2^{i - 1}$ refines that for $2^i$. 
Moreover, when this $\mathcal{H}$ is viewed as a tree,
every cluster has at most $2^{O(\ddim)}$ child clusters.
The key guarantee of this hierarchical decomposition is the cutting-probability bound 
\begin{equation} \label{eq:cutting}
  \forall x,y\in X,
  \qquad
  \Pr[\text{$x, y$ are in different clusters of scale $2^i$ }]
  \leq O(\ddim) \cdot \dist(x, y) / 2^i.
\end{equation}

Our decomposition procedure constructs the partition $\partition$ 
by scanning $\mathcal{H}$ in a \emph{bottom-up} manner,
in order to ensure both the upper bound and lower bound in Property~\eqref{it:decompose_a}.
(As explained later, we actually use a modified version of $\mathcal{H}$, denoted $\mathcal{T}$.)
This is in contrast to the top-down approach of~\cite{CzumajLMS13},
which only guarantees an upper bound on $\opt(C)$. 
More precisely, our procedure scans $\mathcal{H}$ bottom-up,
starting from the leaf clusters, and processing each cluster only after its child clusters:
When a cluster $C$ is examined,
and $P$ denotes the current dataset (initialized to $X$),
the procedure tests if $\opt(C \cap P) \geq \kappa$
for a threshold $\kappa  = (\ddim/\epsilon)^{\Theta(\ddim)}$.
If the test passes, $C\cap P$ is added as a cluster in $\partition$,
the points of $C$ are removed from our current dataset $P$,
and the procedure proceeds to the next cluster in $\mathcal{H}$.

\paragraph{Property~\eqref{it:decompose_a}.}
The bottom-up construction clearly attains
a lower bound $\opt(C) \geq \kappa$ for all $C \in \partition$
(except for the very last cluster, which we can handle separately).
To get an upper bound, observe that a cluster $C$ added to $\partition$
is the union of several child clusters that do not pass the test,
i.e., each child $C'$ has $\opt(C' \cap P) < \kappa$.
The number of children is at most $2^{O(\ddim)}$,
and the union of their optimal solutions is clearly a feasible solution for $C$,
hence $\opt(C) \leq 2^{O(\ddim)} \cdot \kappa$. 
This establishes Property~\eqref{it:decompose_a},
up to relaxing the ratio between the upper and lower bounds to be $2^{O(\ddim)}$;
the formal treatment appears in \Cref{lemma:opt_C_bounds}.

Unfortunately, this bottom-up approach has ramifications 
that complicate the entire analysis.
In particular, a cluster $C$ that is added to $\partition$
is no longer a cluster in the hierarchy $\mathcal{H}$,
because some of its descendants in $\mathcal{H}$ might have been removed earlier.
This misalignment with $\mathcal{H}$ makes it difficult to use
the cutting-probability bound~\eqref{eq:cutting},
which applies to the clustering in $\mathcal{H}$ but not that in $\partition$.
We thus introduce the notion of ``holes'' (\Cref{def:holes}),
which captures the parts of $C \in \partition$ that were removed
(when comparing to this $C$ in $\mathcal{H}$).
For sake of simplicity,
we ignore for now the holes and pretend we are directly analyzing $\decom$,
and we also ignore the complications arising from the ambient space
by assuming that facilities lie inside the dataset $X$.
We will return to discuss these issues later in the section.

\paragraph{Property~\eqref{it:decompose_b}.}
This property is borrowed from~\cite{CzumajLMS13},
but our proof is completely different, because of the different construction.
The high-level idea is to take a set of facilities $F^*\subseteq X$
that is optimal for $X$, i.e., it realizes $\opt(X)$,
and transform it into a modified set $F'$
by \emph{adding facilities} inside each cluster $C\in\partition$.
This $F'$ aligns with our partition $\partition$,
because data points in each cluster $C$ are ``served locally''
by facilities in $F' \cap C$.
We will need to show that, in expectation, 
$\cost(X, F') \leq (1 + \epsilon) \cost(X, F^*) = (1 + \epsilon) \opt(X)$.
To simplify this overview,
we present the construction of $F'$ in a more intuitive but less accurate way:
Start with $F'=F^*$, then examine each $C \in \partition$
and add to $F'$ a set $N_C \subseteq C$ of extra facilities.

To define this set $N_C$ we need the notion of a net,
which is a standard method to discretize a metric space,
and is particularly powerful in doubling metrics.
Formally, a \emph{$\rho$-net} of a point set $S$ is a subset $N\subset S$,
such that the distance between every two points in $N$ is at least $\rho$, 
and every point in $S$ has at least one point of $N$ within distance $\rho$.
Let $N_C$ be an $(\epsilon'\cdot\Diam(C))$-net of $C$,
for $\epsilon' := \epsilon / \ddim$;
we remark that an $(\epsilon\cdot\Diam(C))$-net may seem sufficient here,
however the finer net is needed to compensate for the $O(\ddim)$ factor
in the cutting-probability bound~\eqref{eq:cutting}.
A standard bound on the size of a net
implies that $|N_C| \leq O(1/\epsilon')^{\ddim} = O(\ddim/\epsilon)^{\ddim}$.

\paragraph{Increase in Cost.}
We bound the cost increase $\cost(X,F') - \cost(X,F^*)$
by splitting it into two parts, the opening cost and the connection cost.
The increases in opening cost of a cluster $C$ is at most
$|N_C| \leq \epsilon \kappa \leq \epsilon \opt(C)$
by our choice of $\kappa$ and Property~\eqref{it:decompose_a},
and in total over all clusters, it is at most
$\sum_{C \in \partition} |N_C| \leq \epsilon \sum_{C \in \partition} \opt(C)$,
which can be charged to the left-hand side of Property \eqref{it:decompose_b},
that we shall eventually bound.
For the connection cost of each $C \in \partition$,
recall that we only use facilities in $C \cap F'$,
even though the nearest facility to $x\in C$ might be outside $C$,
and thus the increase in connection cost for $C$ is at most 
$
\Delta_C := \sum_{x \in C} \dist(x, F' \cap C) - \sum_{x \in C} \dist(x, F^*).
$

Now consider $x \in C$ and let $F^*(x)$ be its nearest point in $F^*$.
Observe that if $F^*(x)\in C$ then $\dist(x, F' \cap C) \leq \dist(x, F^*)$,
i.e., there is no increase,
and therefore the nontrivial case is when $F^*(x)$ is outside $C$.
A simple idea is to serve $x$ by its nearest neighbor in $N_C$,
which has connection cost $\dist(x, N_C) \leq \epsilon'\cdot \Diam(C)$.
However, this bound might be much larger than $\dist(x, F^*)$,
and we shall to eliminate this situation by ensuring a \emph{separation property}:
\begin{equation} \label{eq:separation}
  \forall x \in C,
  \qquad
  F^*(x) \notin C \ \Longrightarrow\  \dist(x, F^*(x)) \geq \epsilon'\cdot\Diam(C).
\end{equation}
Indeed, this inequality implies that $\dist(x, N_C)\leq \dist(x, \cFopt)$,
hence serving $C$ by facilities in $F' \cap C$ (instead of $F^*$)
does not increase the connection cost.

\paragraph{Eliminating ``Badly Cut'' Pairs.}
We ensure this separation property~\eqref{eq:separation}
using the concept of ``badly-cut'' pairs from~\cite{Cohen-AddadFS21}.
Let $x\in X$, and call a pair $(x, F^*(x))$ \emph{badly cut}
if it is cut in the hierarchical decomposition $\mathcal{H}$
at some distance scale $2^i > \dist(x, F^*(x)) / \epsilon^\prime$. 
Observe that if a pair is not badly cut,
then every cluster $C$ in $\mathcal{H}$ that contains $x$ but not $F^*(x)$ 
must have $\Diam(C) \leq \dist(x, F^*(x))/\epsilon'$.
Thus, eliminating all badly-cut pairs ensures the separation property~\eqref{eq:separation}.

The badly-cuts pairs are eliminated in~\cite{Cohen-AddadFS21}
by simply moving $x$ to $F^*(x)$ whenever $(x, F^*(x))$ is badly cut.
By the cutting-probability bound~\eqref{eq:cutting},
this happens with probability at most $O(\epsilon' \cdot \ddim) = O(\epsilon)$,
hence these movements modify $X$ into $X'$ that satisfies
$\E_{\mathcal{H}}[\opt(X')] \in (1 \pm O(\epsilon))\cdot \opt(X)$,
which we can afford.
The overall strategy here is to first define $X'$ from $X$,
and then add to $F^*$ (which is now the optimal facilities for $X'$)
more facilities to obtain $F'$.

This approach of moving points is effective as a local fix,
as it does not change $\opt(X)$ by too much,
however it is not useful for globally decomposing $X$ into clusters
that satify Properties~\eqref{it:decompose_a} and~\eqref{it:decompose_b}.
We take an alternative approach of modifying $\mathcal{H}$
(as outlined in \Cref{alg:xchg})
into a new hierarchical decomposition $\mathcal{T}$,
in which no pair $(x, F^*(x))$ is badly cut.
We then construct the final partition $\partition$ from this $\mathcal{T}$,
rather than from $\mathcal{H}$ (in \Cref{alg:bot_up_partition}).
Overall, we establish a refined version of the separation property~\eqref{eq:separation},
as detailed in \Cref{lemma:partition_properties}.

\paragraph{Handling Holes.}
Recall that our bottom-up decomposition might create ``holes'' (\Cref{def:holes}),
because a cluster $C$ that is added to $\partition$
might have some of its points removed earlier,
and we let $\hole_C$ be the set of clusters in $\partition$
that contain these earlier-deleted points.
We can handle holes and still obtain Property~\eqref{it:decompose_b}
using essentially the same arguments as before. 
When we consider $x \in C$ for some $C \in \partition$ and $F^*(x) \notin C$,
we use the net $N_C$ of $C$ (same as before)
only when $F^*(x)$ does not belong to any cluster in $\hole_C$,
and we use $N_{\widehat{C}}$ when $F^*(x) \in \widehat{C}$ for some $\widehat{C} \in \hole_C$.
We need to add the nets $N_{\widehat{C}}$ to our opening cost, but this extra cost can be charged to $\opt(\widehat{C})$, and each $\widehat{C}$ is charged only once by the observation that
$\hole_C \cap \hole_{C'} = \emptyset$ for distinct $C, C' \in \partition$.

\paragraph{Facilities in the Ambient Space.}
When facilities can lie in the ambient space, which need not be doubling,
we face the major obstacle that the tools we developed,
like the cutting-probability bound~\eqref{eq:cutting}
and the separation property~\eqref{eq:separation}
need not apply to the optimal set of facilities $F^* \subset \RR^d$. 
Another, more technical, obstacle is that the net $N_C$
(and similarly $N_{\widehat{C}}$ for $\widehat{C} \in \hole_C$)
might cover the doubling subset $C$ but not $F^*(x)$. 

Our plan is to pick for each $F^*(x)$ a \emph{proxy}
in the dataset $X$ (which is doubling),
adapt our previous arguments to work for that proxy,
and use this to argue about $\Fopt$.
Specifically, the proxy of a facility $F^*(x)$ is its closest point in $X$
that is served (in the optimal solution $F^*$) by the same facility,
formalized by a mapping $g : F^* \to X$, where
$g(F^*(x)) = \argmin_{y \in X} \{\dist(y, \Fopt(x)) :\ \Fopt(y)=\Fopt(x)\}$.
To use these proxies, we modify the step that eliminates badly-cut pairs
to handle pairs $(x, g(\Fopt(x)))$ for $x\in X$,
and obtain the separation property for these pairs. We then show that this translates also to a separation property for $(x, \Fopt(x))$.

However, we cannot apply the previous argument about cost increase, 
because it used that $\Fopt(x)$ is ``covered'' by some net $N_C$
(or $N_{\widehat{C}}$ for $\widehat{C} \in \hole_C$).
We need new steps in the analysis, and a particularly nontrivial case
is when $g(\Fopt(x)) \in \widehat{C}$ for some $\widehat{C}\in \hole_C$.
Now, if the proxy $g(\Fopt(x))$ is close enough to $\Fopt(x)$,
then we can pretend that $\Fopt(x) = g(\Fopt(x))$ and the analysis goes through.
And if they are far apart (compared with $\dist(x, g(\Fopt(x)))$),
then we crucially make use of the \emph{optimality} of $\Fopt$,
and show that $\Fopt(x)$ must be near $\widehat{C}$,
namely, within distance $O(\Diam(\widehat{C}))$.
These facts imply that $x$ is close to $\widehat{C}$,
hence $x$ can be covered by the net $N_{\widehat{C}}$.
This net is fine enough and thus
contains a point within distance $\epsilon \cdot \dist(x, \Fopt(x))$ from $x$
(here we use the separation property between $(x, \Fopt(x))$),
and we can use that net point to serve $x$ instead of $\Fopt(x)$,
with no additional connection cost.
We remark that these steps generally work for any ambient space beyond Euclidean $\mathbb{R}^d$.

\subsubsection{Dimension Reduction}
\label{sec:intro_dim_reduction}

Our proof of dimension reduction for UFL,
i.e., that with high probability $\opt(\pi(X)) \in (1 \pm \epsilon) \opt(X)$,
heavily relies on our decomposition to provide a structurally simple
characterization of the optimal value, namely,
$\opt(X) \in (1\pm\varepsilon)\cdot \sum_{C \in \partition} \opt(C)$. 
At a high level, our proof shows that the right-hand side
is ``preserved'' under a random linear map $\pi$.

We need to prove both an upper bound and a lower bound on $\opt(\pi(X))$. 
The upper bound is easy,
as observed in recent work~\cite{MakarychevMR19,NarayananSIZ21,JiangKS24},
because we may consider one optimal solution $F^*$ for $X$
and analyze its image under $\pi$,
i.e., the cost of the solution $\pi(F^*)$ for $\pi(X)$.
Since we only need $\pi$ to preserve this one specific solution, 
target dimension $m = O(\poly(\epsilon^{-1}))$ suffices.

The lower bound is more interesting and is where we use our decomposition of $X$,
which implies $\opt(X) \geq (1-\epsilon)\cdot \sum_{C \in \partition} \opt(C)$.
We would like to show this inequality is ``preserved'' under $\pi$,
i.e., ``carries over'' to the target space,
and what we actually show, as explained further below, is that
\begin{equation} \label{eq:apx_pi_X_by_clusters}
  \opt(\pi(X)) \geq \sum_{C \in \partition} \opt(\pi(C)) - \epsilon \cdot \opt(X).
\end{equation}
Notice that the additive error here $\epsilon\cdot \opt(X)$
might not be directly comparable to $\opt(\pi(X))$.
Nevertheless, this bound~\eqref{eq:apx_pi_X_by_clusters} turns out to suffice,
because we only need to show in addition that
\begin{equation} \label{eq:projectionC}
  \sum_{C \in \partition} \opt(\pi(C))
  \geq
  (1-\epsilon) \sum_{C \in \partition} \opt(C).
\end{equation}
Putting together \eqref{eq:apx_pi_X_by_clusters}, \eqref{eq:projectionC}
and Property~\eqref{it:decompose_b}
will then conclude the desired lower bound.

The proof of~\eqref{eq:projectionC} relies on~\cite{MakarychevMR19}, as follows.
Let $\med_k(S)$ denote the optimal value of $k$-median on $S\subset\RR^d$;
then we know from~\cite{MakarychevMR19}
that target dimension $\tO(\epsilon^{-2}\log k)$
suffices for dimension reduction for $k$-median,
meaning that for every $S\subset\RR^d$,
with high probability $\med_k(\pi(S)) \in (1\pm\epsilon) \med_k(S)$.
We apply this in our case by letting $S$ be a cluster $C\in\partition$,
and we know from Property~\eqref{it:decompose_a} that the number of facilities
needed for $C$ is at most $O(\kappa) = (\ddim/\epsilon)^{O(\ddim)}$,
hence target dimension $\tO(\epsilon^{-2}\ddim)$ suffices for it.
The only gap is that we need to apply~\cite{MakarychevMR19} multiple times
for our summation over all $C \in \partition$.
We handle this in a series of lemmas
(\Cref{lemma:piX_leq_X_value_general,lemma:jl_small_opt_contraction_pr})
that are based on~\cite{MakarychevMR19},
and bound the additive error for each $C \in \partition$ by
$\opt(C) - \opt(\pi(C)) \leq e^{-\epsilon^2 m} \cdot \poly(\kappa) \leq \epsilon \kappa$.
We then use the fact that $|\partition|$, the number of terms in the summation,
is roughly $\opt(X) / \kappa$ (\Cref{lemma:size_of_partition}),
and thus the total additive error is at most $\epsilon\cdot \opt(X)$,
which we can afford.

Finally, we briefly discuss the proof of~\eqref{eq:apx_pi_X_by_clusters},
which overall is similar to that of Property~\eqref{it:decompose_b}
and its formal treatment appears in \Cref{lemma:apx_pi_X_by_clusters}.
We let $F^*_\pi$ be an optimal set of facilities for $\pi(X)$,
and we modify it into $F'_\pi$ that is ``consistent'' with $\partition$,
i.e., in every cluster $C \in \partition$,
all points $x \in \pi(C)$ are served by facilities in $F'_\pi \cap \pi(C)$.
Implementing this plan encounters new difficulties,
and we focus here on one immediate issue --
that we have to analyze $\pi(X)$, which is random.
To address this, 
we condition on the event $\mathcal{E}_C$, for $C \in \partition$,
that the distances between points in $N_C$ (which is a net on $C$)
and all other data points (a doubling point set) are preserved simultaneously.
For this event to hold with high probability,
it suffices that $m = \tilde{O}(\epsilon^{-2}\ddim)$
(see \Cref{lemma:ball_expansion,lemma:ball_contraction}),
similarly to a lemma from~\cite{IndykN07}
about preserving the nearest-neighbor distance from a query point to a doubling point set.
This is the sole use of the randomness of $\pi$ in this analysis.

     \subsection{Related Work}
\label{sec:related}

Oblivious dimension reduction can be useful in various models of computation,
and one may wonder about algorithms that run in different models
and approximate UFL on high-dimensional Euclidean inputs, 
i.e., inputs as in our results but without the doubling condition. 
For offline approximation in polynomial time,
the state-of-the-art is $(2.406 + \epsilon)$-approximation for UFL,
which follows from the same ratio for $k$-median~\cite{Cohen-AddadEMN22}. 
Aiming for fast approximation algorithms,
one can achieve $O(1/\epsilon)$-approximation in time $\tilde{O}(n^{1 + \epsilon})$~\cite{GoelIV01},
via a reduction to nearest neighbor search.
This reduction-style result was recently improved to be fully-dynamic,
with a similar tradeoff between approximation ratio and time~\cite{bhattacharya2024dynamic}.
In dynamic geometric streams, known algorithms achieve
$O(d/\log d)$-approximation using $\poly(d\log n)$ space,
or $O(1/\epsilon)$-approximation using space $n^{O(\epsilon)} \poly(d)$,
both using a technique of geometric hashing~\cite{CJK+22:arxiv}.
This geometric-hashing technique was recently used in the setting of
massively parallel computing (MPC), to design fully-scalable MPC algorithms
that achieve $O(1/\epsilon)$-approximation in $O(1)$ rounds
using $n^{1+\epsilon}\poly(d)$ total space~\cite{CzumajGJKV24}.

 \section{Preliminaries}
\label{sec:prelim}

Let $(\metrspa, \dist)$ be a metric space. 
The \emph{ball} centered at $x \in \metrspa$ with radius $r > 0$
is defined as $\Ball(x, r):=\Set{y\in \metrspa \colon \dist(x, y) \leq r}$. 
The \emph{$r$-neighborhood} of a point set $X\subseteq \metrspa$
is defined as $\Ball(X, r):=\bigcup_{x\in X} \Ball(x, r)$.
The \emph{diameter} of a point set $X \subseteq \metrspa$
is defined as $\Diam(X) := \max_{x, y} \dist(x, y)$,
and its \emph{aspect ratio} (or \emph{spread}), denoted $\Delta(X)$,
is the ratio between the diameter and the minimum inter-point distance in $X$.
For a point set $X \subseteq \metrspa$ and a point $u \in \metrspa$,
let $X(u)$ denote the point of $X$ that is nearest to $u$.
Denote by $\opt^S(\Xori)$ the optimal UFL value for input $X\subseteq \metrspa$
when facilities are restricted to the set $S\subseteq \metrspa$,
and let $\opt(\Xori) := \opt^{\metrspa}(\Xori)$ for short.

\begin{definition}[Doubling dimension~\cite{GuptaKL03}]
\label{def:ddim}
The \emph{doubling dimension} of a metric space $(\metrspa, \dist)$ 
is the smallest $t\ge0$ such that every metric ball can be covered by at most $2^t$ balls of half the radius.
The doubling dimension of a point set $X\subseteq\metrspa$ is 
the doubling dimension of the metric space $(X, \dist)$,
and is denoted $\ddim(X)$.
\end{definition}

\begin{definition}[Packing, covering and nets]
Consider a metric space $(\metrspa, \dist)$ and let $\rho > 0$. 
A point set $S\subseteq\metrspa$ is a \emph{$\rho$-packing}
if for all $x, y \in S$, $\dist(x, y) \geq \rho$.
The set $S$ is a \emph{$\rho$-covering} for $X$
if for every $x \in X$, there is $y \in S$ such that $\dist(x, y) \leq \rho$.
The set $S$ is a $\rho$-\emph{net} for $X$
if it is both a $\rho$-packing and a $\rho$-covering for $X$.
\end{definition}

\begin{proposition}[Packing property~\cite{GuptaKL03}]
    \label{prop:packing}
    If $S$ is $\rho$-packing then $|S| \leq (2\Diam(S) / \rho)^{\ddim(S)}$.
\end{proposition}

We summarize below the properties of the random linear map $\pi$ are used in this paper.
Recall that $\pi : x \mapsto \frac{1}{\sqrt m} Gx$ where $G \in \mathbb{R}^{m \times d}$ is a random Gaussian matrix.
In some previous work, such as~\cite{MakarychevMR19},
only properties~\eqref{eqn:rlm_expand} and~\eqref{eqn:rlm_expect} below were needed,
and they may hold for other maps $\pi$.
We need also~\eqref{eqn:rlm_contract}, which seems to be more specific to Gaussian.

\begin{proposition}[Properties of random linear maps]
    \label{prop:random_proj}
    Let $\pi\colon \RR^\oridim\to\RR^\tardim$ be a random linear map.
    Then for every unit vector $x\in \RR^\oridim$ and every $t>0$,
    \begin{align}
        \Pr[\Norm{\pi(x)}\not\in 1 \pm t] &\leq e^{-t^2\tardim/8} . \label{eqn:rlm_expand}\\
        \Pr[\Norm{\pi(x)}\leq 1/t] &\leq \left(\frac{3}{t}\right)^\tardim . \label{eqn:rlm_contract}\\
        \EE{ \max\left\{0, \Norm{\pi(x)} - (1+t)\right\} }
        &\leq \frac{1}{mt}e^{-t^2 \tardim/2} .
        \label{eqn:rlm_expect}
    \end{align}
\end{proposition}
\begin{proof}
The bounds~\eqref{eqn:rlm_expand} and \eqref{eqn:rlm_contract}
were established in~\cite[Eq. (7)]{IndykN07} (see also~\cite[Eq. (5)(6)]{NarayananSIZ21}).
To prove~\eqref{eqn:rlm_expect},
we need a known tail bound
$    \PR{\Norm{\pi(x)} \geq 1+t} \leq e^{-t^2 m / 2}.
$
Denote random variables $\xi = \Norm{\pi(x)} - 1$ and $\eta = \max\{0, \xi - t\}$.
Then 
\begin{align*}
    \EE{ \max\left\{0, \Norm{\pi(x)} - (1+t)\right\} }
    &= \EE{\eta}
    = \int_0^{\infty} \PR{\eta \geq u}\,  \mathrm{d} u
    = \int_t^{\infty} \PR{\xi \geq u}\,  \mathrm{d} u\\
& \leq \int_t^{\infty} e^{-u^2 m / 2}\,  \mathrm{d} u
    \leq \int_t^{\infty} \frac{u}{t} \cdot e^{-u^2 m / 2}\,  \mathrm{d} u
    = \frac{1}{mt} e^{-t^2 m / 2}.
\qedhere
\end{align*}
\end{proof}

\section{A New Decomposition for UFL}
\label{sec:partition}

This section introduces our new decomposition for UFL instances,
which technically is a random partition $\partition$ of the dataset $\Xori$,
and effectively reduces the UFL instance into separate low-value UFL instances,
each formed by a different part $C \in \partition$. 
Throughout this section,
we assume that $(\mathcal{X}, \dist)$ is an underlying metric space
and $X\subseteq \mathcal{X}$ is a dataset of doubling dimension at most $\ddim$.
A feasible UFL solution is a set of facilities,
which can be any (finite) subset of $\metrspa$. 
We present the construction of the partition $\partition$
in \cref{subsec:partition_construction},
which includes a summary of its main properties
in \cref{lemma:opt_C_bounds,lemma:size_of_partition,lemma:apx_pi_X_by_clusters}.
We then prove these lemmas in \cref{subsec:opt_C_bounds,subsec:bounding_partition,subsec:cost_lb}, respectively.
The partition $\partition$ is parameterized by $\kappa \ge 1$
(in addition to $0<\varepsilon<1$). 
\begin{restatable}[Bounded local UFL values]{lemma}{LemmaoptCbounds}
    \label{lemma:opt_C_bounds}
    For every $\kappa \ge 1$, the random partition $\partition = \partition(\kappa)$ always satisfies that $\kappa \leq \opt(C) \leq 2^{10\ddim} \kappa$ for all $C \in \partition$.
\end{restatable}

In our applications, we set $\kappa := (\ddim/\varepsilon)^{\Theta(\ddim)}$.
This ensures that $\opt(C)$ is small enough for dimension reduction analysis,
and in particular an optimal solution $\opt(C)$
uses at most $2^{10\ddim}\kappa$ facilities, 
hence finding $\opt(C)$ reduces to a $k$-median problem with $k \leq 2^{O(\ddim)}\kappa$.
This is useful in several ways.
For instance, a target dimension $m = \tilde{O}(\ddim / \varepsilon^2)$
suffices to preserve $\opt(\pi(C)) \in (1 \pm \epsilon) \opt(C)$,
via a black-box application of~\cite{MakarychevMR19},
which shows that target dimension $\tilde{O}(\epsilon^{-2}\log k)$ suffice for $k$-median.
Similarly, as we mentioned, there are also efficient $(1 + \epsilon)$-approximation algorithms for $k$-median with such small $k$,
which implies a PTAS for $\opt(C)$.

\begin{restatable}[Bounding $|\partition|$]{lemma}{Lemmasizeofpartition}
    \label{lemma:size_of_partition}
    There exist universal constants $c_1, \alpha > 0$, 
    such that for every $\varepsilon \in (0, 1)$ and $\kappa > 2(\ddim/\varepsilon)^{c_1\cdot \ddim}$, 
    the partition $\partition = \partition(\kappa)$ satisfies
    \begin{align}
        \EE{\norm{\partition}} \leq \frac{2 \alpha \cdot \opt(\Xori)}{\kappa - 
        2(\ddim/\varepsilon)^{c_1\cdot \ddim}} ,
        \label{eq:lemma:size_of_partition} 
    \end{align}
    where the randomness is over the construction of $\partition$.
\end{restatable}

\Cref{lemma:size_of_partition} essentially says that
$|\partition| \leq O(\opt(X) / \kappa)$.
This is particularly useful when comparing $\sum_{C \in \partition} \opt(C)$
with $\sum_{C \in \partition} \opt(\pi(C))$
in the dimension-reduction analysis,
where we bound the additive error for each $C \in \partition$
by $\opt(C) - \opt(\pi(C)) \leq \epsilon \kappa$.
\Cref{lemma:size_of_partition} then implies that the total additive error
is at most $O(\varepsilon)\cdot \opt(\Xori)$, which we can afford.

We note that the above two lemmas hold for every doubling point set $X$.
The next lemma is specifically for $X\subset \RR^\oridim$
(i.e., for the Euclidean metric space $\RR^\oridim$),
and it analyzes the performance of dimension reduction on $\partition$.
This technical lemma provides a lower bound for $\opt(\pi(\Xori))$
in terms of the local costs $\opt(\pi(C))$ for $C \in \partition$.
This is crucially useful in our dimension reduction analysis.

\begin{restatable}[Lower bound for $\opt(\pi(\Xori))$]{lemma}{lemmaapxpiXbyclusters}
    \label{lemma:apx_pi_X_by_clusters} 
    Let $\pi\colon \RR^\oridim\to \RR^\tardim$ be a random linear map,
    and let $\Xori \subset \RR^\oridim$ be finite with doubling dimension at most $\ddim$. 
    There exist universal constants $c_1, c_2, c_3>0$, such that 
    for every $\varepsilon, \delta \in (0, 1)$, 
    if $\kappa>c_2 (\ddim/(\delta\varepsilon))^{c_1 \cdot \ddim}$
    and $\tardim>c_3 (\log \kappa + \log (1/\delta\varepsilon))$, 
    then \begin{align*}
      \Pr\Big[ 
        \opt(\pi(\Xori))\geq \sum_{C\in\partition} \opt(\pi(C)) -
      \varepsilon \cdot \opt(\Xori)
      \Big]
      \geq 1-\delta ,
    \end{align*}
    where the randomness is over both $\pi$ and $\partition = \partition(\kappa)$.
\end{restatable}

In fact, using similar techniques, we can prove a result analogous to this lemma
but for general metric $(\metrspa, \dist)$ and (finite) doubling subset $\Xori \subseteq \metrspa$,
where $\pi$ is fixed to the identity map,
i.e., $\opt(X) \geq \sum_{C \in \partition} \opt(C) - \epsilon \cdot \opt(X)$;
see \cref{corollary:apx_X_by_clusters}.

\subsection{The Construction of $\partition$}
\label{subsec:partition_construction}

Our construction of $\partition$ has three steps.
The first one is to compute for $\Xori$ a randomized hierarchical decomposition $\cH$,
using the algorithm of Talwar~\cite{Talwar04}.
We restate this computation of $\cH$ in \Cref{alg:decompose_ori},
and review its main properties. 
The second step modifies $\decom$ into another hierarchical decomposition $\xchg$,
to eliminate badly-cut pairs (a notion introduced by~\cite{Cohen-AddadFS21}).
As described in \Cref{alg:xchg},
it works by moving points between clusters separately at each level,
and thus each level remains a partition of $X$,
but the nesting across levels might break. The third step constructs the random partition $\partition$ from $\xchg$, using a bottom-up approach, as described in \Cref{alg:bot_up_partition}. 
We summarize in \cref{lemma:partition_properties} several properties of $\partition$
that follow directly from the construction,
including a separation and a consistency property,
and are essential for proving \cref{lemma:opt_C_bounds,lemma:size_of_partition,lemma:apx_pi_X_by_clusters}.

\paragraph{Random Hierarchical Decomposition~\cite{Talwar04}.}
We use an algorithm of Talwar~\cite{Talwar04}
to construct a random hierarchical decomposition $\decom$,
described in \cref{alg:decompose_ori}.
Let $\gamma :=\min\{\dist(x,y)\colon x\neq y\in \Xori\}$,
let $\Delta:=\Diam(\Xori)/\gamma$ be the aspect ratio of $\Xori$,
and denote $\ell := \lceil \log \Delta\rceil$.
At a high level, the algorithm (randomly) partitions $\Xori$ into clusters,
and then recursively partitions each cluster into children clusters,
where each recursive call decreases the diameter bound by a factor of $2$.
This process creates a recursion tree, where tree nodes correspond to clusters,
and this is referred to as a hierarchical decomposition $\cH$.
The randomness comes from two sources:
(1) the scaling factor $\rho$, picked in \cref{line:select_rho},
which affects the diameter of clusters in \cref{line:radius}; and
(2) the permutation $\mu$, picked in \cref{line:random_perm},
which determines the order in which clusters are formed in \cref{line:new_cluster}. 
By construction, $\cH$ has $\ell+2$ levels.
The root node, at the highest level of $\cH_{\ell+1}$,
corresponds to the trivial cluster $\Xori$,
and each leaf at the lowest level $\cH_0$ corresponds to a single point of $X$. 
Each node $C\in \cH_i$ is the union of all its children at $\cH_{i-1}$; see \cref{line:new_cluster}.
Moreover, clusters at every level $\cH_i$ form a partition of $X$,
and every cluster $C\in \cH_i$ satisfies
$\Diam(C) \leq 2r_i \leq 2^{i}\gamma$. 
We denote the \emph{diameter-bound} of this cluster by $\rang(C):=2^{i} \gamma$,
and its level by $\lev(C) := i$.

\begin{algorithm}
    \caption{\algDecomposeOri~\cite{Talwar04}}
    \label{alg:decompose_ori}
    \DontPrintSemicolon
\KwIn{finite point set $X\subset \RR^d$ with
        minimum distance $\gamma$ and aspect ratio $\Delta$}
construct nested nets
        $\Xori=\Net_0\supset \Net_1\supset\dots\supset \Net_\ell$,
        such that each $\Net_i$ is a $(2^{i-3}\gamma)$-net of $\Net_{i-1}$, where $\ell = \lceil \log \Delta \rceil$\;
        pick $\rho\in (\frac{1}{2},1)$ uniformly at random \label{line:select_rho}\;
        pick $\mu$ as a random permutation of $\Xori$ \label{line:random_perm}\;
        $\cH_{\ell+1} \gets \{\Xori\}$\;
        \For{$i = \ell,\ell-1 \ldots, 0$}{
            $\cH_i\gets\varnothing$ and 
             $r_i\gets \rho\cdot 2^{i-1}\gamma$ \label{line:radius}\;
            \For{cluster $C\in \cH_{i+1}$}{
                \For{each $y\in \Net_i$}{
                    $C_{y} \gets C\cap \Ball(y,r_i) \setminus
                    \bigcup_{z\in N_i: \mu(z)<\mu(y)} \Ball(z,r_i)$ \label{line:new_cluster}
                    \tcp*{new cluster, a child of $C$}
                    $\cH_i\gets \cH_i\cup \{C_y\}$
                    \tcp*{can skip if $C_y=\emptyset$}
                }
            }
}
        \Return $\decom \gets \{ \cH_0, \cH_1, \ldots, \cH_{\ell + 1} \}$\;
\end{algorithm}

We say that a pair $x, \widehat{x}\in X$ is \emph{cut} at level $i$ 
if there are two distinct clusters $C\neq\widehat{C}\in \cH_i$
with $x\in C$ and $\widehat{x}\in \widehat{C}$.
We state below a well-known bound on the probability to be cut in $\cH$.

\begin{proposition}[Cutting probability~\cite{Talwar04}]
\label{prop:decomposition_scaling}
For every pair $x, \widehat{x}\in \Xori$ and level $i$,
\[
  \Pr[ \text{$(x, \widehat{x})$ is cut at level $i$} ]
  \leq O\left(
    \frac{\ddim \cdot \dist(x, \widehat{x})}{2^i \gamma}
  \right).
\]
\end{proposition}

This bound has been used extensively in previous work,
e.g., to argue that nearby points are unlikely to be cut at a high level. 
We also need the following notion of a badly-cut pair.
A similar notion was first introduced in~\cite{Cohen-AddadFS21},
where it is defined with respect to a metric ball,
whereas we focus on a pair of points.

\begin{definition}[Badly-cut pairs]
\label{def:badly_cut}
Let $\varepsilon \in (0, 1)$.
A pair of points $x, \widehat{x} \in \Xori$
is called \emph{$\varepsilon$-badly cut} with respect to $\decom$
if $(x, \widehat{x})$ is cut at any level
$i \geq \log\frac{\ddim \cdot \dist(x, \widehat{x})}{\varepsilon^2\gamma}$.
\end{definition}

\begin{lemma}[Badly-cut probability]
\label{lemma:pr_badly_cut}
Let $\varepsilon \in (0, 1)$.
Then for every pair $x, \widehat{x}\in \Xori$,
\[
  \Pr[ \text{$(x, \widehat{x})$ is $\varepsilon$-badly cut} ]
  \leq
  O(\varepsilon^2) .
\]
\end{lemma}

\begin{proof}
Denote $i_0=\lceil \log \frac{\ddim \dist(x, \widehat{x})}{\varepsilon^2\gamma} \rceil$. 
By \cref{prop:decomposition_scaling}, 
$\Pr[ \text{$(x, \widehat{x})$ is $\varepsilon$-badly cut} ] \leq \sum_{i\ge i_0} O(\ddim) \cdot 2^{-i} \dist(x, \widehat{x})/\gamma \leq O(\ddim)\cdot 2^{-i_0+1} \dist(x, \widehat{x})/\gamma \leq O(\varepsilon^2)$.
\end{proof}

Fix an $\alpha$-approximate solution $\cFa$ 
for the UFL problem on $\Xori$ with $\alpha = O(1)$, 
such that $\cFa\subseteq \Xori$.
(Such a solution always exists by moving the facilities of an optimal solution
to their nearest point in the dataset $\Xori$.)
Recall that $\Fa(x)$ denotes the closest facility to $x$ in $\cFa$. 
Our proof examines not only that a pair $(x, \widehat{x})$ is not badly cut,
but also that related pairs are not badly cut, as described next.

\begin{definition}[Good pairs]
\label{def:good_pair}
Let $\varepsilon \in (0, 1)$.
A pair of points $x, \widehat{x}\in X$ is called
\emph{$\varepsilon$-good} with respect to $(\decom, \cFa)$, 
if none of the three pairs 
$(x, \widehat{x})$, $(x, \Fa(x))$ and $(\widehat{x}, \Fa(\widehat{x}))$ is $\varepsilon$-badly cut with respect to $\decom$.
When $\varepsilon, \decom, \cFa$ are clear from the context,
we may omit them and simply say that $(x, \widehat{x})$ is good. 
\end{definition}

The following lemma is an immediate corollary of 
\cref{lemma:pr_badly_cut} by the union bound.

\begin{lemma}[Probability to be good]
\label{lemma:prob_good_pair}
Let $\varepsilon\in (0, 1)$.
Every pair of points $x, \widehat{x} \in \Xori$ (that does not depend on $\cH$)
is $\varepsilon$-good with probability at least $1-O(\varepsilon^2)$. 
\end{lemma}

Our plan is to construct a partition $\partition$ of $X$
so that it has the so-called separation and consistency properties.
Informally, the separation property means that for every $x \in \Xori$,
if $x$ and $\Fa(x)$ belong to different clusters $C\neq\widehat{C} \in \partition$, 
then $\dist(x, \Fa(x))$ is roughly lower bounded by $\Omega(\epsilon^2 / \ddim)$
times the maximum of $\rang(C)$ and $\rang(\widehat{C})$.
This property enables us to ``represent'' a global solution $\cFa$ with respect to some local centers around clusters in $\partition$.
Consistency means that every cluster in $\partition$ originates 
from a cluster in $\decom$,
and has diameter bound that is not much larger. 
This property allows us to use a fine net with bounded size as a proxy for candidate centers.

\paragraph{Procedure for Eliminating Badly-Cut Pairs.}
To achieve the separation property,
we need to eliminate all badly-cut pairs.
A simple way to eliminate the badly-cut pairs,
which was used in~\cite{Cohen-AddadFS21},
is to build a new dataset $X'$
by moving every point $x\in X$ for which $(x, F_0(x))$ is badly cut to the point $F_0(x)$.
However, this $X'$ clearly depends on the randomness of $\mathcal{H}$,
and thus \Cref{prop:decomposition_scaling} does not apply to $X'$
(which is actually needed in our subsequent analysis).
Hence, we introduce a more sophisticated procedure, in \Cref{alg:xchg},
that directly modifies the clusters in $\decom$ (instead of building a new dataset),
and our $\partition$ is then built from the modified decomposition. 

\begin{algorithm}
\caption{$\algxchg(X, \decom, \cFa, \varepsilon)$}
\label{alg:xchg}
\DontPrintSemicolon
\For{$i = 0,  \dots, \ell + 1$}{
        for each $C \in \decom_i$, let $C^\xchg \gets C$\;
        \For{$x\in X$}{
            \label{line:move_start}
            find $C, \widehat{C} \in \decom_i$ 
            such that $x \in C$ and $\Fa(x)\in \widehat{C}$\;
            \If{$C \neq \widehat{C}$ \emph{and}
            $i\geq \log \frac{\ddim \cdot \dist(x, \Fa(x))}{\varepsilon^2 \gamma}$}{let 
                $C^\xchg \gets C^\xchg \setminus \{x\}$ and
                $\widehat{C}^\xchg \gets \widehat{C}^\xchg \cup \{x\}$\label{line:reassign}\;
            }
        }\label{line:move_end}
        $\xchg_i \gets \{C^\xchg: C \in \decom_i\}$
        \tcp*{modified partition of $X$}
}
    \Return $\xchg \gets \{ \xchg_0, \ldots, \xchg_{\ell + 1}\}$\;
\end{algorithm}

The modified decomposition $\xchg$ is constructed level by level.
Initially, $\xchg$ is a copy of $\decom$. 
Then separately for each level $ 0 \leq i \leq \ell + 1$,
clusters at level $i$ 
exchange their points in the following way:
for every point $x\in X$, if $(x, \Fa(x))$ is cut at level $i$ and
$i\geq \log \frac{\ddim \cdot \dist(x, \Fa(x))}{\varepsilon^2 \gamma}$, 
then $x$ is moved from its current cluster to the cluster containing $\Fa(x)$ 
(Lines \ref{line:move_start}-\ref{line:move_end}).
Notice that $F_0(x)$ never moves (because $F_0(F_0(x)) = F_0(x)$)
and thus the order of processing $x \in X$ does not matter.

\paragraph{Relation between $\xchg$ and $\decom$.}
It is easy to see that every level $\xchg_i \in \xchg$ still forms a partition of $X$. 
We also let $\xchg$ inherit the tree structure from $\decom$,
using the one-to-one correspondence between their clusters (ignoring empty clusters),
and we write $C^\decom$ to denote the cluster in $\decom$
corresponding to a cluster $C^\xchg$ in $\xchg$.
Observe that now a node $C^\xchg \in \xchg_i$ \emph{is not} necessarily the union of its children at $\xchg_{i-1}$.
Although the abovementioned one-to-one correspondence exists between $\xchg$ and $\decom$, a significant difference is that
an actual cluster $C^\xchg \in \xchg_i$ need not be the union of all its children in $\xchg_{i-1}$.

\paragraph{Properties of $\xchg$.}
We can reinterpret \Cref{def:badly_cut} of badly-cut pairs with respect to $\xchg$
(recall it was originally defined with respect to $\decom$): 
A pair $(x, \widehat{x})$ is \emph{$\epsilon$-badly cut} with respect to $\xchg$
if there exists a level $i \geq \log \frac{\ddim \cdot \dist(x, \Fa(x))}{\varepsilon^2 \gamma}$
and different clusters $C^\xchg\neq \widehat{C}^\xchg \in \xchg_i$, 
such that $x\in C^\xchg$ and $\widehat{x} \in \widehat{C}^\xchg$.
The next fact follows immediately from the steps of \Cref{alg:xchg}.
\begin{fact}
\label{fact:xchg_no_bad}
Every pair $(x, F_0(x))$ for $x \in X$ is not badly cut with respect to $\xchg$.
\end{fact}

The next lemma shows that $\xchg$ maintains consistency with $\decom$,
i.e., the diameter of each cluster $C^\xchg$ does not exceed that of $C^\cH$ by much.
Recall that $\rang(C^\cH)=2^{i} \gamma$ for all $C^\cH\in \cH_i$,
and that for a point set $Y$ and $r > 0$,
we denote $\Ball(Y, r)=\bigcup_{x\in Y} \Ball(x, r)$.

\begin{lemma}[Consistency of $\xchg$]
\label{lemma:xchg_local}
Let $\varepsilon \in (0, 1)$ and $\xchg = \xchg(\Xori, \decom, \cFa, \varepsilon)$ be constructed by \cref{alg:xchg}. 
Then for every $i \in \{0, 1, \dots, \ell+1\}$ and cluster $C^\xchg \in \xchg_i$,
it holds that $C^\xchg \subseteq \Ball(C^\decom, \varepsilon^2 \cdot 2^i \gamma)$,
and thus 
$C^\xchg \subseteq \Ball(C^\decom, \varepsilon^2 \rang(C^\decom))$.
\end{lemma}

\begin{proof}
For every point $x\in C^\xchg \setminus C^\decom$, $\Fa(x)\in C^\decom$
    and $i \geq \log \frac{\ddim \cdot \dist(x, \Fa(x))}{\varepsilon^2 \gamma}$.
    Hence, $\dist(x, \Fa(x)) \leq \frac{\varepsilon^2 \rang(C^\decom)}{\ddim}
    \leq \varepsilon^2 \rang(C^\decom)$.
    This completes the proof.
\end{proof}

\paragraph{Constructing the Partition $\partition$.}
We can now present \cref{alg:bot_up_partition} the construction of $\partition$,
which works in a bottom-up manner, as follows. 
Given a threshold $\kappa>0$,
we find the lowest-level cluster $C$ in $\xchg$ such that $\opt(C)\geq \kappa$, and add it to the partition $\partition$ (\cref{line:critical_condition,line:add_cluster}). 
We then remove the points of $C$ from $X$ and from every cluster in $\xchg$ (\cref{line:update_clusters}). 
We repeat this procedure until all points in $\Xori$ are removed,
or not suitable $C$ exists, in which case we simply add the remaining points in $X$ as a separate part (\Cref{line:add_remain}). 
It is easy to see that the output $\partition$ forms a partition of $X$. 

We remark that the last cluster $C$ added to $\partition$
might have $\opt(C^\xchg) < \kappa$, which violates \cref{lemma:opt_C_bounds}.
This special cluster does not affect the correctness of \cref{lemma:size_of_partition,lemma:apx_pi_X_by_clusters}
and thus for simplicity, we assume that all clusters $C \in \partition$ satisfy $\opt(C) \geq \kappa$.
To remove this assumption, we can also merge the last two clusters added to $\partition$,
as it would violate the upper bound on $\opt(C)$ by at most factor $2$.

\begin{algorithm}
\caption{$\algPartition(X, \xchg,\kappa)$}
\label{alg:bot_up_partition}
\DontPrintSemicolon
\While{$X\neq \varnothing$}{
  let $0 \leq i \leq \ell$ be the smallest integer such that
  there is $C \in \xchg_i$ with $\opt(C) \geq \kappa$\;
  \label{line:critical_condition}
  \If{\emph{such $i, C$ exist}}{
    $\partition \gets \partition \cup \{C\}$
    \label{line:add_cluster}\;
    $X\gets X\setminus C$,
    and update for every $j$ all 
    clusters $\widehat{C} \in \xchg_j$ by 
    $\widehat{C} \gets \widehat{C} \setminus C$
    \label{line:update_clusters}\;
}
  \Else{
    $\partition \gets \partition \cup \{X\}$ \label{line:add_remain}
    \tcp*{add last cluster}
    $X\gets \varnothing$\;
  }
}
\Return $\partition$\;
\end{algorithm}

\paragraph{Relation between $\partition$ and $\xchg$.}
Recall that there is a one-to-one correspondence between clusters in $\xchg$ and $\decom$.
We can define a relation also between clusters in $\partition$
and in $\xchg$ (and hence in $\decom$),
by tracking the steps in \Cref{alg:bot_up_partition}.
Specifically, a cluster $C \in \partition$ is usually added to $\partition$ in \cref{line:add_cluster},
so there is a clearly defined correspondence with this cluster $C$ in $\xchg$. 
In the exceptional case of the last cluster, added in \cref{line:add_remain}),
$C$ contains all remaining points so we can define its corresponding cluster
in $\xchg$ to be the root, which is the entire dataset $X$. 
For a part $C \in \partition$, we write $C^\xchg$ and $C^\decom$
to denote its corresponding clusters in $\xchg$ and $\decom$.

Due to the bottom-up nature of the construction of $\partition$,
clusters $C\in \partition$ may not be perfectly aligned with its corresponding cluster $C^\xchg\in \xchg$.
To see this, consider a cluster $C \in \partition$,
and suppose another cluster $\widehat{C} \in \partition$
was added to $\partition$ before $C$ during the execution of \Cref{alg:bot_up_partition}. 
If $\widehat{C}^\xchg$ is a descendant of $C^\xchg$,
then we must remove $\widehat{C}$ from $C^\xchg$ when constructing $\partition$,
which makes the cluster $C$ a subset of $C^\xchg \setminus \widehat{C}$ (instead of a full cluster in $\xchg$).
Thus, we observe that $C\subseteq C^\xchg$ holds for every $C\in \partition$.
Next, we define the following structure called \emph{holes} for clusters $C\in \partition$ to capture such misalignment between $C$ and $C^\xchg$. 

\begin{definition}[Holes]
\label{def:holes}
A cluster $\widehat{C}\in \partition$ is called a \emph{hole} of $C\in \partition$
if among all clusters in $\partition$,
$C$ is the one whose corresponding $C^\xchg$
is the lowest-level ancestor of $\widehat{C}^\xchg$ (in $\xchg$).
The set of holes of $C\in\partition$ is defined as
$\hole_C:= \{\widehat{C}\in \partition :\ \text{$\widehat{C}$ is a hole of $C$} \}$.
\end{definition}

\begin{lemma}[Total number of holes]
    \label{lemma:size_of_holes}
    $\sum_{C\in \partition} 
    \norm{\hole_C} \leq \norm{\partition}$.
\end{lemma}
 
\begin{proof}
    By definition, each $\widehat{C}\in\partition$ is a hole
    of at most one $C$, i.e., $\hole_C \cap \hole_{C^\prime} = \varnothing$ for distinct $C, C^\prime \in \partition$. 
    Therefore, the total size of all 
    $\hole_C$ is upper bounded by the size of $\partition$.
\end{proof}

Finally, the following lemma summarizes the desired properties of $\partition$, which are useful for dimension-reduction analysis.

\begin{lemma}
\label{lemma:partition_properties}
Consider a random partition $\partition = \partition(\Xori, \xchg, \kappa)$.
    \begin{enumerate}[(1)]
        \item \textbf{Separation}: \label{item:separation} 
        For every $\varepsilon$-good pair $(x, \widehat{x})$
        with respect to $(\decom, \cFa)$ with $x\in C$, $\widehat{x}\in \widehat{C}$ and $C\neq \widehat{C}\in \partition$, the following holds.
        \begin{enumerate}[(a)]
            \item \label{item:sep_not_descendant}
              If $C^\decom$ and $\widehat{C}^\decom$
              are not related (as descendant-ancestor) in $\decom$, then 
            $\dist(x, \widehat{x}) \geq \frac{\varepsilon^2}{\ddim} \cdot \max\{\rang(C^\decom), \rang(\widehat{C}^\decom)\}$.
            \item \label{item:sep:hole} If $\widehat{C}^\decom$ is 
            a descendant of $C^\decom$ in $\decom$, then there 
            exists a cluster $\widetilde{C}\in \hole_C$, such that
            $\dist(x, \widehat{x}) \geq \frac{\varepsilon^2}{\ddim} \cdot \rang(\widetilde{C}^\decom)$.
        \end{enumerate}
        \item \textbf{Consistency}: \label{item:consistency}
        For every cluster $C\in \partition$, it holds that
        $C\subseteq \Ball(C^\decom, \varepsilon^2 \rang(C^\decom))$.
    \end{enumerate}
\end{lemma}

Let us explain the separation property of \Cref{lemma:partition_properties}.
Case \eqref{item:sep_not_descendant} is more intuitive,
because if $C^\decom$ and $\widehat{C}^\decom$ are not related in $\decom$
(related means that one is ancestor of the other), 
then $C^\decom \cap \widehat{C}^\decom = \varnothing$,
which implies the distance lower bound.
However, in case \eqref{item:sep:hole},
$\widehat{C}^\decom$ is a subset of $C^\decom$,
meaning that $\dist(\widehat{x}, C^\decom) = 0$,
which corresponds to the misalignment in $\partition$ discussed earlier.
We thus need to use $\hole_C$ to obtain the separation property,
which is a major structural complication for our bottom-up construction of $\partition$.
In particular, in our later arguments where we wish to find a net $N_C$ for $C \in \partition$
whose granularity depends on the separation guarantee in \Cref{lemma:partition_properties},
we not only need a net for $C$ but also a series of nets on clusters in $\hole_C$.

\begin{proof}[Proof of \cref{lemma:partition_properties}]
We first show the separation property.
By the definition of $\varepsilon$-good pairs (\cref{def:good_pair}), 
neither of $x$ and $\widehat{x}$ is moved to another cluster
during the execution of \Cref{alg:xchg}. 
Thus $x\in C^\decom$ and $\widehat{x} \in \widehat{C}^\decom$.

If $C^\decom$ and $\widehat{C}^\decom$ are not related in $\decom$
(related means that one is ancestor of the other), 
then $x$ and $\widehat{x}$ are cut at level
$\max\{\lev(C^\decom), \lev(\widehat{C}^\decom)\}$ of $\decom$. 
Since $(x, \widehat{x})$ is not $\varepsilon$-badly cut with respect to $\decom$,
we have $\max\{\lev(C^\decom), \lev(\widehat{C}^\decom)\} \leq \log \frac{\ddim \cdot \dist(x, \widehat{x})}{\varepsilon^2 \gamma}$,
or equivalently,
$\dist(x, \widehat{x}) \geq \frac{\varepsilon^2}{\ddim} \max\{\rang(C^\decom), \rang(\widehat{C}^\decom)\}$.

If $\widehat{C}^\decom$ is a descendant of $C^\decom$ in $\decom$, 
then there exists $\widetilde{C}\in \hole_C$, 
such that $\widetilde{C}^\decom$ is a descendant of $C^\decom$ and an 
ancestor of $\widehat{C}^\decom$,
and that $(x,\widehat{x})$ is cut at level $\lev(\widetilde{C}^\decom)$ of $\decom$.
Since $(x, \widehat{x})$ is not $\varepsilon$-badly cut with respect to $\decom$, we have 
$\lev(\widetilde{C}^\decom) \leq \log \frac{\ddim\cdot \dist(x, \widehat{x})}{\varepsilon^2 \gamma}$,
or equivalently,
$\dist(x, \widehat{x}) \geq \frac{\varepsilon^2}{\ddim} \rang(\widetilde{C}^\decom)$.

Finally, observe that $C\subseteq C^\xchg$,
hence the consistency of $\partition$ follows immediately from 
the consistency of $\xchg$ (\cref{lemma:xchg_local}).
\end{proof}

\subsection{Proof of \cref{lemma:opt_C_bounds}: Bounded Local UFL Values}
\label{subsec:opt_C_bounds}

\LemmaoptCbounds*

\begin{proof}[Proof of \cref{lemma:opt_C_bounds}]
    The lower bound follows immediately from our construction
    (\cref{alg:bot_up_partition}). 
Recall that we denote the corresponding cluster
    of $C$ with respect to the hierarchical decomposition $\decom$ and $\xchg$ by $C^\decom$ and $C^\xchg$, respectively.
Denote by $\Gamma_C$ the set of children of $C$ with respect to both 
    $\decom$ and $\xchg$ (recalling that their clusters have one-to-one correspondence),
    where for each $\widehat{C}\in \Gamma_C$, the corresponding clusters are denoted by $\widehat{C}^\decom$ and $\widehat{C}^\xchg$, respectively.

    By \cref{alg:bot_up_partition}, once $C$ is constructed, 
    it must be the lowest level ``heavy'' cluster with respect to $\xchg$,
    i.e., $\opt(C) \geq \kappa$.
Then at that moment, $\opt(\widehat{C}) < \kappa$ for every $\widehat{C} \in \Gamma_C$ in \Cref{line:critical_condition} of \Cref{alg:bot_up_partition}.
    By packing property, $\norm{\Gamma_C}\leq 2^{5 \ddim}$.
    Thus, $\opt(\bigcup_{\widehat{C}\in \Gamma_C} \widehat{C}) \leq 2^{5 \ddim} \kappa$.
It seems that we have already completed the proof of \Cref{lemma:opt_C_bounds}.
However, we need to be careful since the construction of $\partition$ relies on $\xchg$, whose cluster $C^\xchg$ may not be the union of its children $\widehat{C}^\xchg$s.
Thus, we also need to bound the UFL value for the set $C \setminus \bigcup_{\widehat{C}\in \Gamma_C} \widehat{C}$, which must be a subset of 
    $C^\xchg \setminus \bigcup_{\widehat{C}\in \Gamma_C} \widehat{C}^\xchg$.
Below, we focus on bounding the UFL value for $C^\xchg \setminus \bigcup_{\widehat{C}\in \Gamma_C} \widehat{C}^\xchg$.
    
    We first claim that $C^\xchg \setminus \bigcup_{\widehat{C}\in \Gamma_C} \widehat{C}^\xchg \subseteq C^\xchg \setminus C^\decom$.
    Indeed, for every point $x \in C^\decom \cap C^\xchg$, $x$ must not be removed from $C^\xchg$ in \cref{line:reassign} of \Cref{alg:xchg}.
Then either $(x, \Fa(x))$ is not badly cut at $\lev(C^\decom)$, 
    or both $x$ and $\Fa(x)$ are in $C^\decom$.
    In the former case, $(x, \Fa(x))$ is not badly cut at lower 
    levels of $\lev(C^\decom)$ and thus $x$ must not be removed from any descendant $\widehat{C}^\xchg$ of $C^\xchg$ in \cref{line:reassign} of \Cref{alg:xchg}. 
In the latter case, both $x$ and $\Fa(x)$ are in $\bigcup_{\widehat{C}\in \Gamma_C} \widehat{C}^\decom$, which implies that $x$ can only be exchanged between 
    clusters in $\Gamma_C$.
Therefore, we have $x \in \bigcup_{\widehat{C}\in \Gamma_C} \widehat{C}^\xchg$.
Overall, we conclude that $C^\xchg \setminus \bigcup_{\widehat{C}\in \Gamma_C} \widehat{C}^\xchg \subseteq C^\xchg \setminus C^\decom$.
Thus, it suffices to bound the UFL value for $C^\xchg \setminus C^\decom$.

    It follows from the above analysis that for any
    $x\in C^\xchg \setminus C^\decom$, $(x,\Fa(x))$ is badly cut
    \emph{exactly} at $\lev(C^\decom)$.
Fix such an $x$ and suppose $x\in \widetilde{C}$ where $\widetilde{C}$ is a cluster at level $\lev(C^\decom)-1$. 
Then $x \in \widetilde{C}^\decom \cap \widetilde{C}^\xchg$.
By \cref{lemma:xchg_local}, $x$ is within the range 
    $\Ball(C^\decom, \varepsilon^2 \rang(C^\decom))$.
Thus, $\widetilde{C}^\decom$ is within the range
    $\Ball(C^\decom, (1+\varepsilon^2) \rang(C^\decom))$.
    Recall that the centers of clusters at level $\lev(C^\decom)-1$ 
    form a $(\rang(C^\decom)/16)$-net.
By the packing property, there are at most $(32(2+\varepsilon^2))^\ddim$ clusters at level
    $\lev(C^\decom)-1$ that
    can make a contribution to 
    $C^\xchg \setminus C^\decom$.
    Since each of these $\widetilde{C}$ is light with 
    $\opt(\widetilde{C}) \leq \kappa$ when $C$ is added to $\partition$, we have 
    $\opt(C^\xchg \setminus C^\decom) \leq (32(2+\varepsilon^2))^\ddim\cdot \kappa$.

    In conclusion, 
    \[
    \opt(C) \leq \opt(\bigcup_{\widehat{C}\in \Gamma_C} \widehat{C}) + \opt(C \setminus \bigcup_{\widehat{C}\in \Gamma_C} \widehat{C}) \leq \opt(\bigcup_{\widehat{C}\in \Gamma_C} \widehat{C}) + \opt(C^\xchg \setminus C^\decom) \leq 2^{10\ddim} \kappa,
    \]
    which completes the proof of \Cref{lemma:opt_C_bounds}.
\end{proof}

\subsection{Proof of \Cref{lemma:size_of_partition}: Bounding $|\mathbf{\Lambda}|$}
\label{subsec:bounding_partition}

\Lemmasizeofpartition*

We first give the following technical lemma,
which states that the sum of the (local) UFL values of $C \in \partition$
is a constant approximation for the global optimal UFL value (up to a small additive term depending on $\EE{\norm{\partition}}$).

\begin{lemma}
    \label{lemma:apx_opt_x_tech}
    Let $0 < \varepsilon < 1$, $\kappa > 0$. 
    Let $\partition=\partition(\kappa)$ be 
    constructed by \cref{alg:bot_up_partition}. 
    Then there is a constant $c_1 > 0$, such that
    \begin{align}
        \EE{\sum_{C\in \partition}\opt(C)}
        \leq 2 \alpha \opt(\Xori)
        + \left(\frac{\ddim}{\varepsilon}\right)^{c_1\cdot \ddim}
        \cdot 2\EE{\norm{\partition}}.
        \label{eq:lemma:apx_opt_x_tech}
    \end{align}
\end{lemma}

Note that \cref{lemma:size_of_partition} is a direct corollary of this lemma.
\begin{proof}[Proof of \cref{lemma:size_of_partition}]
    By \cref{lemma:opt_C_bounds}, each cluster $C\in \partition$ satisfies $\opt(C)\geq \kappa$. 
    Hence, $\ee[\sum_{C\in\partition}\opt(C)]\geq \kappa\cdot \EE{\norm{\partition}}$.
Then \Cref{lemma:size_of_partition} follows immediately from \eqref{eq:lemma:apx_opt_x_tech} of \Cref{lemma:apx_opt_x_tech}.
\end{proof}

We now turn to prove \cref{lemma:apx_opt_x_tech}.
For every $C\in \partition$, we first define a feasible UFL solution $S_C$ for $C$ (see \eqref{eq:def_sol_S_C}),
and relate the connection cost of every single point $x\in C$ in this solution, to $\dist(x, \cFa)$. 
The construction of this solution $S_C$ utilizes the property that $(x, \Fa(x))$ is never badly cut with respect to 
$\xchg$ (\Cref{fact:xchg_no_bad}), and is useful for bounding $\opt(C)$.
Concretely, we have the following lemma to state this $S_C$.

\begin{lemma}
    \label{lemma:feasibel_sol_const_apx}
    Let $0 < \varepsilon < 1$, $\kappa > 0$.
    Let $\partition=\partition(\kappa)$ be 
    constructed by \cref{alg:bot_up_partition}.
Define the following local solution $S_C$ for $C$,
    \begin{align}
        \csol_C:=(\cFa\cap C) \cup \Net_C \cup 
        \left(\bigcup_{\widehat{C}\in \hole_C}\Net_{\widehat{C}}\right), 
        \label{eq:def_sol_S_C}
    \end{align}
    where $\Net_C$ is an $\frac{\varepsilon^3}{\ddim}\cdot \rang(C^\decom)$-net
    of $\Ball(C^\decom, \rang(C^\decom)) \cap \Xori$,
    and for every $\widehat{C} \in \hole_C$, 
    $\Net_{\widehat{C}}$ is an $\frac{\varepsilon^3}{\ddim}\cdot \rang(\widehat{C}^\decom)$-net 
    of $\Ball(\widehat{C}^\decom, \rang(\widehat{C}^\decom))\cap \Xori$.
Then \[
    \forall x\in C, \quad \dist(x, \csol_C) \leq 2 \dist(x, \Fa(x)).
    \]
\end{lemma}

\begin{proof}
    Let $\beta := \frac{\varepsilon^3}{\ddim}$.
Denote $\widehat{x} := \Fa(x)$. 
If $\widehat{x} \in C$, then 
    $\dist(x, \csol_C)\leq \dist(x, \widehat{x}) \leq 2\dist(x, \Fa(x))$.
Thus, it suffices to prove for the case that $\widehat{x} \notin C$.

    Assume $\widehat{x} \in \widehat{C}\neq C$. 
We consider the following cases.
    \begin{itemize}
        \item $\widehat{C}^\xchg$ is not a descendant of $C^\xchg$ with respect to the 
        decomposition $\xchg$. 
Since $\widehat{x}=\Fa(x)$ is never deleted from $\widehat{C}^\xchg$ throughout \Cref{alg:xchg}, 
        $x$ and $\widehat{x}$ are in different clusters at level  
        $\lev(C^\xchg)$ of $\xchg$.
By \cref{fact:xchg_no_bad}, $(x, \Fa(x))$ is never $\varepsilon$-badly cut with respect to $\xchg$.
        Hence, $\lev(C^\xchg)\leq \log \frac{\ddim \cdot \dist(x, \widehat{x})}{\varepsilon^2 \gamma}$.
        Equivalently, $\dist(x, \widehat{x}) \geq \frac{\varepsilon^2 \rang(C^\decom)}{\ddim}$.
Since $x\in C^\xchg$, by \cref{lemma:xchg_local}, $x\in \Ball(C^\decom, \varepsilon^2 \rang(C^\decom))$.
Thus, by the definition of $\Net_C$, we have 
        \[
        \dist(x, \csol_C)\leq \dist(x, \Net_C)\leq \beta \rang(C^\decom)\leq \frac{\varepsilon^2 \rang(C^\decom)}{\ddim} \leq \dist(x, \widehat{x}) = \dist(x, \Fa(x)).
        \]

        \item $\widehat{C}^\xchg$ is a descendant of $C^\xchg$ with respect to the decomposition $\xchg$. 
        In this case, there exists $\widetilde{C}\in \hole_C$, such that $\widetilde{C}^\xchg$ is the descendant of $C^\xchg$ and the ancestor of $\widehat{C}^\xchg$. 
        Analogous to the former case, $x$ and $\widehat{x}$ are cut at level
        $\lev(\widetilde{C}^\xchg)$ with respect to $\xchg$.
        Without loss of generality, we assume $\widehat{C} = \widetilde{C}$.
        Since $(x, \widehat{x})$ is not $\varepsilon$-badly cut (\cref{fact:xchg_no_bad}), we have 
        $\dist(x, \widehat{x}) \geq \frac{\varepsilon^2 \rang(\widehat{C}^\decom)}{\ddim}$. 
        By the covering property of $\Net_{\widehat{C}}$, there is a net 
        point $u\in \Net_{\widehat{C}}$ such that 
        $\dist(\widehat{x}, u) \leq \beta \rang(\widehat{C}^\decom) \leq \frac{\varepsilon^2 \rang(\widehat{C}^\decom)}{\ddim} \leq \dist(x, \widehat{x})$. 
Thus, we have 
        \[
        \dist(x, \csol_C) \leq \dist(x, u) \leq \dist(x, \widehat{x})+\dist(\widehat{x}, u)\leq 2\dist(x, \widehat{x})= 2\dist(x, \Fa(x)).
        \]
    \end{itemize}

    This completes the proof of \cref{lemma:feasibel_sol_const_apx}.
\end{proof}

Now we are ready to prove \cref{lemma:apx_opt_x_tech}.

\begin{proof}[Proof of \cref{lemma:apx_opt_x_tech}]
    For every $C\in \partition$, construct a local solution $S_C$
    for $C$ as \eqref{eq:def_sol_S_C}. 
By \cref{lemma:feasibel_sol_const_apx}, $\dist(x, \csol_C)\leq 
    2 \dist(x, \Fa(x))$ for every $x\in C$. 
Then we have
    \begin{align}
        \EE{\sum_{C\in \partition}\opt(C)}
        &\leq \EE{\sum_{C\in \partition}\cost(C, \csol_C)}
        =\EE{\sum_{C\in \partition}\sum_{x\in C}\dist(x, \csol_C) 
        + \norm{\csol_C}}\notag\\
        &\leq 2 \sum_{x\in X}\dist(x, \Fa(x))
        +\EE{\sum_{C\in \partition} \norm{\cFa\cap C}}
        +\EE{\sum_{C\in\partition} \left(\norm{\Net_C}+\sum_{\widehat{C}\in \hole_C} \norm{\Net_{\widehat{C}}}\right)}\label{eq:use_claim_tech}\\
        &\leq 2 \sum_{x\in X}\dist(x, \Fa(x)) + \norm{\cFa}
        +\left(\frac{\ddim}{\varepsilon}\right)^{O(\ddim)}
        \left(\EE{\norm{\partition}} 
        + \EE{\sum_{C\in\partition} \norm{\hole_C}}\right)\label{eq:use_packing}\\
        &\leq 2 \EE{\cost(\Xori, \cFa)}
        + \left(\frac{\ddim}{\varepsilon}\right)^{O(\ddim)}
        \cdot 2\EE{\norm{\partition}} \label{eq:use_size_of_holes}\\
        &\leq 2 \alpha \opt(\Xori)
        + \left(\frac{\ddim}{\varepsilon}\right)^{O(\ddim)}
        \cdot 2\EE{\norm{\partition}},\notag
    \end{align}
    where \eqref{eq:use_claim_tech} follows from 
    \cref{lemma:feasibel_sol_const_apx},
    \eqref{eq:use_packing} follows from packing property (\Cref{prop:packing}), 
    and \eqref{eq:use_size_of_holes} follows from 
    \cref{lemma:size_of_holes}.
    This completes the proof of \cref{lemma:apx_opt_x_tech}.
\end{proof}

\subsection{Proof of \Cref{lemma:apx_pi_X_by_clusters}: Lower Bounds for $\opt(\Xori)$ and $\opt(\pi(\Xori))$}
\label{subsec:cost_lb}

\lemmaapxpiXbyclusters*

We start by stating two technical lemmas.
First, we have the following lemma that upper bounds the expansion of $\opt(\pi(\Xori))$, which is used to relate the $\varepsilon \cdot \opt(\pi(\Xori))$ additive error to $O(\varepsilon)\cdot \opt(\Xori)$.
This lemma is essential for our proof of \cref{theorem:OPT_value_const_dim} (see \cref{sec:pre_value}), where we apply it directly on $\Xori$ to obtain an upper bound.
The proof of \cref{lemma:piX_leq_X_value_general} can be found in \cref{append:proof_ub_pi_X}.

\begin{restatable}[An upper bound of $\opt(\pi(\Xori))$]{lemma}{lemmapiXleqXvaluegeneral}
    \label{lemma:piX_leq_X_value_general}
    Let $\Xori\subset \RR^\oridim$ be a finite point set. 
    Let $\pi\colon\RR^\oridim\to\RR^\tardim$ 
    be a random linear map. Then for every $t > 0$, 
    \begin{equation*}
        \EE{\max\{0, \opt(\pi(\Xori))-(1+t)\opt(\Xori)\}}
        \leq \frac{1}{\tardim t} e^{-t^2 \tardim/2}\opt(\Xori).
    \end{equation*}
    Furthermore, 
    \begin{align*}
        \PR{\opt(\pi(\Xori))\geq (1+t)\opt(\Xori)}
        \leq \frac{4}{t^2\tardim} e^{-t^2\tardim/8}.
    \end{align*}
\end{restatable}

We also conclude the following lemma from \cite{MakarychevMR19} to control the contraction of each $\opt(C)$, 
which in this section is an essential ``good event'' on $C$.
In addition, we would make use of this lemma in \cref{sec:pre_value} to relate $\sum_{C \in \partition} \opt(\pi(C))$ and $\sum_{C \in \partition} \opt(C)$.
For completeness, we provide the proof of the lemma in \cref{append:proof_small_opt}.

\begin{restatable}{lemma}{lemmajlsmalloptcontractionpr}
    \label{lemma:jl_small_opt_contraction_pr}
    Let $C\subset \RR^\oridim$ be a finite point set with 
    $\opt(C)\leq \tau$.
    Let $\pi\colon\RR^\oridim\to\RR^\tardim$ be 
    a random linear map. 
    Then there exists a universal constant $c > 0$, such that
    for every $\varepsilon\in (0, 1)$, 
    if $\tardim > c\cdot \varepsilon^{-2}\log(1/\varepsilon)$,
    then 
    \begin{align*}
        \Pr\left[\opt(\pi(C))\leq \frac{1}{1+\varepsilon}\opt(C)\right]
        \leq \tau^3 \cdot e^{-\Omega(\varepsilon^2\tardim)}.
\end{align*}
\end{restatable}

Recall that our typical choice of the parameter is $\tau = (\ddim/\varepsilon)^{O(\ddim)}$.
Thus, a target dimension $m = O(\varepsilon^{-2} \ddim \log (\ddim/\varepsilon))$ suffices 
to bound the expected contraction on $C$ within $\varepsilon$, which achieves the target dimension bound in \cref{theorem:OPT_value_const_dim}.

Denote by $\cFpiopt$ the optimal solution of $\pi(\Xori)$, i.e., $\opt(\pi(X)) = \sum_{x\in \pi(X)} \dist(x, \cFpiopt) + |\cFpiopt| $.
Note that $\cFpiopt$ is a random solution whose randomness comes from $\pi$.
For every $x\in\Xori$, denote by $\Fpiopt(\pi(x))$ the closest 
facility of $\pi(x)$ in $\cFpiopt$.

Next, the proof of \Cref{lemma:apx_pi_X_by_clusters}
goes in the following key steps.
We first identify a set of ``bad points'' with respect to $\cFpiopt$ in \eqref{eqn:def_bad_points}.
Intuitively, a point $x$ is bad if the pair $(x, F^*_\pi(\pi(x)))$ violates the ``good pair'' property as in \Cref{def:good_pair}.
However, notice that here we cannot directly work with $F^*_\pi(\pi(x))$ since it is in the ambient space (and \Cref{def:good_pair} is with respect to a doubling subset).
Hence, the idea is to find a \emph{proxy} of $F^*_\pi(\pi(x))$ in $X$, via a mapping $g_\pi$ (in \eqref{eqn:proxy_mapping}).
Then we define a mapping $\psi_\pi$ (in \eqref{eq:def_moving})
that maps/moves each bad point to a nearby point which all belong to a restricted subset, which is useful to control the number of bad points in $\psi_\pi(\pi(X))$.
We show that the cost of ``moving'' points $\pi(X)$ to $\psi_\pi(\pi(X))$
is negligible compared to $\opt(\Xori)$ (\Cref{lemma:moving_dist_pi}).
Now, we argue on $\psi_\pi(\pi(X))$,
where we crucially use the randomness of $\pi$ to derive an additive distance distortion bound (utilizing the doubling dimension),
and this already suffices to bound the connection cost between every point $\pi(x) \in \psi_\pi(\pi(X))$ and $\cFpiopt$
(\Cref{lemma:feasibel_sol_tar_apx_cond}).
This step also uses the properties of $\partition$ summarized in \cref{lemma:partition_properties}.
We finally establish a lower bound $\opt(\pi(X))$ with respect to $\sum_{C\in \partition} \opt(\pi(C))$ by combining the above, to finish the proof.
We elaborate on these key steps in the following.

\paragraph{Proxy for $F_\pi^*$.}

We define a \emph{proxy} for each 
facility $f_\pi\in \cFpiopt$. 
We first find out the closest 
data point $\pi(y)$ to $f_\pi$ in its induced cluster in $\pi(X)$. 
Then, we use the original image $y$ of $\pi(y)$ to represent $f_\pi$ in the original space. 
Formally, we define the following mapping $\reprpi$.
\begin{align}
    &\reprpi\colon \cFpiopt\to \Xori; \nonumber\\
    &\reprpi(f_\pi):=\argmin_{y \in X} \Set{\Norm{\pi(y)-f_\pi}\colon 
    \Fpiopt(\pi(y))=f_\pi}.
    \label{eqn:proxy_mapping}
\end{align}
Note that $g_\pi$ is a random mapping depending on $\pi$.
Furthermore, we define the mapping $\hpi\colon \Xori\to \Xori$ 
as $\hpi(x):=\reprpi\circ \Fpiopt(\pi(x))$, which maps each point 
$x$ to the proxy in the induced cluster of $\cFpiopt(\pi(x))$.

\paragraph{Bad Points.}

Next, we define the ``bad'' points in $X$ as those do not form a 
good pair (\Cref{def:good_pair}) with $\hpi(x)$. Namely,
\begin{align}
    \label{eqn:def_bad_points}
    \badpi^\varepsilon := \Set{x \in \Xori \colon (x, \hpi(x)) \text{ is not an $\varepsilon$-good pair with respect to } (\decom, \cFa)}.
\end{align}

\begin{lemma}
    \label{lemma:bad_points_aft_proj_rare}
    Let $0<\varepsilon<1$. For every $x\in\Xori$, $\Pr_{\decom, \pi}[x\in\badpi^\varepsilon]\leq O(\varepsilon^2)$.
\end{lemma}

\begin{proof}
    By the law of total probability, 
    \begin{align*}
        \Pr_{\decom, \pi}[x\in\badpi^\varepsilon]
        &=\sum_{\pi} \PR{\pi} \cdot \Pr_\decom[x\in \badpi^\varepsilon \mid \pi] \\
        &=\sum_{\pi} \PR{\pi} \cdot \Pr_\decom[(x, \hpi(x))
        \text{ is not an $\varepsilon$-good pair}\mid \pi]
    \end{align*}
    Note that $\pi$ is independent of $\decom$. 
    For every fixed mapping $\pi\colon \RR^\oridim \to \RR^\tardim$, 
    $\hpi(x)$ is a fixed point in $\Xori$ (independent of $\decom$).
    By \cref{lemma:prob_good_pair}, 
    $\Pr_\decom[(x, \hpi(x)) \text{ is not an $\varepsilon$-good pair}\mid \pi] \leq O(\varepsilon^2)$.
    We conclude that $\Pr_{\decom, \pi}[x\in\badpi^\varepsilon] \leq O(\varepsilon^2)$.
\end{proof}

\paragraph{Mapping Bad Points to A Restricted Subset.}
We next define a mapping that maps each bad point
to a restricted subset.
To define this mapping,
fix some bad point $x \in \badpi^\varepsilon$,
and let $C \in \partition$ be the unique part such that $x \in C$.
Consider a local solution $S_C$ for cluster $C$, whose definition is the same as in \eqref{eq:def_sol_S_C} of \cref{lemma:feasibel_sol_const_apx}.
We restate its definition as follows. 
\[
    \csol_C=(\cFa\cap C) \cup \Net_C \cup 
    \left(\bigcup_{\widehat{C}\in \hole_C}\Net_{\widehat{C}}\right).
\]
For technical reasons 
which will be clear soon, we need to slightly enlarge each $\Net_{\widehat{C}}$.
Specifically, let $A > 0$ be a sufficiently large universal constant.
Each $\Net_{\widehat{C}}$ is a $\frac{\varepsilon^3}{100A^2\ddim} \cdot \rang(\widehat{C}^\decom)$-net on a larger ball 
$\Ball(\widehat{C}^\decom, \frac{40A^2 \rang(\widehat{C}^\decom)}{\varepsilon})\cap \Xori$
whose radius is increased by a factor of $O(1 / \epsilon)$ compared with the original definition in \Cref{lemma:feasibel_sol_const_apx}.

Let $\sol_C(x)$ denote the nearest neighbor of $x$ in $\csol_C$. 
We define the mapping $\mov^\varepsilon$ on $\pi(X)$ as follows,
which intuitively ``moves'' a bad point $x \in C$ to its nearest neighbor $S_C(x)$.
\begin{align}
    &\mov^\varepsilon: \pi(\Xori) \to \pi(\Xori);\notag\\
    &\mov^\varepsilon(\pi(x)):=\begin{cases}
        \pi(x), & x\notin \badpi^\varepsilon;\\
        \pi\circ \sol_C(x) \text{ where $C \in \partition$ such that $x \in C$}, & x\in \badpi^\varepsilon.
    \end{cases} \label{eq:def_moving}
\end{align}

We have the following lemma that upper bounds the total moving cost of $\mov^\varepsilon$ with respect to $\pi(X)$, whose proof
can be found in \cref{append:moving_cost}.

\begin{restatable}[Moving cost of $\mov^\varepsilon$]{lemma}{lemmamovingdistpi}
    \label{lemma:moving_dist_pi}
    Let $\Xori\subset\RR^\oridim$ be a finite point set,
    $\cFa \subseteq \Xori$ be an $\alpha$-approximate UFL solution, 
    $\pi\colon \RR^\oridim\to \RR^\tardim$ be a 
    random linear map and $0 < \epsilon < 1$.
    Define $\mov^\varepsilon$ as the mapping in \eqref{eq:def_moving}.
    There exists a constant $c>0$, such that 
    if $\tardim\geq c\cdot \log (1/\varepsilon)$, then
\begin{align*}
        \ee_{\pi, \decom}\left[
            \sum_{x\in \Xori} \Norm{\pi(x)-\mov^\varepsilon(\pi(x))}
        \right]
        \leq O(\alpha \varepsilon^2) \opt(\Xori).
    \end{align*}
\end{restatable}

We next upper bound the connection cost of $\psi^\epsilon_\pi(\pi(X))$ in \Cref{lemma:feasibel_sol_tar_apx_cond},
which involves the analysis of distances between $\pi(x) \in \psi^\epsilon_\pi(\pi(X))$ and $F^*_\pi$.
However, the useful properties such as the separation property (\Cref{lemma:partition_properties}) are only established for the original space,
and we need to show that these properties are carried on to the target space $\pi(X)$.
To this end, we make use of the following two lemmas from~\cite{IndykN07},
and in \Cref{lemma:feasibel_sol_tar_apx_cond} we condition on their success.
Although these only give an additive distance distortion bound,
it still suffices for our analysis.

\begin{restatable}[Expansion~\cite{IndykN07}]{lemma}{lemmaballexpansion}
    \label{lemma:ball_expansion}
    Let $X\subset \Ball(\mathbf{0},1)$ be a finite point set with 
    doubling dimension $\ddim$. Let $\pi\colon\RR^\oridim
    \to\RR^\tardim$ be a random linear map. Then
    there exist universal constants $c_1, c_2, T>0$, such that when 
    $\tardim >c_1 \cdot\ddim$ and $t\geq T$, 
    \begin{equation*}
        \Pr\left[\exists x\in X, \Norm{\pi(x)}\geq t\right]
        \leq e^{-c_2\cdot t^2 \tardim}.
    \end{equation*}
\end{restatable}

\begin{restatable}[Contraction~\cite{IndykN07}]{lemma}{lemmaballcontraction}
    \label{lemma:ball_contraction}
    Let $X\subset\RR^\oridim$ be a finite point set with doubling dimension
    $\ddim$. Let $\pi\colon\RR^\oridim \to\RR^\tardim$ be 
    a random linear map. Then 
    there exist universal constants $c_1, c_2, L>0$, such that 
    when $\tardim>c_1\cdot \ddim$, $\forall r>0$,
    \begin{equation*}
        \Pr\left[\exists x\in X, \Norm{x} > L\cdot r \text{ and }
        \Norm{\pi(x)} \leq r\right]
        \leq e^{-c_2\cdot \tardim}.
    \end{equation*}
\end{restatable}

Both \cref{lemma:ball_expansion,lemma:ball_contraction} are restatements of results found in~\cite{IndykN07}.
For completeness, we provide the proofs in 
\cref{append:ball_expansion,append:contraction_balls}.

\paragraph{Good Events for Distortion.}
Let $A$ be a sufficiently large universal constant.
For a ball $\Ball(z, r)$,
define the event
$\eveE(z, r, A):=\Set{\forall x\in \Ball(z, r)\cap \Xori, 
\Norm{\pi(x)-\pi(z)}\leq Ar}$, 
\sloppy
indicating that the radius of $\Ball(z, r)$ expands by at most 
$A$ times after projection. Similarly, define the event 
$\eveC(z, r, A):=\Set{\forall x\in \Xori\setminus \Ball(z, r), 
\Norm{\pi(x)-\pi(z)}\geq r/A}$, indicating that points outside
the ball will not come too close to the center after projection.
For a cluster $C\in \partition$, we define the following ``good'' 
event $\eveA_C$ with respect to $\varepsilon \in (0, 1)$, $A > 0$ and $\sigma := \frac{\varepsilon^3}{100A^2\ddim}$, which requires that
$\pi$ approximately preserves balls of certain
radii centered at net points in cluster $C$.
\begin{align}
    \eveA_C&:=\bigcap_{u \in \Net_C} \left(\eveE(u, \sigma \rang(C^\decom), A)
    \cap \eveC(u, \frac{\varepsilon^2\rang(C^\decom)}{2\ddim} , A)\right)\cap\notag\\
    &\quad \bigcap_{\widehat{C} \in \hole_C}\bigcap_{v\in \Net_{\widehat{C}}}
    \Bigg(\eveE(v, (1+\varepsilon^2) \rang(\widehat{C}^\decom), A)
    \cap \eveE(v, \sigma \rang(\widehat{C}^\decom), A)\notag\\
    &\qquad \qquad \qquad\cap \eveC(v, \frac{\varepsilon^2\rang(\widehat{C}^\decom)}{2\ddim} , A)
    \cap \eveC(v, \frac{40A^2 \rang(\widehat{C}^\decom)}{\varepsilon} , A)\Bigg)
    \label{eq:def_goodA_C}
\end{align}

We also need a lower bound for all 
$\opt(\pi(\widehat{C}))$, where $\widehat{C}\in \hole_C$.
Formally, we define the following event $\eveB_C$.
\begin{align}
    \eveB_C&:=\bigcap_{\widehat{C}\in \hole_C} \Set{\opt(\pi(\widehat{C})) 
    \geq \frac{2}{3} \opt(\widehat{C})}
    \label{eq:def_goodB_C}
\end{align}

\begin{lemma}
    \label{lemma:feasibel_sol_tar_apx_cond}
    Let $\Xori\subset \RR^\oridim$ be a finite point set 
    with doubling dimension $\ddim$
    and $\cFa \subseteq \Xori$ be an $\alpha$-approximate solution 
    on $\Xori$. 
Let $\pi\colon \RR^\oridim\to \RR^\tardim$ be a random linear map.
    For parameters $0<\varepsilon<1$, $\kappa=\Omega(1)$, 
    construct $\decom:=\decom(\Xori)$, $\xchg:=\xchg(\Xori, \decom, \cFa, \varepsilon)$ and $\partition:=\partition(\Xori, \xchg, \kappa)$
    by \cref{alg:decompose_ori}, \cref{alg:xchg} and \cref{alg:bot_up_partition}, respectively.
    Then there exists a universal constants $A > 0$, 
    such that for every cluster $C\in \partition$,
    conditioning on event $\eveA_C\cap \eveB_C$ defined in 
    \eqref{eq:def_goodA_C},\eqref{eq:def_goodB_C} with parameter $A$, 
    the following solution $\cFnew_{\pi(C)} \subseteq \RR^\tardim$ for $\pi(C)$
    \begin{align}
        \cFnew_{\pi(C)}:=\reprpi^{-1}(C) \cup \pi(\Net_C) \cup 
        \left(\bigcup_{\widehat{C}\in \hole_C}\pi(\Net_{\widehat{C}})\right)
        \cup \pi\Big(\cFa\cap C \cap \badpi^\varepsilon \Big), 
        \label{eq:def_sol_Fnew_C}
    \end{align}
    satisfies that 
    \[
    \forall x\in C, \
    \dist(\mov^\varepsilon(\pi(x)), \cFnew_{\pi(C)}) \leq (1+\varepsilon) \dist(\mov^\varepsilon(\pi(x)), \cFpiopt).
    \]
Recall that $\Net_C$ is a $\sigma\cdot \rang(C^\decom)$-net on $\Ball(C^\decom, \rang(C^\decom))\cap \Xori$, 
    and for every $\widehat{C} \in \hole_C$, 
    $\Net_{\widehat{C}}$ is a $\sigma\cdot \rang(\widehat{C}^\decom)$-net on 
    $\Ball(\widehat{C}^\decom, \frac{40A^2 \rang(\widehat{C}^\decom)}{\varepsilon})\cap \Xori$.
The scaling parameter is set to be $\sigma = \frac{\varepsilon^3}{100A^2\ddim}$.
\end{lemma}

Intuitively, \Cref{lemma:feasibel_sol_tar_apx_cond}
bounds the connection cost of $\psi_\pi^\epsilon(\pi(x))$ (for $C \in \partition$, $x \in C$) with respect to an auxiliary facility set $F'_{\pi(C)}$,
and this readily implies an upper bound for the connection cost of $\psi_\pi^\epsilon(\pi(X))$ (with respect to $F^*_\pi$).
Here, $F'_{\pi(C)}$ is picked in a similar way as the $S_C$ in \Cref{lemma:feasibel_sol_const_apx}, except that one needs to take care of the proxies (therefore the term $g_\pi^{-1}(C)$),
and that we need a subset $\pi\Big(\cFa\cap C \cap \badpi^\varepsilon \Big)$ which helps to handle bad points.

Now we are ready to prove \cref{lemma:apx_pi_X_by_clusters}.
The proof of \Cref{lemma:feasibel_sol_tar_apx_cond} is presented.

\begin{proof}[Proof of \cref{lemma:apx_pi_X_by_clusters}]
    Recall the partition $\partition = \partition(\kappa)$ relies on an $\alpha$-approximate solution $\cFa \subseteq X$, where $\alpha$ is a universal constant.
    Denote $\widehat{\varepsilon} := O(\delta \varepsilon/\alpha)$
    and $\widehat{\sigma} := \Theta(\widehat{\varepsilon}^3/\ddim)$.
    By \cref{lemma:opt_C_bounds}, denote $\tau:=2^{10\ddim}\kappa$ to be an upper bound for every $\opt(C)$ ($C \in \partition$).
    Let $\tardim=c\cdot (\log \tau + \log(1/\widehat{\varepsilon}))$, where $c$ is a sufficiently large constant.

    We start by calculating expectations. Split the left-hand side
    by \begin{align}
        &\quad \ee_{\pi, \decom}\left[
            \sum_{C\in \partition} \opt(\pi(C))\right]\notag \\
        &=\ee_{\pi, \decom}\left[
            \sum_{C\in \partition} \indicator(\eveA_C\cap \eveB_C)
            \opt(\pi(C))\right]
        +\ee_{\pi, \decom}\left[
            \sum_{C\in \partition} \indicator(\overline{\eveA_C}
            \cup \overline{\eveB_C})\opt(\pi(C))\right]
        \label{eq:spilit_OPT_pi_C}
    \end{align}

    To upper bound the first term in \eqref{eq:spilit_OPT_pi_C}, 
    we refine the point set with 
    movement mapping $\mov^{\widehat{\varepsilon}}$. First note that the moving cost is 
    \begin{align*}
        &\quad \ee_{\pi, \decom}\left[\norm{
            \sum_{C\in \partition} \indicator(\eveA_C\cap \eveB_C)
            \opt(\pi(C))
            -\sum_{C\in \partition} \indicator(\eveA_C\cap \eveB_C)
            \opt(\mov^{\widehat{\varepsilon}} \circ \pi(C))
            }\right]\\
        &\leq \ee_{\pi, \decom}\left[\sum_{C\in \partition}
            \Big|\opt(\pi(C))- \opt(\mov^{\widehat{\varepsilon}} \circ \pi(C))\Big|\right]
        \leq \ee_{\pi, \decom}\left[\sum_{C\in \partition}\sum_{x\in C}
        \Norm{\pi(x)-\mov^{\widehat{\varepsilon}}(\pi(x))}\right]\\
        &\leq O(\alpha \widehat{\varepsilon}^2) \opt(\Xori). 
        &\text{(\cref{lemma:moving_dist_pi})}
    \end{align*}

    After moving, we open $\cFnew_{\pi(C)}$ (defined in \eqref{eq:def_sol_Fnew_C}, w.r.t. $\widehat{\varepsilon}$)
    as a solution for $\mov^{\widehat{\varepsilon}} \circ \pi(C)$. 
    The expectation for the sum of optimal values for all clusters is 
    \begin{align}
        &\quad \ee_{\pi, \decom}\left[
            \sum_{C\in \partition} \indicator(\eveA_C\cap \eveB_C)
            \opt(\mov^{\widehat{\varepsilon}} \circ \pi(C))
        \right] \nonumber \\
        &\leq \ee_{\pi, \decom}\left[
            \sum_{C\in \partition} \indicator(\eveA_C\cap \eveB_C)
            \cost(\mov^{\widehat{\varepsilon}} \circ \pi(C), \cFnew_{\pi(C)})
        \right] \nonumber\\
        &= \ee_{\pi, \decom}\left[
            \sum_{C\in \partition} \indicator(\eveA_C\cap \eveB_C)
            \left(\sum_{x\in C} \dist(\mov^{\widehat{\varepsilon}}(\pi(x)), \cFnew_{\pi(C)}) 
            + \norm{\cFnew_{\pi(C)}}\right)
        \right]\nonumber \\
        &\leq \ee_{\pi, \decom}\left[
            \sum_{C\in \partition}
            \left(\sum_{x\in C} 
            (1+\widehat{\varepsilon})\dist(\mov^{\widehat{\varepsilon}}(\pi(x)), \cFpiopt) 
            + \norm{\cFnew_{\pi(C)}}\right)
        \right] & \text{(\cref{lemma:feasibel_sol_tar_apx_cond})}\nonumber\\
        &=(1+\widehat{\varepsilon}) \ee_{\pi, \decom}\left[
            \sum_{x\in \Xori}\dist(\mov^{\widehat{\varepsilon}}(\pi(x)), \cFpiopt) 
        \right]
        + \ee_{\pi, \decom}\left[
            \sum_{C\in \partition} \norm{\cFnew_{\pi(C)}}
        \right]\nonumber\\
        &\leq (1+\widehat{\varepsilon}) \ee_{\pi, \decom}\left[
            \sum_{x\in \Xori} \dist(\mov^{\widehat{\varepsilon}}(\pi(x)), \cFpiopt)
        \right]
        + \ee_{\pi, \decom}\left[
            \sum_{C\in \partition} \norm{\reprpi^{-1}(C)}
        \right]\nonumber\\
        &\quad + \ee_{\pi, \decom}\left[
            \sum_{C\in \partition} \left(\frac{1}{\widehat{\varepsilon} \widehat{\sigma}}\right)^{O(\ddim)}
            (1+\norm{\hole_C})
        \right]
        +\ee_{\pi, \decom}\left[
            \sum_{C\in \partition} 
            \Big|\cFa\cap C \cap \badpi^{\widehat{\varepsilon}}\Big|
        \right]\nonumber\\
        &\leq (1+\widehat{\varepsilon}) \ee_{\pi, \decom}\left[
            \sum_{x\in \Xori} \dist(\mov^{\widehat{\varepsilon}}(\pi(x)), \cFpiopt)
        \right]
        + \ee_{\pi, \decom}\left[
            \norm{\cFpiopt}
        \right]\nonumber\\
        &\quad + \left(\frac{1}{\widehat{\varepsilon} \widehat{\sigma}}\right)^{O(\ddim)}
        \cdot 2\ee_{\decom}\left[\norm{\partition}\right]
        +\ee_{\pi, \decom}\left[
            \Big|\cFa\cap \badpi^{\widehat{\varepsilon}}\Big|
        \right] &\text{(\cref{lemma:size_of_holes})}\nonumber\\
        &\leq (1+\widehat{\varepsilon}) \ee_{\pi, \decom}\left[\cost(\mov^{\widehat{\varepsilon}} \circ \pi(\Xori), \cFpiopt)\right]
        +\left(\frac{1}{\widehat{\varepsilon} \widehat{\sigma}}\right)^{O(\ddim)}
        \frac{4 \alpha \opt(\Xori)}{\kappa - 
        2(\ddim/\varepsilon)^{O(\ddim)}}
        +\widehat{\varepsilon}^2 \norm{\cFa} 
        &\nonumber\\
        &\quad \text{(\cref{lemma:size_of_partition,lemma:bad_points_aft_proj_rare})}\nonumber\\
        &\leq (1+\widehat{\varepsilon}) \ee_{\pi, \decom}\left[\cost(\mov^{\widehat{\varepsilon}} \circ\pi(\Xori), \cFpiopt)\right]
        +\alpha \widehat{\varepsilon}^2 \opt(\Xori)
        +\alpha \widehat{\varepsilon}^2 \opt(\Xori).\label{eqn:last}
    \end{align}
    The last inequality holds as $\widehat{\sigma} = \Theta(\widehat{\varepsilon}^3/\ddim)$ and $\kappa \geq \Omega(\ddim / \widehat{\varepsilon})^{\Omega(\ddim)}$.
    On the other hand, observe that 
    $\cost(\mov^{\widehat{\varepsilon}} \circ \pi(\Xori), \cFpiopt) - \opt(\pi(\Xori)) 
    \leq \sum_{x\in \Xori}\Norm{\pi(x)-\mov^{\widehat{\varepsilon}}(\pi(x))}$.
    Thus \eqref{eqn:last} can be further bounded by
    \begin{align*}
        &\quad (1+\widehat{\varepsilon}) \ee_{\pi}\left[\opt(\pi(\Xori))\right] 
        + (1+\widehat{\varepsilon}) \ee_{\pi, \decom}\left[\sum_{x\in \Xori}\Norm{\pi(x)-\mov^{\widehat{\varepsilon}}(\pi(x))}\right] + 2\alpha \widehat{\varepsilon}^2 \opt(\Xori)\\
        &\leq (1+\widehat{\varepsilon}) \ee_{\pi}[\opt(\pi(\Xori))]
        +O(\alpha \widehat{\varepsilon}^2) \opt(\Xori).
        &\text{(\cref{lemma:moving_dist_pi})}
    \end{align*}

    At this point, we conclude that 
    \begin{align}
        \ee_{\pi, \decom}\left[
            \sum_{C\in \partition} \indicator(\eveA_C\cap \eveB_C)
            \opt(\pi(C))\right]
        \leq (1+\widehat{\varepsilon}) \ee_{\pi}[\opt(\pi(\Xori))]
        +O(\alpha \widehat{\varepsilon}^2) \opt(\Xori).
        \label{eq:sum_OPT_C_eveA_C_eveB_C}
    \end{align}

    Now we turn to the second term in \eqref{eq:spilit_OPT_pi_C}.
    For every $C \in \partition$, 
    we first show that the event $\overline{\eveA_C}\cup \overline{\eveB_C}$
    happens with a small probability. 
Secondly, we prove that $\opt(\pi(C))$ cannot be much larger than $\opt(C)$.

    We start by bounding the probability of event 
    $\overline{\eveA_C}\cup \overline{\eveB_C}$, only using the randomness  of $\pi$ (and conditioning on 
    the randomness of $\decom$).
    By 
    \cref{lemma:ball_contraction,lemma:ball_expansion}, 
    for every $z, r$,
    the probability that
    $\eveE(z, r, A)$ or $\eveC(z, r, A)$ 
    does not happen is at most $e^{-\Omega(\tardim)}$. 
    By union bound, 
    \begin{align*}
        \Pr_{\pi}\left[\overline{\eveA_C} \mid \decom
        \right]
        \leq 4e^{-\Omega(\tardim)} \left(\norm{\Net_C}+\sum_{\widehat{C}\in \hole_C} \norm{\Net_{\widehat{C}}}\right)
        \leq 4e^{-\Omega(\tardim)}\left(\frac{1}{\widehat{\varepsilon} \widehat{\sigma}}\right)^{O(\ddim)}
        (1+\norm{\hole_C}).
    \end{align*}

    On the other hand, by \cref{lemma:jl_small_opt_contraction_pr}, 
    \begin{align*}
        \Pr_{\pi}\left[\overline{\eveB_C} \mid \decom\right]
        \leq \sum_{\widehat{C}\in \hole_C} 
        \Pr_{\pi}\left[\opt(\pi(C))\leq \frac{2}{3} \opt(C)
        \mid \decom\right]
        \leq \tau^3 \cdot e^{-\Omega(\tardim)}\cdot \norm{\hole_C}.
    \end{align*}

    For each cluster $C\in \partition$, define the event 
    $\eveD_C$ as $\eveD_C:= \Set{\opt(\pi(C))\leq 2 \cdot \opt(C)}$.
    Then by \cref{lemma:piX_leq_X_value_general}, 
    \begin{align*}
        \EE{\indicator(\overline{\eveD_C})
        \cdot (\opt(\pi(C))-2 \opt(C))
        \mid \decom}
        \leq e^{-\Omega(\tardim)} \cdot \opt(C)
        \leq \tau \cdot e^{-\Omega(\tardim)}.
    \end{align*}
    Hence, \begin{align*}
        \EE{\indicator(\overline{\eveD_C})
        \cdot \opt(\pi(C))
        \mid \decom}
        \leq 2 \opt(C)\cdot \Pr_\pi\left[
            \opt(\pi(C)) > 2 \opt(C)
        \mid \decom\right]
        +\tau \cdot e^{-\Omega(\tardim)}
        \leq O(\tau) \cdot e^{-\Omega(m)}.
    \end{align*}

    To upper bound the second term in \eqref{eq:spilit_OPT_pi_C}, 
    we split $\indicator(\overline{\eveA_C}\cup \overline{\eveB_C})
    \leq \indicator(\overline{\eveA_C}\cup \overline{\eveB_C})
    \indicator(\eveD_C)+\indicator(\overline{\eveD_C})$.
    Then \begin{align*}
        &\quad \ee_{\pi, \decom}\left[
            \sum_{C\in \partition} \indicator(\overline{\eveA_C}
            \cup \overline{\eveB_C})\opt(\pi(C))\right]\\
        &\leq \ee_{\decom}\left[\sum_{C\in \partition}\ee_{\pi}\left[
             \indicator(\overline{\eveA_C}
            \cup \overline{\eveB_C})\indicator(\eveD_C)
            \opt(\pi(C)) \mid \decom\right]\right]\\
        & \quad + \ee_{\decom}\left[\sum_{C\in \partition}
        \ee_{\pi}\left[\indicator(\overline{\eveD_C})\cdot \opt(\pi(C))
        \mid \decom\right]\right]\\
        &\leq 2 \tau \cdot \ee_{\decom}\left[\sum_{C\in \partition}
        \Pr_{\pi}\left[\overline{\eveA_C}
        \cup \overline{\eveB_C} \mid \decom\right]\right]
        +O(\tau) \cdot e^{-\Omega(\tardim)} \cdot \ee_{\decom}[\norm{\partition}]\\
        &\leq e^{-\Omega(\tardim)}\cdot O(\tau^4) 
        \cdot \left(\frac{1}{\widehat{\varepsilon} \widehat{\sigma}}\right)^{O(\ddim)}
        \ee_{\decom}\left[\sum_{C\in \partition}
        (1+2\norm{\hole_C})\right]
        +O(\tau) \cdot e^{-\Omega(\tardim)} \cdot \ee_{\decom}[\norm{\partition}]\\
        &\leq 4e^{-\Omega(\tardim)}\cdot O(\tau^4) 
        \cdot \left(\frac{1}{\widehat{\varepsilon} \widehat{\sigma}}\right)^{O(\ddim)}
        \ee_{\decom}[\norm{\partition}]\\
        &\leq 4e^{-\Omega(\tardim)}\cdot O(\tau^4) 
        \cdot \left(\frac{1}{\widehat{\varepsilon} \widehat{\sigma}}\right)^{O(\ddim)}
        \frac{2 \alpha \opt(\Xori)}{\kappa - 
        2(\ddim/\varepsilon)^{O(\ddim)}}\\
        &\leq \alpha \widehat{\varepsilon}^2 \opt(\Xori),
    \end{align*}
    given $\tardim=\Omega(\log \tau + \log(1/\widehat{\varepsilon}))$.
    Combining with \eqref{eq:sum_OPT_C_eveA_C_eveB_C}, we have 
    \begin{align*}
        \ee_{\pi, \decom}\left[
            \sum_{C\in \partition}
            \opt(\pi(C))\right]
        \leq (1+\widehat{\varepsilon}) \ee_{\pi}[\opt(\pi(\Xori))]
        +O(\alpha \widehat{\varepsilon}^2) \opt(\Xori).
\end{align*}
    Applying \cref{lemma:piX_leq_X_value_general} to $\Xori$ with
    parameter $t=0.5$ while noting $m=\omega(1)$ is sufficiently large, we have 
    \begin{align*}
        \ee_{\pi}[\opt(\pi(\Xori))]
        \leq 2\opt(\Xori).
    \end{align*}
    Thus \begin{align*}
        \ee_{\pi, \decom}\left[
            \sum_{C\in \partition}
            \opt(\pi(C))\right]
        &\leq \ee_{\pi}[\opt(\pi(\Xori))]
        +(2 \widehat{\varepsilon} + O(\alpha \widehat{\varepsilon}^2)) \opt(\Xori)\\
        & \leq \ee_{\pi}[\opt(\pi(\Xori))]
        +O(\alpha \widehat{\varepsilon}) \opt(\Xori) \\
        & \leq \ee_{\pi}[\opt(\pi(\Xori))]
        + \delta \varepsilon \opt(\Xori)
    \end{align*}

    Applying Markov's inequality to non-negative random variable 
    $\sum_{C\in \partition} \opt(\pi(C))-\opt(\pi(\Xori))$, 
    we conclude that with probability at least $1-\delta$, 
    \begin{align*}
        \sum_{C\in \partition} \opt(\pi(C))
        \leq \opt(\pi(\Xori)) + \varepsilon \opt(\Xori).
    \end{align*}
    This finishes the proof of \cref{lemma:apx_pi_X_by_clusters}.
\end{proof}

It remains to prove \cref{lemma:feasibel_sol_tar_apx_cond}.
For preparation, we first give the following locality lemma, based on the observation that every point in a cluster $\pi(C)$ should be assigned to a facility near $\pi(C)$.
Crucially, here we need to use the optimality of $F_\pi^*$.

\begin{lemma}[Locality of optimal facilities]
    \label{lemma:locality_tar_space}
    Given the conditions in \cref{lemma:feasibel_sol_tar_apx_cond},
    $\cFpiopt\cap \Ball(\pi(C\cap \Net_C), 4\Diam(\pi(C)))
    \neq \varnothing$. Furthermore, for every 
    $\widehat{C} \in \hole_C$, 
    $\cFpiopt\cap \Ball(\pi(\widehat{C} \cap \Net_{\widehat{C}}), 4\Diam(\pi(\widehat{C}))) \neq \varnothing$.
\end{lemma}

\begin{proof}[Proof of \cref{lemma:locality_tar_space}]
    Fix a cluster $\widehat{C} \in \hole_C$ 
    and a net point $z\in \Net_{\widehat{C}} \cap \widehat{C}$.
    Assume by contradiction that $\forall f_\pi\in \cFpiopt$, 
    $\Norm{f_\pi-\pi(z)}>4\Diam(\pi(\widehat{C}))$. 
    Then for every $y\in \widehat{C}$, 
    $\dist(\pi(y), \cFpiopt)\geq \dist(\pi(z), \cFpiopt)
    -\Norm{\pi(y)-\pi(z)}\geq 3\Diam(\pi(\widehat{C}))$.

    Now, open a new facility on $\pi(z)$. Denote 
    the new facility set by $\cFnew_\pi:=\cFpiopt\cup \{\pi(z)\}$.
    For every $y\in \widehat{C}$, 
    $\dist(\pi(y), \cFnew_\pi)\leq \Norm{\pi(y)-\pi(z)}\leq 
    \Diam(\pi(\widehat{C}))\leq \frac{1}{3}\dist(\pi(y), \cFpiopt)$.

    Hence, the difference in cost is 
    \begin{align*}
    \cost(\pi(X), \cFpiopt)-\cost(\pi(X), \cFnew_\pi)
    &\geq \sum_{y\in \widehat{C}} \Big(\dist(\pi(y), \cFpiopt)
    -\dist(\pi(y), \cFnew_\pi)\Big)-1\\
    &\geq 2\sum_{y\in \widehat{C}} \Norm{\pi(y)-\pi(z)}-1.
    \end{align*}
    On the other hand, by event $\eveB_C$, $\opt(\pi(\widehat{C}))
    \geq 2\kappa/3$.
    Since $\{\pi(z)\}$ is a feasible 
    solution for $\pi(\widehat{C})$, we have 
    \begin{align*}
        2\kappa/3\leq \opt(\pi(\widehat{C}))
        \leq \sum_{y\in \widehat{C}} \Norm{\pi(y)-\pi(z)} + 1.
    \end{align*}

    Therefore, $\cost(\pi(X), \cFpiopt)-\cost(\pi(X), \cFnew_\pi)
    \geq 2(2\kappa/3-1)-1=4\kappa/3 - 3 > 0$, given $\kappa > 9/4$.
    This contradicts the optimality of $\cFpiopt$,
    and finishes the proof of \Cref{lemma:locality_tar_space}.
\end{proof}

Now we are ready to prove \cref{lemma:feasibel_sol_tar_apx_cond}.

\begin{proof}[Proof of \cref{lemma:feasibel_sol_tar_apx_cond}]
We consider the following two cases.
\paragraph{Case I: $\mov^\varepsilon(\pi(x))=\pi(x)$.}
This is equivalent to $x\notin \badpi^\varepsilon$,
    so $(x, \hpi(x))$ is an $\varepsilon$-good pair.
    Now it suffices to prove $\dist(\pi(x), \cFnew_{\pi(C)})
    \leq \Norm{\pi(x)-\Fpiopt(\pi(x))}$.

    If $\Fpiopt(\pi(x))\in \reprpi^{-1}(C)$, then 
    $\dist(\pi(x), \cFnew_{\pi(C)}) \leq \dist(\pi(x), \reprpi^{-1}(C)) \leq \Norm{\pi(x)-\Fpiopt(\pi(x))}$, as desired.
If $\Fpiopt(\pi(x))\notin \reprpi^{-1}(C)$, then 
    $\hpi(x)=\reprpi\circ \Fpiopt(\pi(x)) \notin C$. 
    Assume $\widehat{x} := \hpi(x)\in \widehat{C}\neq C$.
    We further consider the following cases.

\begin{enumerate}
        \item[I.(a)] $\widehat{C}^\decom$ is not the descendant of $C^\decom$
        with respect to 
        hierarchical decomposition $\decom$.  
By the separation property of $\partition$ (\cref{lemma:partition_properties}), 
        $\Norm{x-\widehat{x}}\geq \frac{\varepsilon^2 \rang(C^\decom)}{\ddim}$.
        Since $N_C$ is a covering, there exists a net point 
        $u\in \Net_C$ such that $\Norm{x-u}\leq \sigma \rang(C^\decom)$.
        Thus, by triangle inequality,
        \begin{align*}
            \Norm{\widehat{x}-u}\geq \Norm{x-\widehat{x}}-\Norm{x-u}\geq \frac{\varepsilon^2 \rang(C^\decom)}{2\ddim}.
        \end{align*}
        Since event $\eveC(u, \frac{\varepsilon^2 \rang(C^\decom)}{2\ddim}, A)$ happens,
        $\Norm{\pi(\widehat{x})-\pi(u)}\geq \frac{\varepsilon^2 \rang(C^\decom)}{2A\ddim}$.
        Since event $\eveE(u, \sigma \rang(C^\decom), A)$ happens, 
        $\Norm{\pi(x)-\pi(u)}\leq A\sigma \rang(C^\decom)$. 
        Again by triangle inequality, 
        \begin{align*}
            \Norm{\pi(x)-\pi(\widehat{x})}\geq \Norm{\pi(\widehat{x})-\pi(u)}-\Norm{\pi(x)-\pi(u)} \geq \frac{\varepsilon^2 \rang(C^\decom)}{4A\ddim}.
        \end{align*}
        By definition of $\hpi$, both $\pi(x)$ and $\pi(\widehat{x})$ are 
        assigned to the same facility and 
        $\Norm{\pi(\widehat{x})-\Fpiopt(\pi(x))} \leq \Norm{\pi(x)-\Fpiopt(\pi(x))}$. 
        Hence, $\Norm{\pi(x)-\Fpiopt(\pi(x))}\geq\frac{1}{2}\Norm{\pi(x)-\pi(\widehat{x})} \geq \frac{\varepsilon^2 \rang(C^\decom)}{8A\ddim}$.

        Therefore, we conclude case I.(a) that $\dist(\pi(x), \cFnew_{\pi(C)}) \leq \Norm{\pi(x)-\pi(u)}\leq \Norm{\pi(x)-\Fpiopt(\pi(x))}$.

        \item[I.(b)] $\widehat{C}^\decom$ is a descendant of $C^\decom$ with respect to the  hierarchical decomposition $\decom$. 
        By the separation property of $\partition$ (\Cref{lemma:partition_properties}),
        there exists $\widetilde{C}\in \hole_C$, 
        such that $\Norm{x-\widehat{x}}\geq \frac{\varepsilon^2 \rang(\widetilde{C}^\decom)}{\ddim}$.
        Without loss of generality, we assume $\widehat{C} = \widetilde{C}$. 
        Since $\Net_{\widehat{C}}$ is a covering, there exists a net point $v\in \Net_{\widehat{C}}$ such that $\Norm{\widehat{x}-v}\leq \sigma \rang(\widehat{C}^\decom)$.
        By the same arguments, one can show that 
        $\Norm{\pi(\widehat{x})-\pi(v)}\leq A\sigma \rang(\widehat{C}^\decom)$ 
        and $\Norm{\pi(x)-\Fpiopt(\pi(x))} \geq \frac{\varepsilon^2 \rang(\widehat{C}^\decom)}{8A\ddim}$. 

        To have an upper bound for $\dist(\pi(x), \cFnew_{\pi(C)})$, we 
        have to address the following two cases separately.

        \begin{itemize}
            \item $\Norm{\pi(\widehat{x})-\Fpiopt(\pi(x))} \leq \frac{\varepsilon}{2}\Norm{\pi(x)-\Fpiopt(\pi(x))}$.
            In this case, both $\Norm{\pi(\widehat{x})-\Fpiopt(\pi(x))}$ and $\Norm{\pi(\widehat{x})-\pi(v)}$ can be upper bounded by $\frac{\varepsilon}{2}\Norm{\pi(x)-\Fpiopt(\pi(x))}$.
            By triangle inequality, 
            \begin{align*}
                \dist(\pi(x), \cFnew_{\pi(C)})
                &\leq \Norm{\pi(x)-\pi(v)} \\
                &\leq \Norm{\pi(x)-\Fpiopt(\pi(x))} + \Norm{\pi(\widehat{x})-\Fpiopt(\pi(x))} + \Norm{\pi(\widehat{x})-\pi(v)} \\
                &\leq (1+\varepsilon) \Norm{\pi(x)-\Fpiopt(\pi(x))}.
            \end{align*}

            \item $\Norm{\pi(\widehat{x})-\Fpiopt(\pi(x))} \geq \frac{\varepsilon}{2}\Norm{\pi(x)-\Fpiopt(\pi(x))}$.
            In this case, $\pi(x)$ and $\pi(\widehat{x})$ can be of similar distance from $\Fpiopt(\pi(x))$, making it impossible 
            for us to charge the additive error to $\Norm{\pi(x)-\Fpiopt(\pi(x))}$. 
            Nevertheless, we are going to argue that the facility 
            $\Fpiopt(\pi(x))$ is close to point $\pi(\widehat{x})$. 
            As a consequence, $\pi(x)$ should also be close to $\pi(\widehat{x})$,
            which means $x$ may be covered by a net point $w\in \Net_{\widehat{C}}$. 

            By \cref{lemma:locality_tar_space}, there exist
            a facility $f_\pi^0\in \cFpiopt$ and a net point 
            $z\in \widehat{C}\cap \Net_{\widehat{C}}$ such that
            $\Norm{f_\pi^0-\pi(z)}\leq 4\Diam(\pi(\widehat{C}))$.
            Then \begin{align*}
                \Norm{\pi(\widehat{x})-\Fpiopt(\pi(\widehat{x}))} 
                & \leq \Norm{\pi(\widehat{x})-f_\pi^0} 
                \leq \Norm{\pi(\widehat{x})-\pi(z)}+\Norm{f_\pi^0-\pi(z)} \\
                &\leq 5\Diam(\pi(\widehat{C}))
                \leq 5\Diam(\pi(\widehat{C}^\xchg)).
            \end{align*}
            Recall that $C^\xchg \subseteq \Ball(C^\decom, \varepsilon^2 \rang(C^\decom))$ (\cref{lemma:xchg_local}) and that 
            event $\eveE(z, (1+\varepsilon^2)\rang(\widehat{C}^\decom), A)$ happens.
            We have $\Diam(\pi(\widehat{C}^\xchg))\leq A(1+\varepsilon^2)\rang(\widehat{C}^\decom) \leq 2A \rang(\widehat{C}^\decom)$.
            Thus, $\Norm{\pi(\widehat{x})-\Fpiopt(\pi(\widehat{x}))} \leq 10A \rang(\widehat{C}^\decom)$. 
By the definition of $\cFpiopt$, we know that $\Fpiopt(\pi(x)) = \Fpiopt(\pi(\widehat{x}))$.
Recall our assumption in this case that
            $\Norm{\pi(\widehat{x})-\Fpiopt(\pi(\widehat{x}))} \geq \frac{\varepsilon}{2}\Norm{\pi(x)-\Fpiopt(\pi(\widehat{x}))}$, 
            which means $\Norm{\pi(x)-\Fpiopt(\pi(\widehat{x}))} \leq \frac{20A \rang(\widehat{C}^\decom)}{\varepsilon}$.
            Again by triangle inequality, 
            \begin{align*}
                \Norm{\pi(x)-\pi(z)} 
                &\leq \Norm{\pi(x)-\Fpiopt(\pi(\widehat{x}))} + \Norm{\pi(\widehat{x})-\Fpiopt(\pi(\widehat{x}))} + \Norm{\pi(\widehat{x})-\pi(z)} \\
                &\leq \frac{20A \rang(\widehat{C}^\decom)}{\varepsilon} + 10A \rang(\widehat{C}^\decom) + 2A \rang(\widehat{C}^\decom) \\
                &\leq \frac{40A \rang(\widehat{C}^\decom)}{\varepsilon}.
            \end{align*}

            Since the event $\eveC(z, \frac{40A^2 \rang(\widehat{C}^\decom)}{\varepsilon}, A)$ happens, $x$ should also be close to $z$ in the original space. 
            Formally, $\Norm{x-z}\leq \frac{40A^2 \rang(\widehat{C}^\decom)}{\varepsilon}$.
            Recall that $\Net_{\widehat{C}}$ is a $\sigma \rang(\widehat{C}^\decom)$-net in ball $\Ball(\widehat{C}, \frac{40A^2 \rang(\widehat{C}^\decom)}{\varepsilon})\cap \Xori$.
            Hence, $x$ is well covered by $\Net_{\widehat{C}}$. 
            There is a net point $w\in \Net_{\widehat{C}}$ such that 
            $\Norm{x-w} \leq \sigma \rang(\widehat{C}^\decom)$.
            By event $\eveE(w, \sigma \rang(\widehat{C}^\decom), A)$, 
            $\Norm{\pi(x)-\pi(w)} \leq A\sigma \rang(\widehat{C}^\decom)$.
            Therefore, 
            \begin{align*}
                \dist(\pi(x), \cFnew_{\pi(C)}) 
                &\leq \Norm{\pi(x)-\pi(w)}\leq A\sigma \rang(\widehat{C}^\decom) 
                \leq \frac{\varepsilon^2 \rang(\widehat{C}^\decom)}{8A\ddim} \\
                &\leq \Norm{\pi(x)-\Fpiopt(\pi(x))}.
            \end{align*}
        \end{itemize}
    \end{enumerate}

    \paragraph{Case II: $\mov^\varepsilon(\pi(x))=\pi\circ \sol_C(x)$.} 
    This is equivalent to $x\in \badpi^\varepsilon$.
    It suffices to prove $\dist(\pi \circ \sol_C(x), \cFnew_{\pi(C)})\leq (1+\varepsilon) \dist(\pi\circ \sol_C(x), \cFpiopt)$. 
    Consider the following cases:

    \begin{itemize}
        \item $\sol_C(x)$ is a net point, i.e. 
        $\sol_C(x)\in \Net_C\cup 
        (\bigcup_{\widehat{C} \in \hole_C}\Net_{\widehat{C}})$.
        In this case, $\dist(\pi \circ \sol_C(x), \cFnew_{\pi(C)})=0$.

        \item $\sol_C(x)$ is not a net point. By the definition 
        of $\csol_C$, we have $\sol_C(x) \in \cFa\cap C$.
        Denote $y:=\sol_C(x)$. We have the following cases:
        
        \begin{itemize}
            \item $y\in \badpi^\varepsilon$. Then $y\in \cFa \cap C\cap \badpi^\varepsilon$.
            In this case, $\dist(\pi(y), \cFnew_{\pi(C)})=0$.

            \item $y\notin \badpi^\varepsilon$.
            Then $(y, \hpi(y))$ is an $\varepsilon$-good pair.
            By case I, $\dist(\pi(y), \cFnew_{\pi(C)})\leq (1+\varepsilon) 
            \Norm{\pi(y)-\Fpiopt(\pi(y))}$.
        \end{itemize}
    \end{itemize}
Combining the above cases, we conclude that 
$\dist(\mov^\varepsilon(\pi(x)), \cFnew_{\pi(C)}) \leq (1+\varepsilon) 
\dist(\mov^\varepsilon(\pi(x)), \cFpiopt)$.
This completes the proof of \cref{lemma:feasibel_sol_tar_apx_cond}.
\end{proof}

Finally, we remark that we can prove an analogous result of \cref{lemma:apx_pi_X_by_clusters} for general metric $(\metrspa, \dist)$ and (finite) doubling subset $\Xori \subseteq \metrspa$; summarized by the following corollary.

\begin{restatable}[Lower bound for $\opt(\Xori)$]{corollary}{lemmaapxXbyclusters}
    \label{corollary:apx_X_by_clusters}
    Let $(\metrspa, \dist)$ be a metric space and $\Xori \subseteq \metrspa$ be a finite subset with doubling dimension $\ddim$.
    There exist universal constants $c_1, c_2$, such that 
    for every $\varepsilon, \delta \in (0, 1)$ and $\kappa>c_2 (\ddim/(\delta\varepsilon))^{c_1 \cdot \ddim}$, the random partition $\partition := \partition(\kappa)$ satisfies 
    \begin{align}
        \opt(\Xori)\geq \sum_{C\in\partition} \opt(C) -
        \varepsilon \cdot \opt(\Xori), 
        \label{eq:lb_opt_general}
    \end{align}
    with probability at least $1-\delta$.
\end{restatable}

\begin{proof}
    The proof is similar to that of \cref{lemma:apx_pi_X_by_clusters}.
    Simply replace $\pi$ with the identity mapping and replace $\cFpiopt$ with 
    the optimal solution on $\Xori$ in the proof.
\end{proof}

     \section{Proof of \Cref{theorem:OPT_value_const_dim}: Dimension Reduction for UFL}
\label{sec:pre_value}

\theoremOPTvalueconstdim*

\begin{proof}[Proof of \cref{theorem:OPT_value_const_dim}]
    Noting that $\tardim = \Omega(\varepsilon^{-2}\log(1/(\delta\varepsilon)))$,
    the desired upper bound of $\opt(\pi(\Xori))$, i.e. $\Pr[\opt(\pi(\Xori)) \leq (1+\varepsilon) \opt(\Xori)] \geq 1-\delta/2$ follows
    immediately from \cref{lemma:piX_leq_X_value_general}.

    Now we turn to the lower bound of $\opt(\pi(\Xori))$.
Let parameter $\kappa := c_2 (\ddim/(\delta \varepsilon))^{c_1 \cdot \ddim}$
    satisfy the condition in \cref{lemma:apx_pi_X_by_clusters}.
    Let $\partition:=\partition(\kappa)$ be the random partition constructed
    in \cref{sec:partition}.
    By \cref{lemma:opt_C_bounds}, $\kappa \leq \opt(C) \leq 2^{10\ddim} \kappa$ holds for every $C\in \partition$.
    Denote $\tau:=2^{10\ddim}\kappa$ to be an upper bound for every 
    $\opt(C)$.
    We choose 
    $\tardim=c\cdot \varepsilon^{-2}(\log \tau + \log(1/\delta \varepsilon))
    =O(\varepsilon^{-2}\ddim(\log \ddim + \log (1/\delta \varepsilon)))$, where 
    $c$ is a large enough constant.

    We start from relating each $\opt(\pi(C))$ to $\opt(C)$. 
    Conditioning on the randomness of $\decom$, 
\begin{align*}
        &\quad \ee_\pi\Big[\max\left\{0, (1-\varepsilon/3)\opt(C)-
        \opt(\pi(C))\right\} \mid \decom\Big]\\
        &\leq \opt(C)\cdot \Pr_{\pi}\Big[
            \opt(\pi(C))\leq (1-\varepsilon/3) \opt(C)
        \mid \decom\Big]\\
        &\leq \tau \cdot \Pr_{\pi}\left[
            \opt(\pi(C))\leq \frac{1}{1+\varepsilon/3} \opt(C)
        \mid \decom\right]\\
        &\leq \tau^4 \cdot e^{-\Omega(\varepsilon^2\tardim)}.
        &(\text{\cref{lemma:jl_small_opt_contraction_pr}})
    \end{align*}
    Summing over all $C\in \partition$, we have 
    \begin{align*}
        & \quad \ee_{\pi, \decom}\left[\sum_{C\in \partition}\max\left\{0, (1-\varepsilon/3)\opt(C)-
        \opt(\pi(C))\right\}\right]\\
        &\leq \tau^4 \cdot e^{-\Omega(\varepsilon^2\tardim)}
        \cdot \ee_{\decom}[\norm{\partition}]\\
        &\leq \tau^4 \cdot e^{-\Omega(\varepsilon^2\tardim)}
        \cdot \frac{2 \alpha \opt(\Xori)}{\kappa - 
        2(\ddim/\varepsilon)^{O(\ddim)}}
        &\text{(\cref{lemma:size_of_partition})}\\
        &\leq \delta\varepsilon^2/6 \cdot  \opt(\Xori).
    \end{align*}

    By Markov's inequality, with probability at least $1-\delta/2$, 
    \begin{align}
        \sum_{C\in \partition} \opt(\pi(C))
        \geq (1-\varepsilon/3)\sum_{C\in \partition} \opt(C)
        -\varepsilon^2/3 \opt(\Xori) 
        \geq (1-2\varepsilon/3) \opt(\Xori).
        \label{eq:rel_sum_of_opt}
    \end{align}

    On the other hand, by \cref{lemma:apx_pi_X_by_clusters}, 
    with probability at least $1-\delta/2$, 
    \begin{align}
        \opt(\pi(\Xori))
        \geq \sum_{C\in\partition} \opt(\pi(C))
        -\varepsilon/3 \cdot \opt(\Xori).
        \label{eq:upper_bound_sum_opt_pi_C}
    \end{align}

    Combining \eqref{eq:upper_bound_sum_opt_pi_C} and \eqref{eq:rel_sum_of_opt}, with probability at least $1-\delta$, 
    \begin{align*}
        \opt(\pi(\Xori)) \geq 
        (1-\varepsilon) \opt(\Xori),
    \end{align*}
    which completes the proof. 
\end{proof}

\begin{remark}
    \label{remark:dim_reduction_discrete}
    Recall that $\opt^S(X)$ stands for the optimal UFL value on $X$ subject to the constraint that the facilities must be taken from $S$, defined in \Cref{sec:prelim}.
Using a variant of \cref{lemma:jl_small_opt_contraction_pr}, 
    we can prove the same target-dimension bound for the \emph{discrete setting}, i.e.,
    \begin{align*}
        \PR{\opt^{\pi(\Xori)}(\pi(\Xori)) \in (1\pm \varepsilon) \opt^{\Xori}(\Xori)}
        \geq 1-\delta, 
    \end{align*}
    which directly improves over the $O(1)$-approximate of~\cite{NarayananSIZ21}.
\end{remark}

     \section{Proof of \cref{theorem:main_ptas}: PTAS for UFL on Doubling Subsets}
\label{sec:ptas}

\theoremmainptas*

In this section, we provide a PTAS for the UFL problem on a doubling subset $\Xori\subset \mathbb{R}^d$ via dimension reduction (\cref{alg:ptas}). 
The idea is to apply the metric decomposition and dimension reduction approaches in \Cref{sec:partition,sec:pre_value}, in which we construct a partition $\partition$ on $\Xori$ and solve the UFL problem for each projected cluster $\pi(C)$ with $C\in \partition$. 
Then we prove the correctness of 
our algorithm in \cref{subsec:ptas_correctness} and  
analyze the time complexity in \cref{subsec:ptas_time_complexity} respectively.
Since we always consider the doubling dimension of $\Xori$, 
we denote $\ddim := \ddim(\Xori)$ for short.

\subsection{The PTAS}
\label{subsec:ptas_description}

\begin{algorithm}[!ht]
\caption{\ptas\ for UFL on doubling subsets}
\label{alg:ptas}
\DontPrintSemicolon
\KwIn{finite point set 
    $\Xori\subset \RR^\oridim$ with doubling dimension
    $\ddim$, parameter $\varepsilon\in (0, 1)$,  $\alpha$-approximate UFL algorithm $\algapx$, and  $k$-median algorithm
    $\algmedian$}
let $\cFret, \cC\gets \varnothing$ \;
let $c_1, c_2, c_3, c_4 > 0$ be sufficiently large constants,
    $\kappa \leftarrow c_2(\ddim/\varepsilon)^{c_1\cdot \ddim}$, 
    $\tau \leftarrow 2^{10 \ddim} \cdot \alpha \kappa$, and 
    $\tardim\leftarrow c_3\cdot \varepsilon^{-2} 
    (\log \tau + \log (1/\varepsilon))$ \;
run \cref{alg:decompose_ori} on $\Xori$ to obtain a random hierarchical decomposition $\decom$ \label{line:decom} \;
run $\algapx$ on $\Xori$ to obtain an $\alpha$-approximate solution $\cFa\subseteq \Xori$ for UFL \label{line:use_approx_1} \;
run \cref{alg:xchg} on $\Xori$ to compute the modified decomposition $\xchg = \xchg(\Xori, \decom, \cFa, \varepsilon)$ \label{line:xchg} \;
run $\algPartition(\Xori, \xchg, \kappa)$ (\cref{alg:bot_up_partition}) to obtain a partition $\partition$ of $X$ \label{line:partition} \;
construct $\pi(X)$ by a random linear map $\pi: \mathbb{R}^d \rightarrow \mathbb{R}^m$ \label{line:random_proj} \;
\For{$C\in\partition$}{ \label{line:a_cluster_start}
        $(\cFapx_C, \cC_C) \gets \algapx(C)$
        \label{line:use_approx_2} \;
        \tcc{
        $\cFapx_C$ is the solution and $\cC_C$ is the corresponding clustering}
\If{$\exists f, f^\prime \in \cFapx_C$, such that $\Norm{f-f^\prime}
        > (1+\varepsilon) \Norm{\pi(f)-\pi(f^\prime)}$}{ \label{line:eveG}
            \tcc{some distance in $\cFapx_C$ contracts too much}
            $\cC\gets \cC\cup \cC_C$ \;
            \tcc{directly use the constant-approximate clustering $\cC_C$}
        } 
\Else{
            \For{$k=1, 2, \dots, \lfloor c_4 \tau\rfloor$}{
                $(\pi(X_{C,1}^k), \ldots, \pi(X_{C,k}^k), v_C^k) 
                \gets \algmedian_{k, m}(\pi(C), \varepsilon)$
                \label{line:use_median}\;
                \tcc{$\pi(X_{C,i}^k)$'s is the clustering and $v_C^k$ is the cost}
}
$k^*\gets\argmin_{k}\{k+v_C^k\}$ \label{line:select_k} \;
\If{$k^*+v_C^{k^*} > c_4 \tau$}{\label{line:eveH}
                \tcc{$\opt(\pi(C))$ expands too much compared with $\opt(C)$} 
                $\cC\gets \cC\cup \cC_C$ \;
                \tcc{directly use the constant-approximate clustering $\cC_C$}
            } 
\Else{
                $\cC\gets \cC\cup\{X_{C,1}^{k^*}, X_{C,2}^{k^*},\dots, X_{C,k^*}^{k^*}\}$ \label{line:use_k_clustering} \;
                \tcc{use the clustering computed by $\algmedian$}
            }
        }
    } \label{line:a_cluster_end}
\For{$X_i\in \cC$}{\label{line:start_1_med}
        compute a $(1+\varepsilon)$-approximate 1-median center $f_i$ on 
        $X_i$  using \Cref{lemma:CLMPS16} \label{line:compute_1_median}\;
$\cFret\gets \cFret\cup\{f_i\}$ \label{line:add_center} \;
    } 
\Return $\cFret$ \label{line:ret} \;
\end{algorithm}

The PTAS is presented in \cref{alg:ptas}. 
It makes use of two subroutines: an $\alpha$-approximate UFL algorithm $\algapx$ with $\alpha = O(1)$ (\cref{line:use_approx_1,line:use_approx_2}) and a $k$-median algorithm $\algmedian$ (\cref{line:use_median}).
The algorithm $\algapx$ takes a point set $P\subset \RR^d$ as input, and outputs an $\alpha$-approximate solution $F\subseteq P$ for the UFL problem with $\cost(P,F) \leq \alpha \cdot \opt(P)$
and a clustering $\cC$ w.r.t. $F$ that consists of a partition $P_1, \cdots, P_{|F|}$ where $P_i$ contains the points in $P$ whose closest facility in $F$ is $f_i\in F$ (breaking ties arbitrarily).
The algorithm $\algmedian$ takes integers $k, d > 0$, parameter $\varepsilon\in (0, 1)$ and point set $P\subset \RR^d$ as input, $\algmedian_{k, d}(P, \varepsilon)$ returns a $(1+\varepsilon)$-approximate solution $F\subset \RR^d$ for the $k$-median problem with $\sum_{p\in P} \dist(p, F) \leq (1+\varepsilon)\cdot \min_{F'\subset \RR^d} \sum_{p\in P} \dist(p, F')$, a clustering $P_1, \cdots, P_{|F|}$ of $P$ w.r.t. $F$, and a value $v = \sum_{p\in P} \dist(p, F)$.
We would apply $\algmedian$ on projected clusters $\pi(C)$ in the target space $\RR^m$.
These two subroutines are stated in \Cref{lemma:ptas_const_apx,lemma:k_median_oracle},
and they are obtained easily from combining existing algorithms.
We give an analysis for the running time of \cref{alg:ptas},
particularly the dependence on the running time of the two subroutines, in \Cref{lemma:ptas_time_complexity}.

Roughly speaking, the algorithm runs in the following three stages.

\paragraph{Stage 1: Constructing Partition $\partition$ (Lines \ref{line:decom}-\ref{line:random_proj}).}
This stage is a pre-processing stage.
In Lines \ref{line:decom}-\ref{line:partition}, we construct a partition $\partition$ of $X$ by \cref{alg:bot_up_partition}.
Note that \cref{line:critical_condition} of \cref{alg:bot_up_partition} computes $\opt(C)$, which is inefficient.
This step can be replaced by running $\algapx$ on $C$ and checking whether the resulting UFL cost $\cost(C, F_C)\geq \alpha\kappa$, which can ensure that $\kappa\leq \opt(C)\leq \tau$ for each $C\in \partition$.
In \cref{line:random_proj}, we apply a random linear map to construct $\pi(X)\subset \RR^m$.

\paragraph{Stage 2: Constructing Near-optimal Clustering $\cC$ (Lines \ref{line:a_cluster_start}-\ref{line:a_cluster_end}).}
At this stage, we compute a clustering that is near-optimal to $\opt(\pi(C))$ for each $C\in \partition$, and take their union $\cC$ (which is a partition of $X$).
Intuitively, the near-optimal clustering for each $\pi(C)$ can be efficiently constructed
since $\opt(C)\leq \tau$ and the $\opt(\pi(C))$ is within $(1\pm \varepsilon)\cdot \opt(C)$ with high probability,
so that the $k$-median algorithm $\algmedian$ can be applied with only $k \leq O(\tau)$.
Conditions in \cref{line:eveG,line:eveH} examine if the ``bad event'' happens, and if so then directly add a constant-approximate clustering $\cC_C$ for $C$. 
The design of conditions in \cref{line:eveG,line:eveH} is based on \cref{lemma:ptas_exp_single_cluster}.

\paragraph{Stage 3: Constructing Open Facilities $F$ (Lines \ref{line:start_1_med}-\ref{line:ret}).}
At this stage, we already have a clustering $\cC=\{X_1, X_2, \dots, X_{\norm{\cFret}}\}$ of $\Xori$. 
For each cluster $X_i$, we solve the $1$-median algorithm by~\cite{cohenLMPS16}, obtain a center $f_i$ (\cref{line:compute_1_median,line:add_center}), and output their unions $F$ as the solution of the PTAS. 
The $1$-median algorithm is summarized by the following lemma. 

\begin{lemma}[$1$-median approximation~\cite{cohenLMPS16}]
    \label{lemma:CLMPS16}
    There is an algorithm that takes as input $X\subset \RR^d$ 
    of size $n$ and parameter $\varepsilon \in (0, 1)$,  
    and outputs a $(1+\varepsilon)$-approximate $1$-median center
    with probability $1-1/\poly(n)$, running in time $O(nd\log^4(n/\varepsilon))$.
\end{lemma}

We provide concrete realizations of our algorithms $\algmedian$ and $\algapx$ in the following lemmas.

\begin{lemma}[Constant approximate algorithm]
    \label{lemma:ptas_const_apx}
    There exist a universal constant $\alpha > 0$ and an algorithm
    $\algapx$, which takes as input $\Xori \subset \RR^\oridim$
    of size $n$ and doubling dimension $\ddim$,
    and computes an $\alpha$-approximate UFL solution $F$ and its corresponding clustering $\cC$ with probability $1-1/\poly(n)$, 
    running in time $T_{A}(n, d, \ddim) = \tilde{O}(n\cdot d\cdot 2^{O(\ddim)})$.
\end{lemma}

\begin{proof}
    The constant approximation algorithm is a combination of 
    known results. On the one hand, it is shown in~\cite{GoelIV01} that metric facility
    location can be reduced to nearest neighbor search.
    Specifically, there is an algorithm which takes as input $X\subseteq \RR^\oridim$,
    and outputs a $4(1+O(\varepsilon))$-approximation
    of metric FL, using $\tilde{O}(n)$ queries to 
    a $(1+\varepsilon)$-approximate nearest neighbor oracle~\cite{GoelIV01}.

    On the other hand, it is shown in~\cite{Har-PeledM06,Har-PeledK13} that
    there is an algorithm which takes as input $X\subset \RR^\oridim$ 
    of $n$ points, and build a data structure in 
    $O(n\cdot\oridim \varepsilon^{-O(\ddim)} \log n)$ expected time, 
    such that given a point $x\in X$, one can return 
    a $(1+\varepsilon)$-approximate nearest neighbor of $x$ 
    in $X$. The query time is 
    $O(2^{O(\ddim)} d \log n+\varepsilon^{-O(\ddim)} d)$.

    The algorithm $\algapx$ therefore 
    works as follows:
    it first computes a $2$-approximate 
    nearest neighbor data structure in expected time $\tilde{O}(2^{O(\ddim)} \cdot n \cdot \oridim)$.
    Then reduce the UFL problem to $\tilde{O}(n)$ queries to the 
    ANN oracle. Since each query can be answered in time 
    $\tilde{O}(1)$, the total time complexity is 
    $\tilde{O}(n\cdot \oridim \cdot 2^{O(\ddim)})$ in expectation.
    It is clear that $\algapx$ returns an $\alpha = 4(1+O(1))$-approximation.

    By standard boosting techniques, $\algapx$ can be modified so that it runs in deterministic time $\tilde{O}(n\cdot d\cdot 2^{O(\ddim)})$ with failure probability $1/\poly(n)$, as desired.
\end{proof}

\begin{lemma}[$k$-Median algorithm]
    \label{lemma:k_median_oracle}
    There is an algorithm $\algmedian$, such that for arbitrary
    integers $n, d, k > 1$,
    $\algmedian$ takes as input $\Xori\subset \RR^\oridim$ with 
    $n$ points and a parameter $\varepsilon \in (0, 1)$, and 
    outputs a $(1+\varepsilon)$-approximate $k$-median 
    clustering as well as the corresponding cost with probability $1-1/\poly(n)$, 
    running in time $T_{M}(n, d, k, \varepsilon)
    =\tilde{O}(n\oridim k + \varepsilon^{-6} \oridim k^{O(k/\varepsilon^3)})$.
\end{lemma}

\begin{proof}
    The algorithm depends on the coreset construction. 
    It is shown in~\cite{Cohen-AddadLSSS22} that 
    there is an algorithm that takes as input $X\subset \RR^d$,
    integer $k > 0$ and precision
    parameter $\varepsilon\in (0, 1)$, and outputs 
    a $(1+\varepsilon)$-coreset for $k$-median of size 
    $\tilde{O}(\min\{k^{4/3}\cdot \varepsilon^{-2}, 
    k\cdot \varepsilon^{-3}\})$.

    The $k$-median algorithm $\algmedian$ first constructs
    a coreset $X^\prime \subseteq \Xori$ of size 
    $O(k/\varepsilon^3)$, using~\cite{Cohen-AddadLSSS22} algorithm.
    Then it enumerates all possible $k$-partitions of $X^\prime$. 
    For each partition, run the algorithm in \cref{lemma:CLMPS16} to compute the (approximate) $1$-median center for each cluster and sum up the connection cost. 
    Among all these center sets, it selects the one $S$ with the minimum $\cost(X^\prime, S)$.
    In the end, return the partition of $X$ induced by $S$, together with $\cost(X, S)$.

    The coreset construction algorithm~\cite{Cohen-AddadLSSS22} has time complexity $\tilde{O}(ndk)$ and success probability $1-1/\poly(n)$.
    For each one of the $k^{\norm{X^\prime}}$ possible partitions, 
    computing the $k$-median cost involves $k$ calls to the $1$-median algorithm. 
    By \Cref{lemma:CLMPS16}, each call takes time 
    $\tilde{O}(\norm{X^\prime}^2 \cdot \oridim)$ to ensure a $1-1/(k^{|X'|+1} \poly(n))$ success probability.
    Hence, the time complexity of the enumeration is 
    $\tilde{O}(\norm{X^\prime}^2 \cdot \oridim \cdot k^{|X^\prime|}) = \tilde{O}(\varepsilon^{-6} \oridim k^{O(k/\varepsilon^3)})$.
    Finally, computing the induced partition on $X$ takes time $O(ndk)$.
    Moreover, $\algmedian$ succeeds with probability $1-1/\poly(n)$.
    This completes the proof.
\end{proof}

\subsection{Correctness of \cref{alg:ptas}}
\label{subsec:ptas_correctness}

Observe that \cref{line:use_k_clustering} of \Cref{alg:ptas} is executed for some $C$ only when the following events hold.
\begin{align*}
    \eveG_C &:=\Set{\forall f, f^\prime \in \cFapx_C, \Norm{f-f^\prime} 
    \leq (1+\varepsilon) \Norm{\pi(f)-\pi(f^\prime)} \text{ (\cref{line:eveG})}}\\
    \eveH_C&:=\Set{k^* + v_C^{k^*} \leq c_4 \tau \text{ (\cref{line:eveH})}}.
\end{align*}
We also define the following event 
\begin{align*}
    \eveK_C := \Set{\opt(\pi(C)) \leq c_4 \tau}.
\end{align*}
For the sake of presentation, we write the final clustering computed at the end of \cref{line:a_cluster_end} as $\cC=\{X_1, X_2, \dots, X_{\norm{\cFret}}\}$. 
For $C\in\partition$, denote 
by $\U_C:=\{i\in [\norm{\cFret}]\colon X_i \subseteq C\}$ the set of clusters
$X_i$ that are subsets of $C$. 
The following lemma shows that conditioning on event $\eveG_C\cap \eveK_C$,
$\sum_{i\in U_C}\med_1(X_i)$ and $\sum_{i\in U_C}\med_1(\pi(X_i))$ are sufficiently close. 
The proof of \Cref{lemma:ptas_exp_single_cluster} can be found in \Cref{append:ptas_exp_single_cluster}.
We note that \cref{lemma:ptas_exp_single_cluster} may be viewed as a variant of~\cite[Theorem 3.6]{MakarychevMR19},
where our guarantee is with respect to the expectation (and theirs is about probabilities).

\begin{restatable}[Cost preserving for partition $U_C$]{lemma}{lemmaptasexpsinglecluster}
    \label{lemma:ptas_exp_single_cluster}
    For every cluster $C\in \partition$,
    \begin{align*}
        \ee_\pi \left[\indicator(\eveG_C\cap\eveK_C)
        \cdot \sum_{i\in\U_C}
        \max\Set{0, \med_1(X_i)-(1+\varepsilon)\med_1(\pi(X_i))}
        \mid \decom \right] \leq \varepsilon^2.
    \end{align*}
\end{restatable}

Now we are ready to prove the correctness of \cref{alg:ptas}.

\begin{proof}[Proof of \cref{theorem:main_ptas} (correctness)]
    \cref{alg:ptas} makes $O(\tau \cdot \norm{\partition}) = O(\tau n)$ calls to $\algmedian$,
$\sum_i \norm{\xchg_i} + \norm{\partition} = O(n)$ calls to $\algapx$,
    and $O(n)$ calls to the $1$-median algorithm in \Cref{lemma:CLMPS16}.
    By~\Cref{lemma:CLMPS16,lemma:ptas_const_apx,lemma:k_median_oracle}, each of these calls succeeds with probability $1-1/\poly(n)$.
    Hence, with probability $1-1/\poly(n)$, all the calls succeed simultaneously.
    Conditioned on this event, we derive the following analysis.

    We start from analyzing the cost of clustering $\{X_i\}_{i\in U_C}$.
    \begin{align}
        &\quad \sum_{i=1}^{\norm{\cFret}} \med_1(X_i) + \norm{\cFret}
        =\sum_{C\in \partition} \left(\sum_{i\in \U_C} \med_1(X_i)
        +\norm{\U_C}\right)\notag\\
        &=\sum_{C\in \partition} \indicator(\eveG_C\cap \eveH_C)
        \left(\sum_{i\in \U_C} \med_1(X_i)
        +\norm{\U_C}\right)
        +\sum_{C\in \partition} \indicator(\overline{\eveG_C}\cup \overline{\eveH_C})
        \left(\sum_{i\in \U_C} \med_1(X_i)
        +\norm{\U_C}\right).\label{eq:ptas_two_terms}
    \end{align}
Next, we bound the two terms in~\eqref{eq:ptas_two_terms} respectively.
For the first term, we first claim that $\eveH_C \subseteq \eveK_C$. 
Indeed, when event $\eveH_C$ happens, we have $\opt(\pi(C)) \leq k^* + \med_{k^*}(\pi(C)) \leq k^* + v_C^{k^*}\leq c_4\tau$, implying event $\eveK_C$.
Thus $\indicator(\eveG_C\cap\eveH_C) \leq \indicator(\eveG_C\cap\eveK_C)$.
    Summing the result in \cref{lemma:ptas_exp_single_cluster}
    over $C\in \partition$, we have
    \begin{align*}
        &\quad \ee_{\pi, \decom}\left[\sum_{C\in \partition}\sum_{i\in \U_C} 
        \indicator(\eveG_C\cap\eveH_C)
        \cdot \max\Set{0, \med_1(X_i)-(1+\varepsilon)\cdot \med_1(\pi(X_i))}
        \right]\\
        &\leq \varepsilon^2 
        \cdot \ee_{\decom}[\norm{\partition}]
        \leq \varepsilon^2 
        \cdot \frac{2 \alpha \opt(\Xori)}{\kappa - 
        2(\ddim/\varepsilon)^{O(\ddim)}}
        \leq O(\varepsilon^2) \cdot \opt(\Xori). 
    \end{align*}
    By Markov's inequality, 
    with probability at least $1-O(\varepsilon)$, 
    $\sum_{C\in\partition} \sum_{i\in\U_C}
    \indicator(\eveG_C\cap\eveH_C)
    \cdot \max\Set{0, \med_1(X_i)-(1+\varepsilon)\med_1(\pi(X_i))}
    \leq \varepsilon \cdot \opt(\Xori)$. 
Hence 
    \begin{align}
        &\quad \sum_{C\in \partition} \indicator(\eveG_C\cap \eveH_C)
        \left(\sum_{i\in \U_C} \med_1(X_i)
        +\norm{\U_C}\right)\notag \\
        &\leq (1+\varepsilon) 
        \sum_{C\in \partition} \indicator(\eveG_C\cap \eveH_C)
        \left(\sum_{i\in \U_C} \med_1(\pi(X_i))
        +\norm{\U_C}\right)
        +\varepsilon \cdot \opt(\Xori).\label{eq:ptas_first_bound}
    \end{align}
    Fix a cluster $C\in \partition$. 
    Recall that $\opt(\pi(C)) \leq c_4\tau$ conditioning on $\eveH_C$.
Hence there exists 
    $k\in \{1, 2, \dots, \lfloor c_4\tau\rfloor\}$, such that 
    $\opt(\pi(C))=k+\med_k(\pi(C))$. 
On the other hand, when $\indicator(\eveG_C\cap \eveH_C) = 1$, $U_C$ is constructed by $\algmedian$, which returns a
    $(1+\varepsilon)$-approximate $k$-median solution. 
Thus, we have
    \begin{align*}
        &\quad \indicator(\eveG_C\cap \eveH_C)
        \left(\sum_{i\in \U_C} \med_1(\pi(X_i))
        +\norm{\U_C}\right)\\
        & \leq k^* + v_C^{k^*}
        \leq k + v_C^k
        \leq (1+\varepsilon) (k + \med_k(\pi(C)))
        = (1+\varepsilon) \opt(\pi(C)).
    \end{align*}
    
    Therefore, the right side of \eqref{eq:ptas_first_bound} can be further 
    upper bounded by 
    $
        (1+\varepsilon)^2 \sum_{C\in \partition} \opt(\pi(C))
        +\varepsilon \opt(\Xori).
    $

    By \cref{lemma:apx_pi_X_by_clusters}, 
    with probability at least $1-O(\varepsilon)$, 
    $\sum_{C\in \partition} \opt(\pi(C))
    \leq \opt(\pi(\Xori))+\varepsilon \opt(\Xori)$.
    By \cref{lemma:piX_leq_X_value_general} 
    and Markov's inequality, 
    with probability at least $1-\frac{4}{\varepsilon^2 \tardim}e^{-\varepsilon^2 \tardim/2}=1-O(\varepsilon)$, 
    $\opt(\pi(\Xori))\leq (1+\varepsilon) \opt(\Xori)$.
    Therefore, with probability at least $1-O(\varepsilon)$, 
    \begin{align}
        \sum_{C\in \partition} \indicator(\eveG_C\cap \eveH_C)
        \left(\sum_{i\in \U_C} \med_1(X_i)
        +\norm{\U_C}\right)
        &\leq (1+\varepsilon)^2(1+2\varepsilon) \opt(\Xori) + \varepsilon \opt(\Xori) \notag\\
        &\leq (1+O(\varepsilon)) \opt(\Xori).\label{eq:ptas_ideal_bound}
    \end{align}

    We turn to the second term on the right side of \eqref{eq:ptas_two_terms}.
If either $\eveG_C$ or $\eveH_C$ does not happen, 
    the clustering $\{X_i\}_{i\in \U_C}$ is an $\alpha$-approximate
    solution on $C$. Thus 
    \begin{align}
        &\quad \ee_{\pi, \decom}\left[
            \sum_{C\in \partition} 
            \indicator(\overline{\eveG_C}\cup \overline{\eveH_C})
            \left(\sum_{i\in \U_C} \med_1(X_i)
            +\norm{\U_C}\right)
        \right]\notag\\
        &\leq \alpha \ee_{\pi, \decom}\left[
            \sum_{C\in \partition} 
            \indicator(\overline{\eveG_C}\cup \overline{\eveH_C})
            \opt(C)
        \right]
        \leq \alpha\tau \cdot \ee_{\decom}\left[
            \sum_{C\in \partition} \Pr_{\pi} 
            \left[\overline{\eveG_C}\cup \overline{\eveH_C} \mid \decom\right]
        \right]\label{eq:ptas_bound_2}.
    \end{align}
    Recall that $\cFapx_C$ is an $\alpha$-approximate solution
    on $C$. Hence $\norm{\cFapx_C}\leq \alpha \opt(C)\leq \alpha \tau$.
    Thus \begin{align*}
        \Pr_{\pi}\left[\overline{\eveG_C} \mid \decom\right] 
        \leq \binom{\alpha\tau}{2}\cdot e^{-\Omega(\varepsilon^2 \tardim)}
        \leq O(\tau^2) e^{-\Omega(\varepsilon^2 \tardim)}.
    \end{align*} 
    On the other hand, if $\eveH_C$ does not happen, then for every $k \in \{1, 2, \dots, \lfloor c_4 \tau \rfloor\}$, $k + v_C^k > c_4 \tau$.
    Since $v_C^k \leq (1+\varepsilon) \med_k(\pi(C))$, we have $k + \med_k(\pi(C)) \geq c_4 \tau / (1+\varepsilon) > c_4 \tau / 2$. 
    This implies that $\opt(\pi(C)) \geq c_4 \tau / 2$.
    By \cref{lemma:piX_leq_X_value_general}, when constant $c_4$ is sufficiently large, 
    \begin{align*}
        \Pr_{\pi}\left[\overline{\eveH_C} \mid \decom\right] 
        \leq \Pr_{\pi}\left[\opt(\pi(C))\geq c_4 \tau/2
        \mid \decom\right]
        \leq e^{-\Omega(\tardim)}.
    \end{align*}
    
    Therefore, \eqref{eq:ptas_bound_2} can be further upper bounded
    by 
    \begin{align*}
        &\quad \ee_{\pi, \decom}\left[
            \sum_{C\in \partition} 
            \indicator(\overline{\eveG_C}\cup \overline{\eveH_C})
            \left(\sum_{i\in \U_C} \med_1(X_i)
            +\norm{\U_C}\right)
        \right]\\
        &\leq \alpha\tau \cdot e^{-\Omega(\varepsilon^2 \tardim)} O(\tau^2)
        \cdot \ee_{\decom}\left[
            \norm{\partition}
        \right]\\
        &\leq O(\tau^3) \cdot e^{-\Omega(\varepsilon^2 \tardim)} 
        \cdot \frac{2 \alpha \opt(\Xori)}{\kappa - 
        2(\ddim/\varepsilon)^{O(\ddim)}}\\
        &\leq \varepsilon^2 \cdot \opt(\Xori).
    \end{align*}
    By Markov's inequality, with probability at least 
    $1-O(\varepsilon)$, 
    \begin{align}
        \sum_{C\in \partition} 
        \indicator(\overline{\eveG_C}\cup \overline{\eveH_C})
        \left(\sum_{i\in \U_C} \med_1(X_i)
        +\norm{\U_C}\right)
        \leq \varepsilon \cdot \opt(\Xori).\label{eq:ptas_error_bound}
    \end{align}

    Combining \eqref{eq:ptas_ideal_bound} with \eqref{eq:ptas_error_bound},
    we conclude that with constant probability,
    \begin{align*}
        \norm{\cFret}+\sum_{i=1}^{\norm{\cFret}} \med_1(X_i)
        \leq (1+O(\varepsilon)) \cdot \opt(\Xori)
    \end{align*}
    Recall that each $f\in \cFret$ is a $(1+\varepsilon)$-approximate
    $1$-median center. Thus 
    \begin{align*}
        \cost(\Xori, \cFret)
        \leq (1+\varepsilon) \left(
            \norm{\cFret}+\sum_{i=1}^{\norm{\cFret}} \med_1(X_i)\right)
        \leq (1+O(\varepsilon)) \opt(\Xori).
    \end{align*}
    This completes the proof.
\end{proof}

\subsection{Time Complexity of \cref{alg:ptas}}
\label{subsec:ptas_time_complexity}
In this section, we analyze the time complexity of \cref{alg:ptas}.
Formally, we prove the following lemma.

\begin{lemma}[Time complexity of \cref{alg:ptas}]
    \label{lemma:ptas_time_complexity}
    Given a $k$-median approximation algorithm $\algmedian$ with
    running time $T_{M}(n, d, k, \varepsilon)$,
    and an $\alpha$-approximate UFL algorithm $\algapx$ with running time $T_{A}(n, d, \ddim)$
    for any input size $n$, ambient dimension $d$, doubling dimension $\ddim$, integer $k$ and precision parameter $\varepsilon$,
    \cref{alg:ptas} takes as input $\Xori \subset \RR^\oridim$ 
    with $n$ points and $\varepsilon\in (0, 1)$, and makes $r$ 
    calls to algorithm $\algmedian$ and $s$ calls to $\algapx$, 
    running in time
\begin{align*}
        \tilde{O}\Bigg(
            \tau^2 \oridim \cdot n
             + \max_{\substack{n_1,\ldots, n_r \geq 0: \\ \sum_{i=1}^r n_i = r}} \sum_{i=1}^r T_{M}(n_i, m, \tau, \varepsilon)
            + \max_{\substack{n_1,\ldots, n_s \geq 0: \\ \sum_{i=1}^s n_i = n \log n}} \sum_{i=1}^s T_A(n_i, d, \ddim(\Xori))
        \Bigg),
    \end{align*}
    where $\tau = O((\ddim(\Xori)/\varepsilon)^{O(\ddim(\Xori))})$, 
    $r=O(\tau n)$, $s = O(n)$ and
    $m=O(\varepsilon^{-2}\cdot \log \tau)$.
\end{lemma}

Since we always focus on the doubling dimension of $\Xori$, we denote $\ddim := \ddim(\Xori)$ for short in our following analysis.
In addition, we assume the aspect ratio of $X$ to be $\Delta = \poly(n)$. (See e.g. \cite[Lemmas A.1, A.2]{Cohen-AddadFS21} for details.)

\begin{proof}[Proof of \cref{lemma:ptas_time_complexity}]
By~\cite[Lemma 9]{Cohen-AddadFS21}, the hierarchical decomposition $\decom$ can be constructed within time $O(n \cdot \oridim \cdot 2^{O(\ddim)} \log \Delta)$ in \cref{line:decom}.

    To construct the $\alpha$-approximate solution $\cFa$, $\algapx$ is run on point set $\Xori$, which has time complexity $T_{A}(n, d, \ddim)$ (\cref{line:use_approx_1}).
To construct $\xchg$ (\cref{line:xchg}), one need to compute the distance 
    $\Norm{x-\Fa(x)}$ for all $x \in \Xori$ and update 
    all the level $i$ clusters correspondingly for each $i \in [\ell]$, 
    which has time complexity $O(n \cdot \oridim + n \cdot\log\Delta)$. 
    Hence, the modified decomposition $\xchg$ can be constructed in time 
    $\tilde{O}(n \cdot d + T_{A}(n, d, \ddim))$.

    The time complexity of constructing $\partition$ contains
    two stages -- checking and updating. 
    In the checking stage, the algorithm scans every cluster $C\in \bigcup_i \xchg_i$ and apply $\algapx$ to estimate $\opt(C)$ (\cref{alg:bot_up_partition}, \cref{line:critical_condition}). 
    Since each cluster is checked at most once, the checking stage has time complexity $\sum_i \sum_{C \in \xchg_i} T_{A}(\norm{C}, \oridim, \ddim)$.
In the updating stage, our algorithm deletes the points
    from the dataset (\cref{alg:bot_up_partition}, \cref{line:update_clusters}). 
    Since each point $x$ is deleted at most once, and there are at most 
    $\log \Delta$ clusters in $\xchg$ containing $x$, 
    the updating time complexity is $\tilde{O}(n)$.
    In conclusion, the overall time complexity of constructing 
    $\partition$ is $\tilde{O}(n + \sum_i \sum_{C \in \xchg_i} T_{A}(\norm{C}, \oridim, \ddim))$.

    To compute the set $\pi(\Xori)$ in \cref{line:random_proj}, one only need to perform matrix multiplication, which has time complexity $O(n\oridim m)$.

    On each cluster $C\in \partition$, our algorithm first invokes
    $\algapx$ to compute the $O(1)$-approximate solution on $C$ (\cref{line:use_approx_2}), which takes time $T_{A}(\norm{C}, \oridim, \ddim)$.
Then the contraction $\Norm{f-f^\prime}/\Norm{\pi(f)-\pi(f^\prime)}$
    can be checked in time $O(\norm{\cFapx_C}^2 \oridim) \leq O(\tau^2\oridim)$ (\cref{line:eveG}).
After that, either the work on $C$ is done, or 
    $\algmedian$ is further called with parameter $k\leq c_4\tau$ (\cref{line:use_median}).
    This takes time $T_{M}(\norm{C}, m, k, \varepsilon)$.
Therefore, each $C$ can be handled in time 
    $O(T_{A}(\norm{C}, \oridim, \ddim) + \tau^2\oridim + \sum_{k=1}^{\lfloor c_4\tau\rfloor} T_{M}(\norm{C}, m, k, \varepsilon))$.
    The overall time complexity of handling all $C\in \partition$
    is \begin{align*}
        \tilde{O}\left(
            \sum_{C \in \partition} T_{A}(\norm{C}, \oridim, \ddim) + \tau^2\oridim \norm{\partition}
            + \sum_{C\in \partition}\sum_{k=1}^{\lfloor c_4\tau\rfloor} T_{M}(\norm{C}, m, k, \varepsilon)
        \right).
    \end{align*}

    Finally, on each cluster $X_i$, the $(1+\varepsilon)$-approximate
    $1$-median center can be computed by the algorithm in 
    \cref{lemma:CLMPS16} (\cref{line:compute_1_median,line:add_center}), which runs in time 
    $O(\norm{X_i}\cdot \oridim\cdot \log^3(\norm{X_i}/\varepsilon))$, 
    thus a total of $\tilde{O}(n\cdot \oridim)$ time.

    In conclusion, \cref{alg:ptas} runs in time 
    \begin{align*}
        &\quad n \cdot d \cdot 2^{O(\ddim)}
        + T_{A}(n, \oridim, \ddim)
        + \sum_i \sum_{C \in \xchg_i} T_{A}(\norm{C}, \oridim, \ddim)
        + n \oridim \tardim \\
        & + \sum_{C \in \partition} T_{A}(\norm{C}, \oridim, \ddim) + \tau^2\oridim \norm{\partition}
        + \sum_{C\in \partition}\sum_{k=1}^{\lfloor c_4\tau\rfloor} T_{M}(\norm{C}, m, k, \varepsilon).
    \end{align*}
Therefore, \cref{alg:ptas} makes $r = O(\tau \cdot \norm{\partition}) = O(\tau n)$ calls to $\algmedian$, with a total input size $\tau \cdot \sum_{C \in \partition} \norm{C} = O(\tau n)$.
It makes $s = \sum_i \norm{\xchg_i} + \norm{\partition} = O(n)$ calls to $\algapx$, with a total input size $\sum_{i} \sum_{C \in \xchg_i} \norm{C} + \sum_{C \in \partition} \norm{C} = O(n \log \Delta)$.
This completes the proof of \Cref{lemma:ptas_time_complexity}.
\end{proof}

By specializing the algorithms  $\algmedian$ and $\algapx$ as stated in \cref{lemma:k_median_oracle,lemma:ptas_const_apx}, we can complete the proof of \cref{theorem:main_ptas}.

\begin{proof}[Proof of \cref{theorem:main_ptas} (time complexity)]
    By \cref{lemma:ptas_time_complexity}, \cref{alg:ptas}
    runs in time 
    $\tilde{O}(n\oridim\cdot 2^{O(\ddim)} + n\oridim m + \tau\oridim\cdot\opt(\Xori) + \sum_{i=1}^r T_{M}(n_i, m, k_i, \varepsilon)) + \sum_{i=1}^s T_{A}(n_i, \oridim, \ddim)$.

    By \cref{lemma:ptas_const_apx}, each call to the algorithm
    $\algapx$ takes time $T_{A}(n_i, \oridim, \ddim) = \tilde{O}(n_i \oridim \cdot 2^{O(\ddim)})$.
Thus, the total time complexity of running $\algapx$ is $\tilde{O}(n \oridim \cdot 2^{O(\ddim)})$.

    By \cref{lemma:k_median_oracle}, each call to the algorithm 
    $\algmedian$ takes time 
    $T_{M}(n_i, \tardim, k_i, \varepsilon) = \tilde{O}(n_i \tardim k_i+\varepsilon^{-6} \tardim k_i^{O(k_i/\varepsilon^3)})=\tilde{O}(n_i \tardim \tau+\varepsilon^{-6} \tardim \tau^{O(\tau/\varepsilon^3)})$. 
Thus, the total time complexity of running $\algmedian$ is 
    $\tilde{O}(n \tardim \tau^2 + \tau n \cdot \varepsilon^{-6} \tardim \tau^{O(\tau/\varepsilon^3)})$.

    Therefore, the overall time complexity is 
    $\tilde{O}(\tau^2 d \cdot n + n \oridim \cdot 2^{O(\ddim)} + n \tardim \tau^2 + \tau n \cdot \varepsilon^{-6} \tardim \tau^{O(\tau/\varepsilon^3)})$.
    Recall that $\tau=(\ddim/\varepsilon)^{O(\ddim)}$, completing the proof.
\end{proof} 
    \addcontentsline{toc}{section}{References}
    \bibliographystyle{alphaurl}
    \begin{small}	
        \bibliography{ref}

\newcommand{\etalchar}[1]{$^{#1}$}
\begin{thebibliography}{CEM{\etalchar{+}}15}

\bibitem[Ass83]{Assouad83}
Patrice Assouad.
\newblock Plongements lipschitziens dans {${\bf R}\sp{n}$}.
\newblock {\em Bull. Soc. Math. France}, 111(4):429--448, 1983.

\bibitem[BBC{\etalchar{+}}19]{BecchettiBC0S19}
Luca Becchetti, Marc Bury, Vincent Cohen{-}Addad, Fabrizio Grandoni, and Chris Schwiegelshohn.
\newblock Oblivious dimension reduction for $k$-means: beyond subspaces and the {Johnson-Lindenstrauss} lemma.
\newblock In {\em Proceedings of the 51st Annual {ACM} {SIGACT} Symposium on Theory of Computing, {STOC} 2019}, pages 1039--1050. {ACM}, 2019.
\newblock \href {https://doi.org/10.1145/3313276.3316318} {\path{doi:10.1145/3313276.3316318}}.

\bibitem[BGJ{\etalchar{+}}24]{bhattacharya2024dynamic}
Sayan Bhattacharya, Gramoz Goranci, Shaofeng H.-C. Jiang, Yi~Qian, and Yubo Zhang.
\newblock Dynamic facility location in high dimensional euclidean spaces.
\newblock In {\em Forty-first International Conference on Machine Learning}, 2024.
\newblock URL: \url{https://openreview.net/forum?id=rucbIsWoEV}.

\bibitem[BGK16]{DBLP:journals/siamcomp/BartalGK16}
Yair Bartal, Lee{-}Ad Gottlieb, and Robert Krauthgamer.
\newblock The traveling salesman problem: Low-dimensionality implies a polynomial time approximation scheme.
\newblock {\em {SIAM} J. Comput.}, 45(4):1563--1581, 2016.
\newblock \href {https://doi.org/10.1137/130913328} {\path{doi:10.1137/130913328}}.

\bibitem[BRS11]{BartalRS11}
Yair Bartal, Ben Recht, and Leonard~J. Schulman.
\newblock Dimensionality reduction: Beyond the {Johnson-Lindenstrauss} bound.
\newblock In {\em {SODA}}, pages 868--887. {SIAM}, 2011.

\bibitem[BZD10]{BoutsidisZD10}
Christos Boutsidis, Anastasios Zouzias, and Petros Drineas.
\newblock Random projections for $k$-means clustering.
\newblock In {\em 24th Annual Conference on Neural Information Processing Systems, {NeurIPS}}, pages 298--306. Curran Associates, Inc., 2010.
\newblock URL: \url{https://proceedings.neurips.cc/paper/2010/hash/73278a4a86960eeb576a8fd4c9ec6997-Abstract.html}.

\bibitem[CCJ{\etalchar{+}}23]{ChenCJLW23}
Xi~Chen, Vincent Cohen{-}Addad, Rajesh Jayaram, Amit Levi, and Erik Waingarten.
\newblock Streaming {Euclidean} {MST} to a constant factor.
\newblock In {\em {STOC}}, pages 156--169. {ACM}, 2023.

\bibitem[CEM{\etalchar{+}}15]{CohenEMMP15}
Michael~B. Cohen, Sam Elder, Cameron Musco, Christopher Musco, and Madalina Persu.
\newblock Dimensionality reduction for k-means clustering and low rank approximation.
\newblock In {\em Proceedings of the Forty-Seventh Annual {ACM} on Symposium on Theory of Computing, {STOC}}, pages 163--172, 2015.
\newblock \href {https://doi.org/10.1145/2746539.2746569} {\path{doi:10.1145/2746539.2746569}}.

\bibitem[CEMN22]{Cohen-AddadEMN22}
Vincent Cohen{-}Addad, Hossein Esfandiari, Vahab~S. Mirrokni, and Shyam Narayanan.
\newblock Improved approximations for {Euclidean} $k$-means and $k$-median, via nested quasi-independent sets.
\newblock In {\em {STOC}}, pages 1621--1628. {ACM}, 2022.

\bibitem[CFJ{\etalchar{+}}22]{CJK+22:arxiv}
Artur Czumaj, Arnold Filtser, Shaofeng~H.{-}C. Jiang, Robert Krauthgamer, Pavel Vesel{\`y}, and Mingwei Yang.
\newblock Streaming facility location in high dimension via geometric hashing.
\newblock {\em CoRR}, 2022.
\newblock The latest version has additional results compared to the preliminary version in \cite{CzumajJK0Y22}.
\newblock \href {https://arxiv.org/abs/2204.02095} {\path{arXiv:2204.02095}}.

\bibitem[CFS21]{Cohen-AddadFS21}
Vincent Cohen{-}Addad, Andreas~Emil Feldmann, and David Saulpic.
\newblock Near-linear time approximation schemes for clustering in doubling metrics.
\newblock {\em J. {ACM}}, 68(6):44:1--44:34, 2021.

\bibitem[CGJ{\etalchar{+}}24]{CzumajGJKV24}
Artur Czumaj, Guichen Gao, Shaofeng H.-C. Jiang, Robert Krauthgamer, and Pavel Vesel\'{y}.
\newblock {Fully-Scalable MPC Algorithms for Clustering in High Dimension}.
\newblock In {\em 51st International Colloquium on Automata, Languages, and Programming (ICALP 2024)}, volume 297 of {\em Leibniz International Proceedings in Informatics (LIPIcs)}, pages 50:1--50:20. Schloss Dagstuhl -- Leibniz-Zentrum f{\"u}r Informatik, 2024.
\newblock \href {https://doi.org/10.4230/LIPIcs.ICALP.2024.50} {\path{doi:10.4230/LIPIcs.ICALP.2024.50}}.

\bibitem[CHJ18]{DBLP:journals/siamcomp/ChanHJ18}
T.{-}H.~Hubert Chan, Shuguang Hu, and Shaofeng~H.{-}C. Jiang.
\newblock A {PTAS} for the {Steiner} forest problem in doubling metrics.
\newblock {\em {SIAM} J. Comput.}, 47(4):1705--1734, 2018.

\bibitem[CJ18]{DBLP:journals/talg/ChanJ18}
T.{-}H.~Hubert Chan and Shaofeng~H.{-}C. Jiang.
\newblock Reducing curse of dimensionality: Improved {PTAS} for {TSP} (with neighborhoods) in doubling metrics.
\newblock {\em {ACM} Trans. Algorithms}, 14(1):9:1--9:18, 2018.

\bibitem[CJJ20]{DBLP:journals/talg/ChanJJ20}
T.{-}H.~Hubert Chan, Haotian Jiang, and Shaofeng~H.{-}C. Jiang.
\newblock A unified {PTAS} for prize collecting {TSP} and {Steiner} tree problem in doubling metrics.
\newblock {\em {ACM} Trans. Algorithms}, 16(2):24:1--24:23, 2020.

\bibitem[CJK{\etalchar{+}}22]{CzumajJK0Y22}
Artur Czumaj, Shaofeng~H.{-}C. Jiang, Robert Krauthgamer, Pavel Vesel{\'{y}}, and Mingwei Yang.
\newblock Streaming facility location in high dimension via geometric hashing.
\newblock In {\em {FOCS}}, pages 450--461. {IEEE}, 2022.

\bibitem[CJK23]{ChenJK23}
Xiaoyu Chen, Shaofeng~H.{-}C. Jiang, and Robert Krauthgamer.
\newblock Streaming {Euclidean} {Max-Cut}: {Dimension} vs data reduction.
\newblock In {\em {STOC}}, pages 170--182. {ACM}, 2023.
\newblock \href {https://doi.org/10.1145/3564246.3585170} {\path{doi:10.1145/3564246.3585170}}.

\bibitem[Cla99]{Clarkson99}
K.~L. Clarkson.
\newblock Nearest neighbor queries in metric spaces.
\newblock {\em Discrete Comput. Geom.}, 22(1):63--93, 1999.
\newblock \href {https://doi.org/10.1007/PL00009449} {\path{doi:10.1007/PL00009449}}.

\bibitem[CLM{\etalchar{+}}16]{cohenLMPS16}
Michael~B. Cohen, Yin~Tat Lee, Gary~L. Miller, Jakub Pachocki, and Aaron Sidford.
\newblock Geometric median in nearly linear time.
\newblock In {\em {STOC}}, pages 9--21. {ACM}, 2016.

\bibitem[CLMS13]{CzumajLMS13}
Artur Czumaj, Christiane Lammersen, Morteza Monemizadeh, and Christian Sohler.
\newblock ($1 + \varepsilon$)-approximation for facility location in data streams.
\newblock In {\em {SODA}}, pages 1710--1728. {SIAM}, 2013.

\bibitem[CLS{\etalchar{+}}22]{Cohen-AddadLSSS22}
Vincent Cohen{-}Addad, Kasper~Green Larsen, David Saulpic, Chris Schwiegelshohn, and Omar~Ali Sheikh{-}Omar.
\newblock Improved coresets for {Euclidean} $k$-means.
\newblock In {\em NeurIPS}, 2022.

\bibitem[CLSS22]{Cohen-AddadLSS22}
Vincent Cohen{-}Addad, Kasper~Green Larsen, David Saulpic, and Chris Schwiegelshohn.
\newblock Towards optimal lower bounds for $k$-median and $k$-means coresets.
\newblock In {\em {STOC}}, pages 1038--1051. {ACM}, 2022.

\bibitem[CSS21]{Cohen-AddadSS21}
Vincent Cohen{-}Addad, David Saulpic, and Chris Schwiegelshohn.
\newblock A new coreset framework for clustering.
\newblock In {\em {STOC}}, pages 169--182. {ACM}, 2021.

\bibitem[CW22]{CW22}
Moses Charikar and Erik Waingarten.
\newblock The {Johnson-Lindenstrauss Lemma} for clustering and subspace approximation: From coresets to dimension reduction.
\newblock {\em CoRR}, abs/2205.00371, 2022.
\newblock \href {https://arxiv.org/abs/2205.00371} {\path{arXiv:2205.00371}}.

\bibitem[FRS19]{FriggstadRS19}
Zachary Friggstad, Mohsen Rezapour, and Mohammad~R. Salavatipour.
\newblock Local search yields a {PTAS} for $k$-means in doubling metrics.
\newblock {\em {SIAM} J. Comput.}, 48(2):452--480, 2019.

\bibitem[GIV01]{GoelIV01}
Ashish Goel, Piotr Indyk, and Kasturi~R. Varadarajan.
\newblock Reductions among high dimensional proximity problems.
\newblock In {\em {SODA}}, pages 769--778. {ACM/SIAM}, 2001.

\bibitem[GK15]{GottliebK15}
Lee{-}Ad Gottlieb and Robert Krauthgamer.
\newblock A nonlinear approach to dimension reduction.
\newblock {\em Discret. Comput. Geom.}, 54(2):291--315, 2015.
\newblock \href {https://doi.org/10.1007/s00454-015-9707-9} {\path{doi:10.1007/s00454-015-9707-9}}.

\bibitem[GKK14]{GottliebKK14}
Lee{-}Ad Gottlieb, Aryeh Kontorovich, and Robert Krauthgamer.
\newblock Efficient classification for metric data.
\newblock {\em {IEEE} Trans. Inf. Theory}, 60(9):5750--5759, 2014.
\newblock \href {https://doi.org/10.1109/TIT.2014.2339840} {\path{doi:10.1109/TIT.2014.2339840}}.

\bibitem[GKL03]{GuptaKL03}
Anupam Gupta, Robert Krauthgamer, and James~R. Lee.
\newblock Bounded geometries, fractals, and low-distortion embeddings.
\newblock In {\em {FOCS}}, pages 534--543. {IEEE} Computer Society, 2003.
\newblock \href {https://doi.org/10.1109/SFCS.2003.1238226} {\path{doi:10.1109/SFCS.2003.1238226}}.

\bibitem[HK13]{Har-PeledK13}
Sariel Har{-}Peled and Nirman Kumar.
\newblock Approximate nearest neighbor search for low-dimensional queries.
\newblock {\em {SIAM} J. Comput.}, 42(1):138--159, 2013.

\bibitem[HM06]{Har-PeledM06}
Sariel Har{-}Peled and Manor Mendel.
\newblock Fast construction of nets in low-dimensional metrics and their applications.
\newblock {\em {SIAM} J. Comput.}, 35(5):1148--1184, 2006.
\newblock \href {https://doi.org/10.1137/S0097539704446281} {\path{doi:10.1137/S0097539704446281}}.

\bibitem[IN07]{IndykN07}
Piotr Indyk and Assaf Naor.
\newblock Nearest-neighbor-preserving embeddings.
\newblock {\em {ACM} Trans. Algorithms}, 3(3):31, 2007.
\newblock \href {https://doi.org/10.1145/1273340.1273347} {\path{doi:10.1145/1273340.1273347}}.

\bibitem[Ind06]{Indyk06}
P.~Indyk.
\newblock Stable distributions, pseudorandom generators, embeddings, and data stream computation.
\newblock {\em J. ACM}, 53(3):307--323, 2006.
\newblock \href {https://doi.org/10.1145/1147954.1147955} {\path{doi:10.1145/1147954.1147955}}.

\bibitem[JKS24]{JiangKS24}
Shaofeng~H.{-}C. Jiang, Robert Krauthgamer, and Shay Sapir.
\newblock Moderate dimension reduction for $k$-center clustering.
\newblock In {\em SoCG}, volume 293 of {\em LIPIcs}, pages 64:1--64:16. Schloss Dagstuhl - Leibniz-Zentrum f{\"{u}}r Informatik, 2024.
\newblock \href {https://doi.org/10.4230/LIPIcs.SoCG.2024.64} {\path{doi:10.4230/LIPIcs.SoCG.2024.64}}.

\bibitem[JL84]{JL84}
William Johnson and Joram Lindenstrauss.
\newblock Extensions of {Lipschitz} maps into a {Hilbert} space.
\newblock {\em Contemporary Mathematics}, 26:189--206, 01 1984.
\newblock \href {https://doi.org/10.1090/conm/026/737400} {\path{doi:10.1090/conm/026/737400}}.

\bibitem[KR07]{KolliopoulosR07}
Stavros~G. Kolliopoulos and Satish Rao.
\newblock A nearly linear-time approximation scheme for the {Euclidean $k$-Median} problem.
\newblock {\em {SIAM} J. Comput.}, 37(3):757--782, 2007.
\newblock \href {https://doi.org/10.1137/S0097539702404055} {\path{doi:10.1137/S0097539702404055}}.

\bibitem[KR15]{KerberR15}
Michael Kerber and Sharath Raghvendra.
\newblock Approximation and streaming algorithms for projective clustering via random projections.
\newblock In {\em Proceedings of the 27th Canadian Conference on Computational Geometry, {CCCG} 2015}. Queen's University, Ontario, Canada, 2015.
\newblock URL: \url{http://research.cs.queensu.ca/cccg2015/CCCG15-papers/16.pdf}.

\bibitem[KSS10]{KumarSS10}
Amit Kumar, Yogish Sabharwal, and Sandeep Sen.
\newblock Linear-time approximation schemes for clustering problems in any dimensions.
\newblock {\em J. ACM}, 57(2), 2010.
\newblock \href {https://doi.org/10.1145/1667053.1667054} {\path{doi:10.1145/1667053.1667054}}.

\bibitem[LN17]{LarsenN17}
Kasper~Green Larsen and Jelani Nelson.
\newblock Optimality of the {Johnson-Lindenstrauss} lemma.
\newblock In {\em {FOCS}}, pages 633--638. {IEEE} Computer Society, 2017.

\bibitem[LP01]{LP01}
Urs Lang and Conrad Plaut.
\newblock Bilipschitz embeddings of metric spaces into space forms.
\newblock {\em Geometriae Dedicata}, 87(1-3):285--307, 2001.

\bibitem[LSS09]{LammersenSS09}
Christiane Lammersen, Anastasios Sidiropoulos, and Christian Sohler.
\newblock Streaming embeddings with slack.
\newblock In {\em {WADS}}, volume 5664 of {\em Lecture Notes in Computer Science}, pages 483--494. Springer, 2009.

\bibitem[MMR19]{MakarychevMR19}
Konstantin Makarychev, Yury Makarychev, and Ilya~P. Razenshteyn.
\newblock Performance of {Johnson-Lindenstrauss} transform for $k$-means and $k$-medians clustering.
\newblock In {\em {STOC}}, pages 1027--1038. {ACM}, 2019.
\newblock \href {https://arxiv.org/abs/1811.03195} {\path{arXiv:1811.03195}}, \href {https://doi.org/10.1145/3313276.3316350} {\path{doi:10.1145/3313276.3316350}}.

\bibitem[MP03]{MettuP03}
Ramgopal~R. Mettu and C.~Greg Plaxton.
\newblock The online median problem.
\newblock {\em {SIAM} J. Comput.}, 32(3):816--832, 2003.
\newblock \href {https://doi.org/10.1137/S0097539701383443} {\path{doi:10.1137/S0097539701383443}}.

\bibitem[Nao18]{Naor18}
Assaf Naor.
\newblock {Metric dimension reduction: A snapshot of the Ribe program}.
\newblock In {\em Proceedings of the International Congress of Mathematicians (ICM 2018)}, pages 759--837, 2018.
\newblock \href {https://doi.org/10.1142/9789813272880_0029} {\path{doi:10.1142/9789813272880_0029}}.

\bibitem[Nei16]{Neiman16}
Ofer Neiman.
\newblock Low dimensional embeddings of doubling metrics.
\newblock {\em Theory Comput. Syst.}, 58(1):133--152, 2016.

\bibitem[NN12]{NN12}
Assaf Naor and Ofer Neiman.
\newblock {A}ssouad’s theorem with dimension independent of the snowflaking.
\newblock {\em Rev. Mat. Iberoam.}, 28(4):1123–1142, 2012.
\newblock \href {https://doi.org/10.4171/RMI/706} {\path{doi:10.4171/RMI/706}}.

\bibitem[NSIZ21]{NarayananSIZ21}
Shyam Narayanan, Sandeep Silwal, Piotr Indyk, and Or~Zamir.
\newblock Randomized dimensionality reduction for facility location and single-linkage clustering.
\newblock In {\em {ICML}}, volume 139 of {\em Proceedings of Machine Learning Research}, pages 7948--7957. {PMLR}, 2021.
\newblock URL: \url{https://proceedings.mlr.press/v139/narayanan21b.html}.

\bibitem[Tal04]{Talwar04}
Kunal Talwar.
\newblock Bypassing the embedding: algorithms for low dimensional metrics.
\newblock In {\em {STOC}}, pages 281--290. {ACM}, 2004.
\newblock \href {https://doi.org/10.1145/1007352.1007399} {\path{doi:10.1145/1007352.1007399}}.

\bibitem[Tre00]{Trevisan00}
Luca Trevisan.
\newblock When {Hamming} meets {Euclid}: The approximability of geometric {TSP} and {Steiner} tree.
\newblock {\em {SIAM} J. Comput.}, 30(2):475--485, 2000.

\bibitem[Wil18]{Williams18}
Ryan Williams.
\newblock On the difference between closest, furthest, and orthogonal pairs: Nearly-linear vs barely-subquadratic complexity.
\newblock In {\em {SODA}}, pages 1207--1215. {SIAM}, 2018.

\end{thebibliography}
    \end{small}

    \begin{appendices}
        \crefalias{section}{appendix}
        \section{A Faster PTAS for UFL in Discrete Doubling Metrics}
\label{append:ptas_discrete}

In this section, we propose a PTAS for the UFL problem in general discrete doubling metrics, in which points do not necessarily have a vector representation.
We assume that $(\Xori, \dist)$ is a finite metric space with doubling dimension $\ddim$.
A feasible UFL solution can be any subset $F \subseteq \Xori$.
Given access to a distance oracle, our new PTAS runs in time $\tilde{O}(2^{2^{\ddim \cdot \log \ddim}} n)$, making an improvement over the $\tilde{O}(2^{2^{\ddim^2}} n)$ PTAS in~\cite{Cohen-AddadFS21}.
Formally, we are going to prove the following result.

\begin{corollary}
    \label{corollary:ptas_discrete}
    There is an algorithm
    that given as input $0 < \epsilon < 1$ and a finite metric space $(\Xori, \dist)$ of size $n$,
    computes with constant probability a $(1 + \epsilon)$-approximation for UFL,
    and time $\tilde{O}(2^{2^{m'}} n)$, for
    \begin{equation*}
        m' = O\left(\ddim(\Xori) \cdot \log\frac{\ddim(\Xori)}{\varepsilon}\right).
    \end{equation*}
\end{corollary}

Our PTAS is an immediate corollary of our metric decomposition, as proposed in \cref{alg:ptas_discrete}.
Similar to \cref{alg:ptas}, we first compute the hierarchical decomposition $\decom$, and then refine it under the guidance of an $\alpha$-approximate solution $\cFa$ (Lines \ref{line:discrete_decom}-\ref{line:discrete_xchg}).
Next, we use a slightly modified version of \cref{alg:bot_up_partition} to compute the partition $\partition$ (Lines \ref{line:start_cand}-\ref{line:discrete_partition}).
The main difference is that we first compute a candidate facility set $\cand_C$ for each cluster $C \in \bigcup_i \decom_i$, which is defined as all ``nearby'' points around $C$, which serves as potential facilities of $C$ (\cref{line:cand}).
Then, we run $\algPartition(\Xori, \xchg, \kappa)$, with the critical condition in \cref{line:critical_condition} of \cref{alg:bot_up_partition} changed into $\opt^{\cand_C}(C^\xchg) \geq \kappa$, where $\opt^{\cand_C}(C)$ denotes the optimal value of $C$ with facilities restricted to $\cand_C$ (See \cref{sec:prelim}).
We note that, this modification changes the bound in \cref{lemma:opt_C_bounds} into $\kappa \leq \opt^{\cand_C}(C) \leq \tau$, and also changes \cref{lemma:size_of_partition,corollary:apx_X_by_clusters} slightly, which will be further discussed in \cref{append:ptas_discrete_correctness}.
Finally, in Lines \ref{line:enum_k_start}-\ref{line:enum_k_end}, we solve the sub-problem on each $C$ locally, with an ambient space (candidate facility set) $\cand_C$.
The local near optimal $k$-median centers will be added to the final solution (\cref{line:discrete_select_k,line:update_centers}).

\begin{algorithm}[!ht]
    \caption{\ptas\ for UFL on doubling metrics}
    \label{alg:ptas_discrete}
    \DontPrintSemicolon
\KwIn{A finite metric space $(\Xori, \dist)$ with doubling dimension
        $\ddim$, parameter $\varepsilon\in (0, 1)$, an $\alpha$-approximate UFL algorithm $\algapx$ with $\alpha = O(1)$, and a $k$-median oracle
        $\algmedian$.}
let $\cFret \gets \varnothing$ \;
let $c_1, c_2 > 0$ be sufficiently large constant,
        $\kappa \leftarrow c_2(\ddim/\varepsilon)^{c_1\cdot \ddim}$, 
        $\tau \leftarrow 2^{10 \ddim} \cdot \alpha \kappa$ \;
run \cref{alg:decompose_ori} on $\Xori$ to obtain a random hierarchical decomposition $\decom$ \label{line:discrete_decom} \;
run $\algapx$ on $\Xori$ to obtain an $\alpha$-approximate
        solution $\cFa\subseteq \Xori$ for UFL \label{line:discrete_approx} \;
run \cref{alg:xchg} on $\Xori$ to compute the modified decomposition $\xchg = \xchg(\Xori, \decom, \cFa, \varepsilon)$ \label{line:discrete_xchg} \;
\For{$i = 0, 1, \dots, \ell+1$}{ \label{line:start_cand}
            \For{$C \in \decom_i$}{
                compute the candidate facility set of $C$ as $\cand_C := \Ball(C, \frac{100}{\varepsilon} \cdot \rang(C))$ \label{line:cand} \;
            }
        } \label{line:end_cand}
run $\algPartition(\Xori, \xchg, \kappa)$ (\cref{alg:bot_up_partition}) to obtain a partition $\partition$ of $X$,
        such that $\forall C \in \partition$, $\kappa \leq \opt^{\cand_C}(C) \leq \tau$ \label{line:discrete_partition} \;
\For{$C\in\partition$}{ \label{line:discrete_cluster_start}
\For{$k=1, 2, \dots, \lfloor \tau\rfloor$}{ \label{line:enum_k_start}
                $(\cFret_C^k, v_C^k) \gets \algmedian_{k}(C, \cand_C, \varepsilon)$ \label{line:discrete_median} \;
                \tcc{$\cFret_C^k$ is a $(1+\varepsilon)$-approximate $k$-median solution, and $v_C^k$ is the cost.}
            } \label{line:enum_k_end}
$k^*\gets\argmin_{k}\{k+v_C^k\}$ \label{line:discrete_select_k} \;
$\cFret \gets \cFret \cup \cFret_C^{k^*}$ \label{line:update_centers} \;
        } \label{line:discrete_cluster_end}
\Return $\cFret$. \label{line:discrete_ret} \;
\end{algorithm}

\subsection{Correctness of \cref{alg:ptas_discrete}}
\label{append:ptas_discrete_correctness}

In this section, we prove the correctness of \cref{alg:ptas_discrete}.
We first state the refined versions of \cref{lemma:opt_C_bounds,lemma:size_of_partition,corollary:apx_X_by_clusters}.

\begin{lemma}[$\opt^{\cand_C}$ version of \cref{lemma:opt_C_bounds}]
    \label{lemma:discrete_opt_C_bounds}
    For every $C \in \partition$, $\kappa \leq \opt^{\cand_C}(C) \leq 2^{10\ddim} \kappa$. 
\end{lemma}

\begin{proof}[Proof (sketch)]
    The lower bound follows from the construction.
For the upper bound, fix a level $i$ cluster $C^\xchg \in \xchg_i$.
It is not hard to see that, if a level $(i-1)$-level cluster $\widehat{C}^\xchg$ contributes to $\opt^{\cand_C}(C^\xchg)$, then $\cand_{\widehat{C}^\xchg} \subseteq \cand_{C^\xchg}$.
Hence, $\opt^{\cand_\cdot}(\cdot)$ still holds subadditivity on $\xchg$, namely,
    \[
        \opt^{\cand_C}(C) \leq \sum_{\substack{\widehat{C}^\xchg \in \xchg_{i-1} \\ \widehat{C}^\xchg \text{ contributes to } C^\xchg}} \opt^{\cand_{\widehat{C}}}(\widehat{C}).
    \]
Therefore, we can use the same argument in \cref{subsec:opt_C_bounds} to prove \cref{lemma:discrete_opt_C_bounds}.
\end{proof}

\begin{lemma}[$\opt^{\cand_C}$ version of \cref{lemma:size_of_partition}]
    \label{lemma:discrete_size_of_partition}
    \cref{lemma:size_of_partition} remains unchanged for the modified $\partition$.
\end{lemma}

\begin{proof}[Proof (sketch)]
    Recall the definition of $\csol_C$ in \eqref{eq:def_sol_S_C}.
    Clearly, $\csol_C \subseteq \cand_C$.
    Thus, we can replace $\opt(\cdot)$ in \eqref{eq:lemma:apx_opt_x_tech} with $\opt^{\cand_C}(\cdot)$, stated as follows.
    \begin{align*}
        \EE{\sum_{C\in \partition}\opt^{\cand_C}(C)}
        \leq 2 \alpha \opt(\Xori)
        + \left(\frac{\ddim}{\varepsilon}\right)^{c_1\cdot \ddim}
        \cdot 2\EE{\norm{\partition}}.
    \end{align*}
    The desired bound follows immediately from that $\opt^{\cand_C}(C) \geq \kappa$.
\end{proof}

\begin{lemma}[$\opt^{\cand_C}$ version of \cref{corollary:apx_X_by_clusters}]
    \label{lemma:discrete_apx_pi_X_by_clusters}
    For the modified partition $\partition$, 
    \begin{align*}
        \opt(\Xori)\geq \sum_{C\in\partition} \opt^{\cand_C}(C) -
        \varepsilon \cdot \opt(\Xori)
    \end{align*}
    holds with probability at least $1-\delta$.
\end{lemma}

\begin{proof}[Proof (sketch)]
    Recall the definition of $\cFnew_{\pi(C)}$ in \eqref{eq:def_sol_Fnew_C}.
    Let $\pi$ be the identity mapping. 
    Then, $\cFnew_{\pi(C)} \subseteq \cand_C$.
    Therefore, $\opt^{\cand_C}(C) \leq \cost(C, \cFnew_{\pi(C)})$.
    Following the same argument of \Cref{lemma:apx_pi_X_by_clusters} finishes the proof.
\end{proof}

Now, we are ready to prove the correctness of \cref{alg:ptas_discrete}.

\begin{proof}[Proof of \cref{corollary:ptas_discrete} (correctness)]
    For each $C \in \partition$, by \cref{lemma:discrete_opt_C_bounds}, 
    $\opt^{\cand_C}(C) \leq \tau$.
Thus, there exists $k \in \{1, 2, \dots, \lfloor \tau \rfloor\}$, such that $\opt^{\cand_C}(C) = \med_k^{\cand_C}(C)$.
By our selection of $k^*$ in \cref{line:discrete_select_k}, 
    \begin{align*}
        \cost(C, \cFret_C^{k^*}) 
        = k^* + v_C^{k^*}
        \leq k + v_C^k 
        \leq (1 + \varepsilon) (k + \med_k^{\cand_C}(C)) 
        = (1 + \varepsilon) \opt^{\cand_C}(C).
    \end{align*}
    Hence, 
    \begin{align*}
        \cost(\Xori, \cFret) 
        \leq \sum_{C \in \partition} \cost(C, \cFret_C^{k^*})
        \leq (1 + \varepsilon) \sum_{C \in \partition} \opt^{\cand_C}(C) 
        \leq (1 + \varepsilon)^2 \opt(\Xori),
    \end{align*}
    where the last inequality follows from \cref{lemma:discrete_apx_pi_X_by_clusters}.
    This completes the proof.
\end{proof}

\subsection{Time Complexity of \cref{alg:ptas_discrete}}
\label{append:ptas_discrete_time}

In this section, we analyze the time complexity of \cref{alg:ptas_discrete}.
We assume the access to a distance oracle.

\begin{proof}[Proof of \cref{corollary:ptas_discrete} (time complexity)]
    The preprocessing steps (Lines \ref{line:discrete_decom}-\ref{line:discrete_xchg}) can be computed within time $\tilde{O}(n \cdot 2^{O(\ddim)})$. We refer the readers to \cref{subsec:ptas_time_complexity} for more details.

    For each level $i \in [\ell]$, the candidate set of all $C \in \decom_i$ can be computed in time $\tilde{O}(n \cdot \varepsilon^{-O(\ddim)})$, using an ANN data structure (\cref{line:cand}).
Thus, computing $\bigcup_i \{\cand_C\}_{C \in \decom_i}$ has time complexity $\tilde{O}(n \cdot \varepsilon^{-O(\ddim)})$.
Next, the partition $\partition$ can be computed within time $\tilde{O}(n \cdot 2^{O(\ddim)})$ in \cref{line:discrete_partition}.

    On each $C \in \partition$, the oracle $\algmedian$ is invoked to compute the $k$-median solution for $1 \leq k \leq \tau$.
To compute the $k$-median, $\algmedian$ first constructs an $\varepsilon$-coreset $S \subseteq C$ with respect to ambient space $\cand_C$, running in time $\tilde{O}(|C| \cdot k)$.
By~\cite{Cohen-AddadSS21,Cohen-AddadLSS22}, the coreset has size at most $|S| = \tilde{O}(\varepsilon^{-2} k \cdot \ddim)$.
Then $\algmedian$ enumerates all $k^{\tilde{O}(\varepsilon^{-2} k \cdot \ddim)}$ possible partitions of $S$.
Let $\{X_C^1, X_C^2, \dots, X_C^k\}$ be any one of these partitions. 
For every $i \in [k]$, $\algmedian$ enumerates the candidate facility set $\cand_C$ to find the geometric center of $X_C^i$.
Therefore, the time complexity of computing $\cFret$ (Lines \ref{line:discrete_cluster_start}-\ref{line:discrete_cluster_end}) is 
    \begin{align*}
        \sum_{C \in \partition} \sum_{k=1}^{\lfloor \tau \rfloor} 
        \left(
            \norm{C} \cdot k 
            + k^{\tilde{O}(\varepsilon^{-2} k \cdot \ddim)} \sum_{i=1}^k \norm{X_C^i} \cdot \norm{\cand_C}
        \right) 
        \leq \tau^2 n 
        + \tau^{O(\varepsilon^{-2} \tau \cdot \ddim)} \sum_{C \in \partition} \norm{\cand_C}.
    \end{align*}
Recall that $\cand_C = \Ball(C, \frac{100}{\varepsilon} \rang(C))$.
For every fixed point $x \in \Xori$ and level $i \in [\ell]$, by packing property, there are at most $(1600/\varepsilon)^{\ddim}$ clusters $C \in \decom_i$ satisfying $x \in \cand_C$.
Thus, $\sum_{C \in \partition} \norm{\cand_C} \leq \varepsilon^{-O(\ddim)} n \log \Delta$.
    
    This completes the proof.
\end{proof}

\section{Missing Proofs in \cref{sec:partition}}

\subsection{Proof of \cref{lemma:piX_leq_X_value_general}: An Upper Bound of $\opt(\pi(\Xori))$}
\label{append:proof_ub_pi_X}

\lemmapiXleqXvaluegeneral*

The high-level idea is straightforward.
Since $\pi(\cFopt)$ is naturally a feasible solution
for $\pi(\Xori)$, $\opt(\pi(\Xori))$ can be upper bounded by 
$\cost(\pi(\Xori), \pi(\cFopt))$, where $\cFopt$ is the optimal 
solution of $\Xori$.

\begin{proof}[Proof of \cref{lemma:piX_leq_X_value_general}]
    Denote the optimal solution of $\Xori$ by 
    $\cFopt \subset \RR^\oridim$.
    For every $x\in \Xori$, denote the nearest facility of $x$ by 
    $\Fopt(x)$. Then 
    $\opt(\Xori)=\sum_{x\in\Xori}\Norm{x-\Fopt(x)}$.
    By \cref{prop:random_proj}, for every $x\in \Xori$, 
    \begin{align*}
        \EE{\max\Set{0, \frac{\Norm{\pi(x)-\pi(\Fopt(x))}}{\Norm{x-\Fopt(x)}}-(1+t)}}
        \leq \frac{1}{\tardim t}e^{-t^2 \tardim/2}.
    \end{align*}
    Thus \begin{align*}
        \EE{\max\Set{0, \Norm{\pi(x)
        -\pi(\Fopt(x))}-(1+t)\Norm{x-\Fopt(x)}}}
        \leq  \frac{1}{\tardim t}e^{-t^2 \tardim/2}
        \Norm{x-\Fopt(x)}.
    \end{align*}
    Note that $\pi(\cFopt) \subset \RR^{\tardim}$ 
    is a solution for $\pi(\Xori)$. Hence
    \begin{align*}
        \opt(\pi(\Xori))-(1+t)\opt(\Xori)
        &\leq \cost(\pi(\Xori), \pi(\cFopt))-(1+t)\opt(\Xori)\\
        &\leq \sum_{x\in \Xori} 
        (\Norm{\pi(x)-\pi(\Fopt(x))}-(1+t)\Norm{x-\Fopt(x)})\\
        &\leq \sum_{x\in \Xori} 
        \max\Set{0, \Norm{\pi(x)-\pi(\Fopt(x))}-(1+t)
        \Norm{x-\Fopt(x)}}.
    \end{align*}
    Taking expectation, we have 
    \begin{align*}
        &\quad \EE{\max\Set{0, \opt(\pi(\Xori))-(1+t)\opt(\Xori)}}\\
        &\leq \sum_{x\in \Xori}
        \EE{\max\Set{0, \Norm{\pi(x)-\pi(\Fopt(x))}
        -(1+t)\Norm{x-\Fopt(x)}}}\\
        &\leq \frac{1}{\tardim t}e^{-t^2 \tardim/2}
        \sum_{x\in \Xori}\Norm{x-\Fopt(x)}\\
        &= \frac{1}{\tardim t}e^{-t^2 \tardim/2} \cdot\opt(\Xori).
    \end{align*}

    Finally,
    \begin{align*}
        \PR{\opt(\pi(\Xori))\geq (1+t)\opt(\Xori)}
        &\leq \PR{\max\{0, \opt(\pi(\Xori))-(1+t/2)\opt(\Xori)\}
        \geq \frac{t}{2}\opt(\Xori)}\\
        &\leq \frac{4}{t^2\tardim} e^{-t^2\tardim/8},
    \end{align*}
    where the last inequality follows from Markov's inequality.
    This completes the proof.
\end{proof}

\subsection{Proof of \cref{lemma:jl_small_opt_contraction_pr}: Contraction on A Single Cluster}
\label{append:proof_small_opt}

\lemmajlsmalloptcontractionpr*

We start with stating a variant of Theorem 3.4 in~\cite{MakarychevMR19}.

\begin{lemma}[{A variant of~\cite[Theorem 3.4]{MakarychevMR19}}]
    \label{lemma:MMR_3.4}
    Let $X\subset \RR^\oridim$ be a finite point set and $k\in\NN$.
    Let $\pi\colon\RR^\oridim \to\RR^\tardim$ be a random linear map.
    Let $C$ be a random subset of $X$ (which may depend on $\pi$).
    Then there exists a constant $c>0$, such that for every 
    $\varepsilon\in (0, 1)$, 
    if $\tardim\geq c\cdot \varepsilon^{-2}\log(1/\varepsilon)$, 
    then for every $\beta > 0$,
    \begin{align*}
&\med_1(C)\leq (1+\varepsilon)\med_1(\pi(C))+\beta\cdot \med_k(X),
    \end{align*}
    holds with probability at least 
    $1-e^{-\Omega(\varepsilon^2\tardim)}O(k^2 + k/\beta)$, 
\end{lemma}

\begin{lemma}
    \label{lemma:prmeddistortion}
    Let $X\subset \RR^\oridim$ be a finite point set and $k\in\NN$.
    Let $\pi\colon\RR^\oridim \to\RR^\tardim$ be a random linear map.
    Then there exists a constant $c>0$, such that for every 
    $\varepsilon\in (0, 1)$, 
    if $\tardim\geq c\cdot \varepsilon^{-2}\log(1/\varepsilon)$, 
    then
    \begin{align}
&\Pr\left[\med_k(\pi(X))\leq\frac{1}{1+\varepsilon}\med_k(X)\right]
        \leq k^2\cdot e^{-\Omega(\varepsilon^2 \tardim)}
        \label{for:prmedpileqmedx}
    \end{align}
\end{lemma}

\begin{proof}
    If $\med_k(X)\geq(1+\varepsilon)\med_k(\pi(X))$, then 
    the optimal clustering for $\pi(X)$ satisfies 
    $\cost_k(\cC)\geq (1+\varepsilon) \cost_k(\pi(\cC))$.
    Denote $\cC=\{C_1, C_2, \dots, C_k\}$. 
    Then there exists $i\in [k]$ such that 
    $\med_1(C_i)\geq (1+\varepsilon/2)\med_1(\pi(C_i))
    +\varepsilon/(4k)\med_k(X)$. Since $C_i$ is a random subset 
    of $X$, by \cref{lemma:MMR_3.4}, 
    \begin{align*}
        &\Pr_{\pi,C}\left[\med_1(C_i)\geq(1+\varepsilon/2)\med_1(\pi(C_i))+
        \frac{\varepsilon}{4k}\med_k(X)\right]\\
        &\leq e^{-\Omega(\varepsilon^2 \tardim)}
        \cdot O\left(k^2 + \frac{k}{\varepsilon/(4k)}\right)\\
        &\leq O(k^2)\cdot e^{-\Omega(\varepsilon^2 \tardim)}
    \end{align*}
\end{proof}

Now we are ready to prove \cref{lemma:jl_small_opt_contraction_pr}.

\begin{proof}[Proof of \cref{lemma:jl_small_opt_contraction_pr}]
    Suppose $\OPT(\pi(C))\leq \OPT(C)/(1+\varepsilon) < \tau$, 
    then there exists some $k\in [\tau]$, such that
    \begin{align*}
        k+\med_k(\pi(C))
        =\OPT(\pi(C))
        \leq \frac{1}{1+\varepsilon}\OPT(C)
        \leq \frac{1}{1+\varepsilon}(k+\med_k(C))
        \leq k+\frac{1}{1+\varepsilon}\med_k(C)
    \end{align*}
    By \cref{lemma:prmeddistortion}, this happens
    with probability at most $O(k^2)\cdot e^{-\Omega(\varepsilon^2\tardim)}$.
    By union bound,
    \begin{align*}
        \Pr\left[\OPT(\pi(C))\leq \frac{1}{1+\varepsilon}\OPT(C)\right]
        \leq O(\tau^2)\cdot e^{-\Omega(\varepsilon^2\tardim)}\cdot \tau
        \leq O(\tau^3)\cdot e^{-\Omega(\varepsilon^2\tardim)}.
    \end{align*}
\end{proof}

\subsection{Proof of \cref{lemma:moving_dist_pi}: Cost of Moving Bad Points}
\label{append:moving_cost}

\lemmamovingdistpi*

\begin{proof}[Proof of \cref{lemma:moving_dist_pi}]
    Calculate the expectation as follows:
    \begin{align*}
        &\quad \ee_{\pi, \decom}\left[
            \sum_{x\in \Xori} \Norm{\pi(x)-\mov^\varepsilon(\pi(x))}
        \right]\\
        &=\ee_{\pi, \decom}\left[\sum_{x\in\badpi^\varepsilon}
        \Norm{\pi(x)-\pi\circ \sol_C(x)}\right]\\
        &=\ee_\decom\left[
            \sum_{x\in \Xori}\ee_\pi\Big[
                \indicator(x\in \badpi^\varepsilon)
                \Norm{\pi(x)-\pi\circ \sol_C(x)}
            \mid \decom\Big]
        \right]
    \end{align*}
    We note that conditioning on $\decom$, $\csol$ is independent
    of $\pi$. By \cref{prop:random_proj}, 
    \begin{align*}
        \ee_{\pi}\left[
            \indicator(x\in \badpi^\varepsilon)
            \max\left\{0, \frac{\Norm{\pi(x)-\pi\circ \sol_C(x)}}
            {\Norm{x-\sol_C(x)}}-2\right\}
        \mid \decom \right]
        \leq e^{-\tardim/2}.
    \end{align*}
    Thus,
    \begin{align*}
        &\quad \ee_\pi\Big[
                \indicator(x\in \badpi^\varepsilon)
                \Norm{\pi(x)-\pi\circ \sol_C(x)}
            \mid \decom\Big]\\
        &\leq 2\Norm{x- \sol_C(x)}\cdot \ee_\pi\Big[
            \indicator(x\in \badpi^\varepsilon)
            \mid \decom\Big]
        + e^{-\Omega(\tardim)} \Norm{x-\sol_C(x)}.
    \end{align*}

    Hence, the total moving distance can be further upper bounded by
    \begin{align*}
        &\quad 2\ee_{\pi, \decom}\left[
            \sum_{x\in \Xori} \indicator(x\in \badpi^\varepsilon)
            \cdot \Norm{x- \sol_C(x)}
        \right]
        +e^{-\Omega(\tardim)}\ee_{\decom}\left[
            \sum_{x\in \Xori} \Norm{x- \sol_C(x)}
        \right]\\
        &\leq 4\ee_{\pi, \decom}\left[
            \sum_{x\in \Xori} \indicator(x\in \badpi^\varepsilon)
            \cdot \Norm{x- \Fa(x)}
        \right]
        +2e^{-\Omega(\tardim)} \sum_{x\in \Xori} \Norm{x- \Fa(x)} 
        &\text{(\cref{lemma:feasibel_sol_const_apx})}\\
        &=4\sum_{x\in \Xori}\Norm{x- \Fa(x)}\cdot 
        \Pr_{\pi, \decom}[x \in \badpi^\varepsilon]
        +2e^{-\Omega(\tardim)} \sum_{x\in \Xori} \Norm{x- \Fa(x)}\\
        &\leq (O(\varepsilon^2) + 2e^{-\Omega(\tardim)})\cdot \alpha\opt(\Xori)\\
        &\leq O(\alpha\varepsilon^2) \opt(\Xori).
    \end{align*}
\end{proof}

\subsection{Proof of \cref{lemma:ball_expansion}}
\label{append:ball_expansion}

\lemmaballexpansion*

\begin{proof}[Proof of \cref{lemma:ball_expansion}]
    For $i\in \NN$, let $\Net_i$ be a $2^{-i}$-net on $X$. For 
    every $x\in X$, denote by $u_i(x)\in \Net_i$ the net point 
    satisfying $\Norm{x-u_i(x)}\leq 2^{-i}$. Assume without loss 
    of generality that $u_0(x)=\bm{0}$. Note that 
    \begin{align*}
        \Norm{\pi(x)}\leq \sum_{i=1}^\infty \Norm{\pi(u_i(x))-\pi(u_{i+1}(x))}.
    \end{align*}
    Hence \begin{align*}
        \PR{\exists x\in X, \Norm{\pi(x)}\geq t}
        \leq \PR{\exists i\in \NN, \Norm{\pi(u_i(x))-\pi(u_{i+1}(x))}
        \geq \frac{1}{3}\left(\frac{2}{3}\right)^i\cdot t}.
    \end{align*}
    On the other hand, by triangle inequality, 
    $\Norm{u_i(x)-u_{i+1}(x)}\leq \Norm{u_i(x)-x}+\Norm{x-u_{i+1}(x)}
    \leq 2^{-i} + 2^{-i-1}\leq 3\cdot 2^{-i-1}$. Thus above can be
    further upper bounded by 
    \begin{align*}
        &\quad \PR{\exists i\in \NN, \exists u_i \in \Net_i, 
        \exists u_{i+1}\in \Net_{i+1}, \Norm{\pi(u_i)-\pi(u_{i+1})}
        \geq \frac{2}{9}\left(\frac{4}{3}\right)^i t\cdot \Norm{u_i-u_{i+1}}}\\
        &\leq \sum_{i=0}^\infty \norm{\Net_i}\cdot \norm{\Net_{i+1}}
        \cdot \exp\left(-\left(\frac{2^{2i+1}t}{3^{i+2}}-1\right)^2 \tardim\right)\\
        &\leq \sum_{i=0}^\infty 2^{(i+1)\ddim}\cdot 2^{(i+2)\ddim}
        \cdot \exp\left(-\left(\frac{2^{2i+1}t}{3^{i+2}}-1\right)^2 \tardim\right).
    \end{align*}
    When $t$ is sufficiently large and $\tardim=\Omega(\ddim)$, 
    above can be upper bounded by
    \begin{align*}
        \sum_{i=0}^\infty 2^{(i+1)\ddim}\cdot 2^{(i+2)\ddim}
        \cdot e^{-c^\prime (i+1) \cdot t^2 \tardim}
        \leq e^{-c_2 \cdot t^2\tardim}.
    \end{align*}
\end{proof}

\subsection{Proof of \cref{lemma:ball_contraction}}
\label{append:contraction_balls}

\lemmaballcontraction*

We use similar techniques as in~\cite{IndykN07,NarayananSIZ21}.

\begin{proof}[Proof of \cref{lemma:ball_contraction}]
    For $i\in \NN$, define the $i$-th ``ring'' as 
    $R_i:=X\cap (\Ball(\bm{0}, (L+i+1)r)\setminus \Ball(\bm{0}, (L+i)r))$.
    Let $\Net_i$ be an $(r/2)$-net on $R_i$. Let $T_0$ be a sufficiently
    large constant which satisfies the condition in 
    \cref{lemma:ball_expansion} (i.e. $T_0 > T$).
    For $i\in \NN$ and $x\in R_i$, if $\Norm{\pi(x)}\leq r$, 
    then the closest net point $u \in \Net_i$ must satisfy at least one of 
    the following conditions:
    \begin{itemize}
        \item $\Norm{\pi(x)-\pi(u)}\geq \frac{\sqrt{i}+T_0}{2} r$;
        \item $\Norm{\pi(u)}\leq \frac{\sqrt{i}+T_0+2}{2}r$.
    \end{itemize}

    Then one can upper bound the probability as follows:
    \begin{align*}
        &\quad \Pr\left[\exists x\in X, \Norm{x} > L\cdot r \text{ and }
        \Norm{\pi(x)} \leq r\right]\\
        &\leq \sum_{i=0}^\infty \PR{\exists x\in R_i, \Norm{\pi(x)}\leq r}\\
        &\leq \sum_{i=0}^\infty \PR{\exists u\in \Net_i, \exists x\in X\cap \Ball(u, r/2), \Norm{\pi(x)-\pi(u)}\geq \frac{\sqrt{i}+T_0}{2} r}\\
        &\quad +\sum_{i=0}^\infty \PR{\exists u\in \Net_i, \Norm{\pi(u)}\leq \frac{\sqrt{i}+T_0+2}{2} r}.
    \end{align*}

    By \cref{lemma:ball_expansion} and a union bound over 
    $u\in \Net_i$, the first summation can be upper bounded by 
    \begin{align}
        &\quad \sum_{i=0}^\infty \norm{\Net_i} \cdot 
        \PR{\exists x\in X\cap \Ball(u, r/2), \Norm{\pi(x)-\pi(u)}\geq \frac{\sqrt{i}+T_0}{2} r}\notag\\
        &\leq \sum_{i=0}^\infty (2L+2i+2)^\ddim \cdot 
        e^{-c_2 \cdot (\sqrt{i}+T_0)^2\tardim}
        \leq \sum_{i=0}^\infty e^{-c_2 \cdot (i+1)\tardim/2}
        \leq e^{-c_3 \tardim}.
        \label{eq:ball_first_summation}
    \end{align}
    For each net point $u\in \Net_i$, 
    $\Norm{u}\geq \Norm{x}-\Norm{u-x} \geq (L+i-1/2)r$. 
    Thus by \cref{prop:random_proj} and a union bound over 
    $u\in \Net_i$, the second summation can be upper bounded by 
    \begin{align}
        &\quad \sum_{i=0}^\infty \norm{\Net_i}\cdot 
        \PR{\Norm{\pi(u)}\leq \frac{\sqrt{i}+T_0+2}{2i+2L-1} \Norm{u}}
        \leq \sum_{i=0}^\infty \norm{\Net_i}\cdot 
        \PR{\Norm{\pi(u)}\leq \frac{1}{\sqrt{i}+8} \Norm{u}}\notag\\
        &\leq \sum_{i=0}^\infty (2L+2i+2)^\ddim \cdot 
        \left(\frac{3}{\sqrt{i}+8}\right)^{\tardim}
        \leq \sum_{i=0}^\infty
        \left(\frac{1}{i+2}\right)^{c_4\cdot \tardim}
        \leq \int_2^{+\infty} \frac{1}{x^{c_4\cdot m}}\, \mathrm{d} x
        \leq e^{-c_5\cdot \tardim}
        \label{eq:ball_second_summation}
    \end{align}
    The first inequality follows that when we choose a constant
    $L\gg T_0$, it holds
    $\frac{\sqrt{i}+T_0+2}{2i+2L-1} \leq \frac{1}{\sqrt{i}+8}$.

    Combining \eqref{eq:ball_first_summation} with 
    \eqref{eq:ball_second_summation} completes the proof.
\end{proof}

\section{Missing Proofs in \cref{sec:ptas}}
\subsection{Proof of \cref{lemma:ptas_exp_single_cluster}: Solution Preserved on A Single Cluster}
\label{append:ptas_exp_single_cluster}

\lemmaptasexpsinglecluster*

In this section, we prove \cref{lemma:ptas_exp_single_cluster}.
Technically, we need some results in~\cite{MakarychevMR19}
for our following proofs.
We first state Definition 3.1, Theorems 3.2 and 3.3 in~\cite{MakarychevMR19}.

\begin{definition}[Everywhere sparse graphs~\cite{MakarychevMR19}]
    \label{def:ptas_everywhere_sparse}
    A graph $H=(V, E)$ is $\theta$-everywhere sparse if 
    $\deg(u)\leq \theta\norm{V}$ for every $u\in V$.
\end{definition}

\begin{definition}[Distortion graphs]
    \label{def:ptas_distortion_graph}
    Let $X\subset \RR^\tardim$ be a finite multiset of points and 
    $\varphi\colon \RR^\tardim\to \RR^{d^\prime}$ be a map.
    For $A>1$, the $A$-expansion graph of $X$ with respect to
    $\varphi$ is a graph $G=(X, E)$, such that for every $x, y\in X$, 
    $(x, y)\in E$ iff $\Norm{\varphi(x)-\varphi(y)}\geq A\cdot \Norm{x-y}$.

    For $0\leq B\leq 1$, the $B$-contraction graph of $X$ with respect to
    $\varphi$ is a graph $G=(X, E)$, such that for every $x, y\in X$, 
    $(x, y)\in E$ iff $\Norm{\varphi(x)-\varphi(y)}\leq B\cdot \Norm{x-y}$.
\end{definition}

\begin{lemma}[{\cite[Theorem 3.2]{MakarychevMR19}}]
    \label{lemma:ptas_MMR_3.2}
    Consider a finite set $X$ and a random graph $H=(V, E)$,
    where $V$ is a random subset of $X$ and $E$ is a random set 
    of edges between vertices in $V$. Let $\theta\in (0, 1/2)$.
    Assume that $\Pr[(x, y)\in E]\leq \beta$ for every 
    $x, y\in X$, where $\beta\leq \theta^7/600$. Then there exists
    a random subset $V^\prime \subset V$ ($V^\prime$ is defined 
    on the same probabilistic space as $H$) such that 
    \begin{itemize}
        \item $H[V^\prime]$ is $\theta$-everywhere sparse.
        \item $\Pr[u\in V\setminus V^\prime]\leq \theta$ for all 
        $u\in X$.
    \end{itemize}
\end{lemma}

\begin{lemma}[{\cite[Theorem 3.3]{MakarychevMR19}}]
    \label{lemma:ptas_MMR_3.3}
    Let $C\subset \RR^\tardim$ be a finite multiset of points and 
    $\varphi\colon \RR^\tardim\to \RR^{d^\prime}$ be a map.
    Assume that the $(1+t)$-expansion graph of $C$ is 
    $\theta$-sparse, where $\theta\leq 1/100$. Then 
    \begin{equation*}
        \med_1(\varphi(C))\leq (1+t)(1+\sqrt{\theta})\med_1(C).
    \end{equation*}
\end{lemma}

\begin{proof}[Proof of \cref{lemma:ptas_exp_single_cluster}]
    Fix a cluster $C\in\partition$ and assume $\eveG_C\cap\eveK_C$
    happens. For every $i\in\U_C$, $X_i$ is a random subset of $C$. 
    Denote by $X_i^\circ$ all points $x$ in $X_i$ such that the 
    distances between $x$ and all facilities in $\cFapx_C$ are
    well-preserved, i.e.
    \begin{align}
        X_i^\circ:=\Set{x\in X_i\colon \forall f\in \cFapx_C, 
        \Norm{\pi(x)-\pi(f)}\geq \frac{1}{1+\varepsilon}
        \Norm{x-f}}.
    \end{align}

    Let $G=(X,E)$ be the $1/(1+\varepsilon)$-contraction
    graph of $X$ with respect to $\pi$. 
    Then $\Pr[(x,y)\in E]\leq e^{-\Omega(\varepsilon^2 \tardim)}$.
    Since $X_i$ is a random subset of $X$, by 
    \cref{lemma:ptas_MMR_3.2}, there exists a 
    random subset $X_i^\prime\subseteq X_i^\circ$, such that 
    $G[X_i^\prime]$ is $\theta$-everywhere sparse, 
    where $\theta=e^{-\Omega(\varepsilon^2 \tardim)}$. Furthermore, 
    for any $x\in X$, $\Pr[x\in X_i^\circ\setminus X_i^\prime]
    \leq \theta$.

    Define a function $\phi\colon X_i\to X\cup \cFapx_C$ as follows,
    \begin{align*}
        \phi(x)=\begin{cases}
            x, &x\in X_i^\prime;\\
            F_C(x), &x\in X_i\setminus X_i^\prime
        \end{cases}
    \end{align*}
    Denote multiset $\widetilde{X}_i=\phi(X_i)$. No hard to see that 
    conditioning on event $\eveG_C$, the $1/(1+\varepsilon)$-contraction
    graph $G[\widetilde{X}_i]$
    is $\theta$-everywhere sparse. Applying \cref{lemma:ptas_MMR_3.3}
    on map $\varphi=\pi^{-1}$, we have
    \begin{equation}
        \med_1(\widetilde{X}_i)\leq (1+\varepsilon)(1+\sqrt{\theta})
        \cdot \med_1(\pi(\widetilde{X}_i))
        \leq (1+2\varepsilon) \med_1(\pi(\widetilde{X}_i)),
        \label{eq:ptas_MMR_theo_3.3}
    \end{equation}
    given $\tardim=\Omega(\varepsilon^{-2}\ddim \log(\ddim/\varepsilon))$.

    Let $f$ be the optimal 1-median center of $\widetilde{X}_i$.
    Notice that
    \begin{align*}
        &\med_1(X_i)\leq \sum_{x\in X_i} \Norm{x-f},\\
        &\med_1(\widetilde{X}_i)=\sum_{x\in X_i} \Norm{\phi(x)-f}
        =\sum_{x\in X_i^\prime} \Norm{x-f}+\sum_{x\in X_i\setminus X_i^\prime}
        \Norm{F_C(x)-f}.
    \end{align*}
    Thus \begin{align*}
        \med_1(X_i)-\med_1(\widetilde{X}_i)
        \leq \sum_{x\in X_i\setminus X_i^\prime}
        \left(\Norm{x-f}-\Norm{F_C(x)-f}\right)
        \leq \sum_{x\in X_i\setminus X_i^\prime} \Norm{x-F_C(x)}
    \end{align*}
    Analogously, \begin{align*}
        \med_1(\pi(\widetilde{X}_i))-\med_1(\pi(X_i))
        \leq \sum_{x\in X_i\setminus X_i^\prime} 
        \Norm{\pi(x)-\pi(F_C(x))}
    \end{align*}
    Combining with \eqref{eq:ptas_MMR_theo_3.3}, we have 
    \begin{equation}
        \med_1(X_i)\leq (1+2\varepsilon)\cdot \med_1(\pi(X_i))+(1+2\varepsilon)
        \cdot \sum_{x\in X_i\setminus X_i^\prime}
        \Big(\Norm{x-F_C(x)}+\Norm{\pi(x)-\pi(F_C(x))}\Big).
\end{equation}

    Therefore, 
    \begin{align*}
        &\quad \indicator(\eveG_C\cap\eveK_C)
        \cdot \max\Set{0, \med_1(X_i)-(1+2\varepsilon)\cdot \med_1(\pi(X_i))}\\
        &\leq \indicator(\eveG_C\cap\eveK_C)
        \cdot (1+2\varepsilon) \cdot \sum_{x\in X_i\setminus X_i^\prime}
        \Big(\Norm{x-F_C(x)}+\Norm{\pi(x)-\pi(F_C(x))}\Big)
        \\
        &\leq 2\cdot \sum_{x\in C}\indicator(\eveK_C)
        \indicator(x\in X_i\setminus X_i^\prime)\cdot 
        \Big(\Norm{x-F_C(x)}+\Norm{\pi(x)-\pi(F_C(x))}\Big).
    \end{align*}

    Summing over $i\in \U_C$ yields
    \begin{align*}
        &\quad \sum_{i\in \U_C} \indicator(\eveG_C\cap\eveK_C)
        \cdot \max\Set{0, \med_1(X_i)-(1+2\varepsilon)\cdot \med_1(\pi(X_i))}\\
        &\leq 2\cdot \sum_{x\in C}\indicator(\eveK_C)
        \cdot 
        \Big(\Norm{x-F_C(x)}+\Norm{\pi(x)-\pi(F_C(x))}\Big)
        \sum_{i\in \U_C} \indicator(x\in X_i\setminus X_i^\prime).
    \end{align*}
    By definition, $\norm{\U_C}\leq \algmedian(\pi(C)) 
    \leq (1+\varepsilon) \opt(\pi(C))$. Conditioning on event
    $\eveK_C$, this is further upper bounded by 
    $\norm{\U_C} \leq (1+\varepsilon) c_4\tau \leq 2c_4\tau$.
    Thus $\sum_{i\in \U_C}\indicator(x\in X_i\setminus X_i^\prime) 
    \leq \norm{\U_C}\cdot \indicator(x\in \bigcup_{i\in \U_C} 
    (X_i\setminus X_i^\prime))
    \leq 2c_4 \tau \cdot \indicator(x\in \bigcup_{i\in \U_C} 
    (X_i\setminus X_i^\prime))$. Hence 
    \begin{align*}
        &\quad \sum_{i\in \U_C} \indicator(\eveG_C\cap\eveK_C)
        \cdot \max\Set{0, \med_1(X_i)-(1+2\varepsilon)\cdot \med_1(\pi(X_i))}\\
        &\leq 4c_4 \tau \cdot \sum_{x\in C}\indicator(\eveK_C)
        \cdot \Indicator{x\in \bigcup_{i\in \U_C} 
        (X_i\setminus X_i^\prime)} \cdot 
        \Big(\Norm{x-F_C(x)}+\Norm{\pi(x)-\pi(F_C(x))}\Big).
    \end{align*}

    By \cref{prop:random_proj}, for each $x\in C$, 
    $\Norm{\pi(x)-\pi(F_C(x))}$ is comparable with 
    $\Norm{x-F_C(x)}$. More concretely, 
    \begin{align*}
        &\quad \ee_{\pi}\left[
            \indicator(\eveK_C)
            \cdot \Indicator{x\in \bigcup_{i\in \U_C} 
            (X_i\setminus X_i^\prime)}\cdot 
            \Norm{\pi(x)-\pi(F_C(x))}
        \mid \decom 
        \right]\\
        &\leq 2\ee_{\pi}\left[
            \indicator(\eveK_C)
            \cdot \Indicator{x\in \bigcup_{i\in \U_C} 
            (X_i\setminus X_i^\prime)}\cdot 
            \Norm{x-F_C(x)}
        \mid \decom 
        \right]
        +e^{-\Omega(\tardim)} \Norm{x-F_C(x)}.
    \end{align*}
    Then \begin{align}
        &\quad \ee_{\pi}\left[\sum_{i\in \U_C} 
        \indicator(\eveG_C\cap\eveK_C)
        \cdot \max\Set{0, \med_1(X_i)-(1+2\varepsilon)\cdot \med_1(\pi(X_i))}
        \mid \decom\right]\notag\\
        &\leq 12c_4 \tau \sum_{x\in C} \Norm{x-F_C(x)}
        \cdot \Pr_{\pi}\left[\left.\begin{matrix}
            \eveK_C,\\
            x\in \bigcup_{i\in \U_C} 
            (X_i\setminus X_i^\prime)
        \end{matrix}
        \right|\decom\right]
        +4c_4 \tau \cdot e^{-\Omega(\tardim)} \sum_{x\in C}\Norm{x-F_C(x)}.
        \label{eq:ptas_further_3}
    \end{align}
    Each $X_i\setminus X_i^\prime$ is a random subset of $ C$. 
    For every $x\in C$, $\Pr[x\in X_i\setminus X_i^\prime \mid \decom]
    \leq \Pr[x\in X_i\setminus X_i^\circ \mid \decom] 
    + \Pr[x\in X_i^\circ\setminus X_i^\prime \mid \decom]$. Notice that 
    \begin{align*}
        \PR{x\in  X_i\setminus X_i^\circ \mid \decom}
        &\leq \PR{\exists f\in \cFapx_C, 
        \Norm{\pi(x)-\pi(f)}< \frac{1}{1+\varepsilon}
        \Norm{x-f}}\\
        &\leq e^{-\Omega(\varepsilon^2 \tardim)}
        \norm{\cFapx_C}
        \leq e^{-\Omega(\varepsilon^2 \tardim)}\cdot \alpha\tau,
    \end{align*}
    and $\Pr[x\in X_i^\circ\setminus X_i^\prime \mid \decom]\leq e^{-\Omega(\varepsilon^2 \tardim)}$.
    Thus $\Pr[x\in X_i\setminus X_i^\prime \mid \decom]\leq O(\tau) 
    \cdot e^{-\Omega(\varepsilon^2 \tardim)}$
    for every $i\in\U_C$. Conditioning on event $\eveK_C$, 
    $\norm{\U_C}\leq 2c_4 \tau$. Hence
    \begin{align*}
        \Pr_{\pi}\left[\left.\begin{matrix}
            \eveK_C,\\
            x\in \bigcup_{i\in \U_C} 
            (X_i\setminus X_i^\prime)
        \end{matrix}
        \right|\decom\right]
        \leq O(\tau^2)\cdot e^{-\Omega(\varepsilon^2 \tardim)}.
    \end{align*}
    Thus the right hand side of \eqref{eq:ptas_further_3} can 
    be further bounded by 
    \begin{align*}
        &\quad \ee_{\pi}\left[\sum_{i\in \U_C} 
        \indicator(\eveG_C\cap\eveK_C)
        \cdot \max\Set{0, \med_1(X_i)-(1+2\varepsilon)\cdot \med_1(\pi(X_i))}
        \mid \decom\right]\\
        &\leq O(\tau^3) \cdot e^{-\Omega(\varepsilon^2 \tardim)}
        \sum_{x\in C} \Norm{x-F_C(x)}
        = O(\tau^3) \cdot e^{-\Omega(\varepsilon^2 \tardim)}
        \cost(C, \cFapx_C)\\
        &\leq O(\tau^3) \cdot e^{-\Omega(\varepsilon^2 \tardim)}
        \opt(C)
        \leq O(\tau^4) \cdot e^{-\Omega(\varepsilon^2 \tardim)},
    \end{align*}
    where we use the fact that $\cFapx_C$ is 
    an $O(1)$-approximate solution on $C$ and that $\opt(C)\leq \tau$.
    Since $\tardim = \Omega(\varepsilon^{-2}(\log \tau + \log(1/\varepsilon)))$,
    above can be further bounded by $\varepsilon^2$. This completes
    the proof.
\end{proof} \end{appendices}
\end{document}